\documentclass[prb,twocolumn,showpacs,amsmath,amssymb,pra,superscriptaddress,longbibliography
]{revtex4-1}

\usepackage{amsthm}
\usepackage[pdftex]{graphicx}
\usepackage[usenames,dvipsnames]{color}
\usepackage{epstopdf}
\usepackage[hidelinks]{hyperref}
\usepackage{tabularx}

\usepackage{mathrsfs}

\def\t{{t}}
\def\lrange{{L}}
\def\Op{\hat{O}}
\def\m{M}

\newcommand{\x}{\mathbf{x}}
\newcommand{\Z}{F}
\newcommand{\force}{\Upsilon}
\newcommand{\Np}{N_p}

\def\dr{\mathrm d{\bf r}}

\def\MLcostfct{{\textsl{g}}}

\def\br{{\bf r}}
\def\Ne{N_e}
\def\Nb{N}

\def\dbr{\mathrm{d}\br}
\def\dbrp{\mathrm{d}\br'}

\def\s{_\mathrm{s}}

\def\prefact{{\mathcal{T}}}

\usepackage[encapsulated]{CJK}

\newtheorem{theorem}{Theorem}[section]

\newtheorem{lemma}[theorem]{Lemma}

    \makeatletter
\def\@fnsymbol#1{\ensuremath{\ifcase#1\or \ddagger\or \dagger\or
   \mathsection\or \mathparagraph\or \|\or **\or \dagger\dagger
   \or \ddagger\ddagger \else\@ctrerr\fi}}
    \makeatother

\graphicspath{{FIGs/}}

\begin{document}

\title{Density functionals and Kohn-Sham potentials\\
with minimal wavefunction preparations on a quantum computer}
\author{Thomas E.~Baker}
\affiliation{Institut quantique \& D\'epartement de physique, Universit\'e de Sherbrooke, Sherbrooke, Qu\'ebec J1K 2R1 Canada}
\author{David Poulin}
\thanks{{\it N\'e le 1$^{er}$ d\'ecembre 1976, d\'ec\'ed\'e le 25 juin 2020.}}
\affiliation{Institut quantique \& D\'epartement de physique, Universit\'e de Sherbrooke, Sherbrooke, Qu\'ebec J1K 2R1 Canada}
\affiliation{Quantum Architecture and Computation Group, Microsoft Research, Redmond, WA 98052 USA}
\affiliation{Canadian Institute for Advanced Research, Toronto, Ontario, M5G 1Z8 Canada}
\date{\today}

\begin{abstract}
One of the potential applications of a quantum computer is solving quantum chemical systems. It is known that one of the fastest ways to obtain somewhat accurate solutions classically is to use approximations of density functional theory. We demonstrate a general method for obtaining the exact functional as a machine learned model from a sufficiently powerful quantum computer.  Only existing assumptions for the current feasibility of solutions on the quantum computer are used. Several known algorithms including quantum phase estimation, quantum amplitude estimation, and quantum gradient methods are used to train a machine learned model. One advantage of this combination of algorithms is that the quantum wavefunction does not need to be completely re-prepared at each step, lowering a sizable prefactor.  Using the assumptions for solutions of the ground state algorithms on a quantum computer, we demonstrate that finding the Kohn-Sham potential is not necessarily more difficult than the ground state density. Once constructed, a classical user can use the resulting machine learned functional to solve for the ground state of a system self-consistently, provided the machine learned approximation is accurate enough for the input system.  It is also demonstrated how the classical user can access commonly used time- and temperature-dependent approximations from the ground state model. Minor modifications to the algorithm can learn other types of functional theories including exact time- and temperature-dependence.  Several other algorithms--including quantum machine learning--are demonstrated to be impractical in the general case for this problem.
\end{abstract}

\maketitle
\tableofcontents

\section{Introduction}

Quantum computing has been proposed as an alternative to classical computing,\cite{nielsen2010quantum} and there are some problems which can be solved faster than known classical algorithms.\cite{deutsch1992rapid,shor1994proceedings,grover1997quantum,brassard2002quantum,shor1999polynomial,grover2001schrodinger}  One of the most sought after and potentially far reaching applications on a quantum computer is the solution of quantum chemistry problems.\cite{aspuru2005simulated,brown2010using,lanyon2010towards,whitfield2011simulation,cao2019quantum,mcardle2020quantum}  Obtaining exact solutions from a quantum computer efficiently could revolutionize modern applications including the creation of new medicines, fertilizers, batteries, superconductors, and more.\cite{kang2006electrodes,mavros2014can, cudazzo2008ab,rod2000ammonia,norskov2006density,flores2016high,rydberg2014use}

To do this, one important quantity to determine is the ground state energy. The energy is a highly useful quantity for determining properties such as the equilibrium geometry of a molecule. Yet, the energy is not descriptive enough to fully characterize all desired properties of a system.   For example, the band structure can be a useful tool to characterize a material, but this requires measurements at several $k$-points.\cite{ashcroft1976introduction} So, many measurements of the wavefunction would be required for some simple quantities.

Measurement on the quantum computer is expensive because the wavefunction must often be re-prepared before a second measurement is performed. It has already been shown on a quantum computer that obtaining the wavefunction can be extremely costly, \cite{jansen2007bounds,wecker2014gate,poulin2014trotter,lemieux2020resource} taking months or years even for moderately sized systems,\cite{poulin2014trotter,lemieux2020resource} and that the wavefunction cannot be copied.\cite{park1970concept}  The wavefunction is therefore a valuable commodity and measurements should be minimized.\cite{aaronson2018shadow}  

One option to encode many solutions into one measurement is to use a machine learned (ML) model.\cite{carleo2019machine}   In general, ML models can interpolate remarkably well between input data to give access to many systems, including those not already solved.   In principle, the ML model can be constructed directly on the quantum computer or from classical data generated by the quantum computer.

Using ML models would also allow for users of the quantum computer to export solutions to classical users. The results could then be quickly retrieved classically from the model and are generally accurate over a training manifold on which the model was constructed in a desire to generate the best machine learned model possible. 

Finding the full wavefunction would require exponentially many measurements, so this can be difficult to implement on a quantum computer. But the same information can be expressed in a more compact form.  So, we can look at alternative formulations of quantum physics for the most descriptive model.  

The route pursued here is with density functional theory (DFT).\cite{hohenberg1964inhomogeneous} Hohenberg and Kohn established that the one-body density, $n(\br)$, is one-to-one with the external potential $v(\br)$ up to a constant shift of the potential.  In essence, the density can replace the wavefunction, but it has fewer variables.

In order to use DFT, we must find some other means of obtaining the energy, since the Hamiltonian is not used in DFT.  Instead, the universal functional, $F[n]$, must be found.  It was proven that the universal functional exists and is common to all problems of the same electron-electron interaction.\cite{hohenberg1964inhomogeneous,kohn1965self}

The quantities required for the classical user to find self-consistent solutions are the exact functional (determining the energy) and the functional derivative.\cite{SRHM12}  So, in addition to finding $F[n]$, we also must find some other quantity such as the density, $n(\br)$, or the Kohn-Sham (KS) potential, $v\s(\br)$.\cite{kohn1965self} With these components, we can fully characterize a quantum ground state and solve for other measurable quantities.

It has already been established that the density functional can be successfully modeled with ML methods on the classical computer.\cite{behler2007generalized,SRHM12,LBWB16,brockherde2017bypassing,bogojeski2019density,bogojeski2018efficient,nagai2018neural,li2016understanding,grisafi2018transferable,fabrizio2019electron,denner2020active,nagai2020completing,manzhos2020machine,suzuki2020machine,wetherell2020insights} Exact quantities at several different external potentials must be found for the ML models to be trained. The number of training points needed to construct accurate models are not prohibitively large.    From the ML functionals, self-consistent solutions can be obtained.\cite{SRHB13,SRMB15,li2016understanding,vu2015understanding} Numerically accurate ML functionals satisfy all exact conditions of $F[n]$\cite{LBWB16,hollingsworth2018can} and escape the common errors of approximated density functionals.\cite{engel2011density,gross2013density}

To apply the classical ML-DFT methods on a quantum computer, some additional constraints must be minded. Previous attempts to obtain functionals from the quantum computer have relied on many measurements of the wavefunction for each system of interest.\cite{hatcher2019method,whitfield2014computational,brown2019solver,yang2019measurement,rall2020quantum,mcardle2020quantum}  In our view, a worthy goal is to avoid both excessive measurements and re-preparations of the wavefunction especially in the case of time-dependent quantities.\cite{whitfield2014computational,ullrich2011time}  

This work proposes a feasible algorithm that finds the ML model for $F[n]$ on the quantum computer if a ground state wavefunction is available.  The algorithm leaves the wavefunction largely undisturbed so it can be used as the starting state for another system, greatly reducing the prefactor required to solve other systems. This is accomplished by using a state-preserving quantum counting algorithm to extract descriptive quantities such as the density.\cite{temme2011quantum,bakerGreen20} Much of the algorithm is kept entirely on the quantum computer to motivate future improvements for speed, but the counting algorithm does allow for information to be output classically. 

This algorithm is an alternative to running one very long computation and just measuring one energy, in that each step of the wavefunction preparation is proposed to solve another system.  Thus, no step from the ground state solver is wasted when using the algorithm here.

We also demonstrate that the Kohn-Sham potential can be solved using a similar strategy as the wavefunction.  A gradient evaluated on a cost function for the KS system allows for the determination of the exact KS potential. This strategy can be faster than obtaining the density. Further, we demonstrate how access to the functional can be used to find approximate time- and temperature-dependent behavior in a system from the ground state functional, and modifications can be added to obtain exact results.

One temptation would be to use quantum machine learning, but long-known bounds on the efficiency of these methods preclude their use here.\cite{servedio2004equivalences}  This agrees with recent demonstrations that some known quantum machine learning algorithms are not universally advantageous.\cite{tang2019quantum,tang2018quantum,gilyen2018quantum,chia2020sampling}   We also discuss general limitations on known algorithms such as quantum machine learning and why these algorithms are expected to be inefficient here.

Section~\ref{alg} presents the algorithm Section~\ref{discuss} will discuss several uses of the resulting functional and considerations in choosing quantities to solve for.  Section~\ref{limitations} will discuss known limitations and justify why the algorithm is constructed as presented.  Necessary background information on quantum chemistry, DFT, ML, and quantum computing algorithms is given in the appendixes.

\section{Algorithm for the functional}\label{alg}

In order to establish an algorithm for the quantum computer that gives results in quantum chemistry, knowledge of both fields must be understood.  To avoid a lengthy summary in the main text, we have included relevant background knowledge in the appendixes in case they are needed.  Nearly all of the computational steps ({\it e.g.}, machine learning the functional) have already been demonstrated by us in the references and the algorithms performed accurately. This section will contain all the elements of the algorithm and assumes only a background of algorithms in quantum computing.

We provide Fig.~\ref{compositealg} to illustrate the steps necessary for one iteration of the algorithm, which we will refer to as a recycled wavefunction for minimal prefactor (RWMP) method. Although many quantities could be produced from this algorithm, we will focus on the components useful for the density functional.  The inputs are the external potential for some system ($|v(\br)\rangle$) and initial guess weights for the ML model ($|w^{(i)}\rangle$) represented as classical variables throughout. The following steps are required to obtain the solution for a given system and then update the parameters of the ML model.

\begin{figure}[b]
\includegraphics[width=\columnwidth]{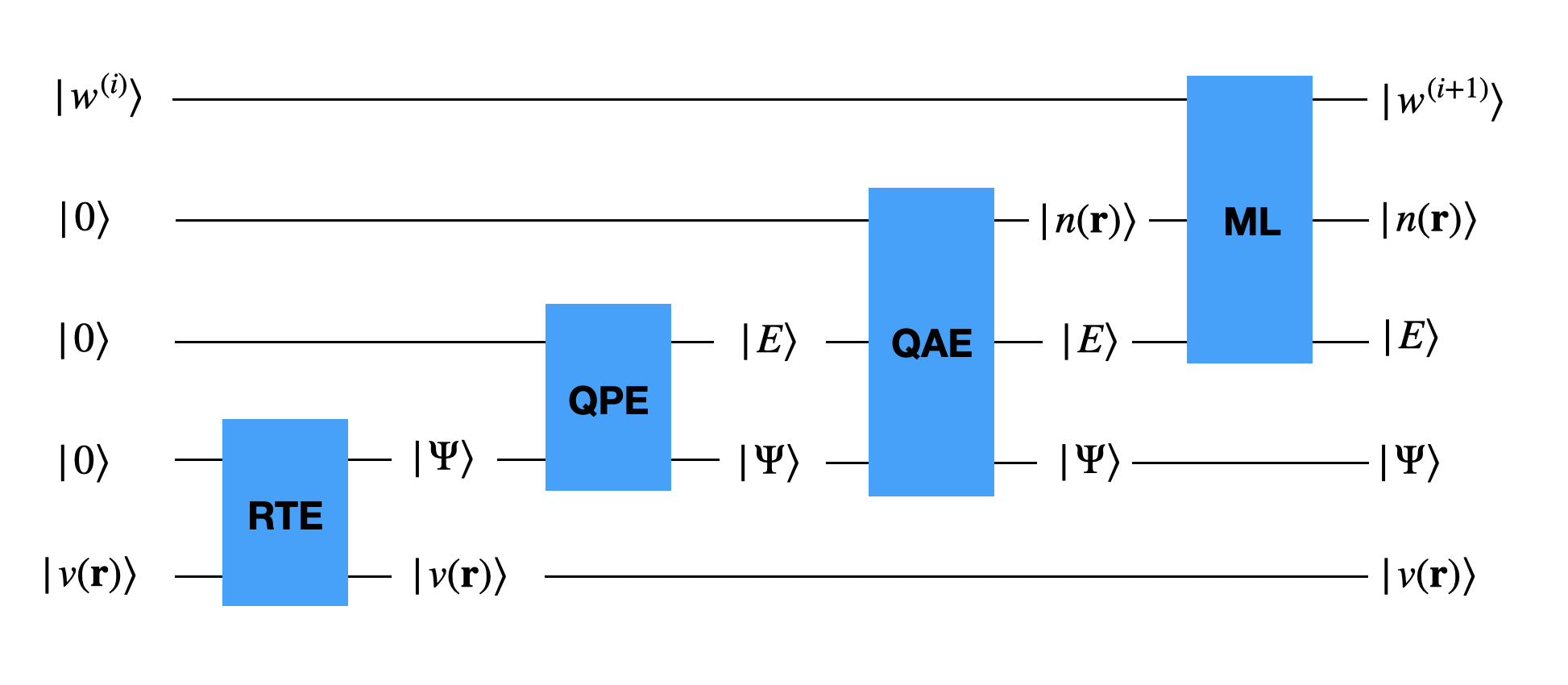}
\caption{One step of the RWMP algorithm for the density.  The next iteration uses the output wavefunction as the starting state for the next system to reduce the prefactor.  A similar procedure can be used to find the KS potential, replacing $n(\br)$ with $v\s(\br)$ (with a QGA as an oracle query for the QAE).  The ML model may need inputs from other registers.
\label{compositealg}}
\end{figure}

\begin{enumerate}
\item We prepare a ground state wavefunction $|\Psi\rangle$ for a given external potential, $|v(\br)\rangle$, and number of electrons, $N_e$.
\end{enumerate}

This can be done by real-time evolution (RTE, see Apx~\ref{rte}).  In Fig.~\ref{compositealg}, this is denoted by a box for RTE.  The subroutine that obtains the wavefunction here does not have to be RTE.  If another, more advanced solver is developed and used, then this can be substituted with no change to the rest of the RWMP algorithm

Note that the methods to obtain quantities from the result of classical computations are not available since the quantum wavefunction has coefficients that are stored in superposition (a linear combination of $|0\rangle$ and $|1\rangle$). This means that the coefficients of the wavefunction cannot be found except with many measurements.  We must first find the energy before obtaining other relevant quantities.

\begin{enumerate}
\item[2.] We obtain the ground state energy $|E\rangle$ from the ground state wavefunction $|\Psi\rangle$ given in the previous step.
\end{enumerate}

Access to the ground state energy is provided by quantum phase estimation (QPE, Appendix.~\ref{QPE}) or with some other method like qubitization.\cite{low2019hamiltonian} On Fig.~\ref{compositealg}, this step is denoted by QPE.

The next task is to determine some quantities of interest without requiring a full measurement of the wavefunction.  The counting algorithm used allows for the wavefunction to be used again on the next iteration.

\begin{enumerate}
\item[3.] Given the energy $|E\rangle$ and wavefunction $|\Psi\rangle$, we generate some quantity that is denoted here as $|n(\br)\rangle$ from a quantum amplitude estimation (QAE, Appendix.~\ref{qcount}), a name which we use interchangeably with quantum counting throughout.
\end{enumerate}

The symbol used, $n(\br)$, is for the density (Appendix.~\ref{DFT}), but we can substitute this quantity for others. For example, one could also determine the KS potential $|v\s(\br)\rangle$.  This step is denoted as QAE in Fig.~\ref{compositealg}.  Note that this step may involve an oracle query such as a quantum gradient algorithm (QGA, Appendix.~\ref{QGA}) or have other subroutines.  This is a point we will expand on in the next section, Sec.~\ref{qinterest}, when discussing how to obtain either $n(\br)$ or $v\s(\br)$. For now, we focus on what to do once the quantity is obtained.

The output wavefunction is slightly modified by the QAE but remains nearly the same state with a small amount of error.  The procedure to return the wavefunction to its original state does not completely re-prepare the wavefunction and instead has an iterative set of steps to repair the wavefunction as explained in Ref.~\onlinecite{temme2011quantum} (see also Appendix.~\ref{qcount}).

\begin{enumerate}
\item[4.] The output of the QAE can be used to update the ML model's parameters $|w^{(i)}\rangle$ to the next iteration, $|w^{(i+1)}\rangle$.  More than one ML model can be trained here ({\it e.g.}, a ML model for $|E\rangle$ and for $|n(\br)\rangle$).
\end{enumerate}

The typical steps in a stochastic gradient descent (SGD, Appendix.~\ref{SGD}) of forward and backward propagation can be used.  For backward propagation, the output of a QGA can be used to update the ML parameters.  This step is marked as ML in Fig.~\ref{compositealg}. This ML operation may need to be controlled on $|v(\br)\rangle$ depending on which quantity is being trained. Note that the QAE output could also be stored clasically and then machine learned not on the quantum computer, skipping this step but eliminating the opportunity for an improvement in the ML box from quantum advantages.

\begin{enumerate}
\item[5.] Another external potential is provided ($|v+\Delta v\rangle$ (not shown in Fig.~\ref{compositealg}) and the wavefunction is re-used as the starting state for the next RTE.
\end{enumerate}

Here, $\Delta v$ is chosen by the user and simply added to the coefficients of $v(\br)$ to give the next potential. Many passes through a set of potentials must occur to obtain an accurate ML model.  

Recall that the resource estimate given in Ref.~\onlinecite{poulin2014trotter}, several months may be required for a small molecule with RTE.  The RWMP algorithm allows for a sequence of intermediate systems to be visited, effectively making use of that time to obtain data.  So, if the system starts in a configuration where the initial wavefunction is accurate, then this can become the first data point for the ML model.  Each subsequent time step could be another potential data point for the model.

The second advantage to this strategy is that the sequence of potentials in the RWMP algorithm can allow for the ordering of the next potential to be close by. This allows for the best starting state for the next system to be used and reduce the total amount of RTE steps that must be run over all systems.  In summary, the RWMP algorithm here could make use of the preparation time for a hard to solve system by finding data from the intermediate steps and reduce the prefactor in the solution.

The RWMP algorithm repeats until all systems are visited.  The final step is to measure the parameters of the ML model, $|w\rangle$.  The model can then be used classically.

In what follows, we discuss which quantities should be obtained by the QAE for the best description of the ground state. Then in Sec.~\ref{discuss}, we discuss aspects of this strategy that were not absolutely necessary in defining the algorithm and how a classical user can use the result. Finally, we expand on many points that were not crucial for the basic understanding of the RWMP algorithm and explain limits that ultimately lead to this algorithm (and why other subroutines were not used) in Sec.~\ref{limitations}.

\subsection{Quantities of interest}\label{qinterest}

We move on to discussing which quantities are best to determine from the wavefunction via the QAE.  There are two main options: the coefficients of the density matrix and the KS potential.  When choosing the quantities of interest, both the functional and the functional derivative must be determined with the ML model in order to find solutions on the classical computer.  There are several types of functionals that can be trained for this, some of which are presented here.  

\subsubsection{Density functionals}\label{densityfunctionals}

In the RWMP algorithm, one option is to find the density, $n(\br)$, since this is proven to be a suitable replacement for the wavefunction from the Hohenberg-Kohn theorem in Ref.~\onlinecite{hohenberg1964inhomogeneous}.  The $N^2$ elements of the density matrix are expectation values of the operator $\hat c^\dagger_i\hat c_j$ (not just diagonal elements; see Apx~\ref{manybody}) where $\hat c$ is a fermionic operator defined in Appendix.~\ref{manybody}.  Spin indices are ignored for simplicity for now.  Since the expectation value is on the interval $[-1,1]$, we can use the operator $(\hat c^\dagger_i\hat c_j+1)/2$ (defined on [0,1]) and afterward shift the result back to the original interval.  Using this shifted operator is a necessary component to using the QAE because the expectation value can now be related to a probability.  The number of rounds required for the QAE relates to the inverse of the probability of failure requested.\cite{temme2011quantum}

There are several options for the ML model. The ML step can train directly from $n(\br)$ to $E$ or we can take as an input $v(\br)$ to the model and train both $n[v](\br)$ and $E[v]$.  The first option gives a pure density functional (Appendix.~\ref{DFT}).  The second option gives a potential functional, which is a dual functional to DFT (Appendix.~\ref{PFT}).\cite{yang2004potential}  Both of these theories can be solved self-consistently.  One can also train the bifunctional $E[n,v]$.\cite{brockherde2017bypassing}

\subsubsection{Kohn-Sham potentials}

The other main quantity, which has more value, is the KS potential, $v\s(\br)$.  The defining feature of the KS system is a noninteracting system that has the same density as a given interacting system (see more discussion and ways to realize this potential in Appendix.~\ref{kohnsham}).  The potential defining this noninteracting system is defined as $v\s(\br)$ which is described by $N$ parameters.  

Two potentials can have the same density if the systems do not have the same electron-electron interactions, so as not to violate the Hohenberg-Kohn theorem.\cite{hohenberg1964inhomogeneous}  The KS potential is highly valuable since it can be applied in several other instances. This includes finding time- and temperature-dependent calculations from $v\s(\br)$ or the KS band structure.\cite{ullrich2011time,perdew1985kohn,perdew1997comment}

A central question is whether $v\s(\br)$ exists for a given interacting system. This is known as the problem of $v$-representability.  Since it is proven that the KS potential always exists on a lattice, $v\s(\br)$ always exists here.\cite{levy1979universal,kohn1983v,chayes1985density,wagnerPRB14}

The KS potential must be converged to, just as we had to evolve an initial state in RTE to a final state. The difference in finding the KS potential is that a gradient is applied instead of a time evolution operator for the wavefunction. 

The KS potential, $v\s$, will satisfy the minimization,\cite{gidopoulos2011progress,callow2018optimal,callow2020density}
\begin{equation}\label{variationalKS}
\underset{v\s}{\mathrm{min}}\left(\langle\Psi[v]|\hat T+\hat{V}\s|\Psi[v]\rangle-\langle\Phi[v\s]|\hat T+\hat{V}\s|\Phi[v\s]\rangle\right)
\end{equation}
where $\Psi$ is the interacting wavefunction and $\Phi$ is the KS wavefunction (see Appendix.~\ref{variationalKSdetails} for more information).   There are other methods to obtain the KS potential,\cite{jensen2016numerical,jensen2018numerical,kanungo2019exact,kumar2020general} but the method used here is straightforward on a quantum computer given a close enough starting guess or small enough molecule ({\it i.e.}, those solvable by RTE).  Many other methods to find the Kohn-Sham potential would either require numerous measurements of the wavefunction or large overhead in terms of qubits for operations that are simple on the classical computer but costly on the quantum computer including addition, division, etc.\cite{draper2000addition}

Equation~\ref{variationalKS} is used as the output of the oracle query in QAE for the QGA.\cite{jordan2005fast,gilyen2019optimizing} Note that the QGA is particularly useful for finding the functional derivatives here, notably taking the variation of all possible $v\s(\br)$ in one oracle query.\cite{jordan2005fast} The resulting gradient is applied on the coefficients of the KS potential and the process is repeated sufficient times until the true KS potential is obtained.  An initial guess for the parameters could be taken from existing semi-local approximations such as local density approximations, etc.~\cite{medvedev2017density,baker2017methods,*perdew2017IPAM} or using a classical method.\cite{gidopoulos2011progress,wagnerPRB14,jensen2016numerical,jensen2018numerical,callow2018optimal,kanungo2019exact,callow2020density,kumar2020general,kumar2020using} The other classical methods where a gradient is used to evolve the potential may be useful, but the density does not need to be constructed to use Eq.~\eqref{variationalKS}.

In order to construct Eq.~\eqref{variationalKS} on a quantum computer, the eigenvalues of free Hamiltonians, such as the KS Hamiltonian ($\hat T\s+\hat V\s$) can be mapped to the interval [0,1] for the QAE.  The operator must also be scaled by a constant, but we also note that shifting the potential by a constant, $v\s(\br)\rightarrow v\s(\br)+\mathcal{C}$, is allowed to ensure all eigenvalues are positive without changing the eigenvectors.\cite{hohenberg1964inhomogeneous} Further, identifying some upper bound on the expectation value, $\varepsilon_\mathrm{max}$, the scaled operator could appear as $\langle\Psi|\hat T\s+\hat V\s + \mathcal{C}|\Psi\rangle/\varepsilon_\mathrm{max}$.

The expectation value $\langle\Phi|\hat T\s+\hat V\s|\Phi\rangle$ can be evaluated in one of two ways.  First, it can be computed by diagonalizing $\hat T\s+\hat V\s$ (determining $h_k$) and taking the sum $\langle\tilde\Phi|\sum_kh_k\hat c^\dagger_k\hat c_k|\tilde\Phi\rangle$ in the diagonalized basis with $\tilde\Phi$ is now used.  This procedure uses QAE to find the expectation value analogously to finding the coefficients of the density.   Alternatively, one could apply a QPE (or qubitization) directly for the KS system, noting that the gate count is drastically less for the noninteracting system.

The KS potential as encoded in the ML model can be expressed as either $v\s[v](\br)$ or as $v\s[n](\br)$, where both are well-defined.\cite{gross2009adiabatic}

\subsection{Example for the Kohn-Sham potential}

To illustrate further some of the more abstract quantities in the RWMP algorithm in Sec.~\ref{alg}, we provide an expanded example of the RWMP algorithm of how to obtain the ML-KS potential.
\begin{enumerate}
\item An initial potential $v(\br)$ is chosen.

This can be done by assuming the external potential is a set of nuclei with a Coulombic interaction, $v(\br)=\sum_a(-Z/|\br-\br_a|)$, where $a$ indexes the positions of the nuclei, $\br_a$ with atomic number $Z$.

\item A basis set is chosen, $\varphi_k(\br)$.

The model can then be discretized in this basis and the resulting model has fermionic operators (see Appendix.~\ref{manybody}).

\item RTE is run and the ground state is obtained.
\item QPE is used to find the ground state energy, $E$.
\item $\Psi$ and $E$ are used in the QAE to find the first term in Eq.~\eqref{variationalKS}. The energy of the noninteracting system is also obtained for the second term in Eq.~\eqref{variationalKS} using some initial guess for $v_s(\br)$. This constructs the oracle in the QGA.

\item The QGA result is added to the coefficients of the KS potential and the last step is repeated until the $v\s(\br)$ a certain number of times to find the minimum.

\item The values $E$ and $v\s(\br)$ are input into a ML model.  The gradient of the parameters in the model are updated with a QGA by repeatedly computing the gradient and adding it to the current coefficients.

We can also use the QGA to output $E$ and $v\s(\br)$ to the classical user and learn the final set of potentials classically.

\item A new set of atomic coordinates are provided, $\br_a'$, the difference $\Delta v$ between the previous and this potential is computed, and the result is added to the old potential.

We now have the next potential and can start again at step 3 until all systems are visited.
\end{enumerate}

  Note that this computation can be restarted at any time (at the significant cost of re-preparing the wavefunction).  Note that it has been assumed that all systems are run for the same basis set, although this condition could be relaxed in principle.

One advantage of using this method is that the Kohn-Sham potential is characterized by $N$ coefficients $\kappa_k$, $v\s(\br)=\sum_{k=1}^N\kappa_k\varphi_k(\br)$. This is a factor $N$ less than the density required.  The evolution of the Kohn-Sham potential from an initial guess potential by gradients is very similar to how the wavefunction is evolved with RTE.

\section{Additional considerations and use of the functional}\label{discuss}

In the previous section, an RWMP algorithm for finding the density functional and Kohn-Sham potential were detailed.  It was also noted that other variations of the density functional could be found using the same technology.

Several points that could extend the RWMP algorithm--and other details about what was introduced--are discussed here.  These include resource costs, a comparison between finding the density and the KS potential, the universality of the functional, how the functional can be used, applications to other types of functional theories, and opportunities for near-term studies.  Some points relating to a quantum advantage that are discussed here are continued in Sec.~\ref{quantumspeedup?}.

\subsection{Scaling and resources required}\label{scaling}

The algorithm will require a fully scalable quantum computer, probably with self-correcting memory.  Near-term examples are available for density functionals, and a full discussion of the known features of finding the density functional with this method is given here.

\subsubsection{Algorithmic scaling}

The scaling of the RWMP algorithm for the density functional in terms of the number of basis sets is $\mathcal{O}(N^2)$ for one system asymptotically in the number of basis functions due to time evolution. The actual complexity in practice is between $\mathcal{O}(N^2)$ and $\mathcal{O}(N^4)$ for intermediate system size (see Appendix.~\ref{basis_sets} and \ref{rte}).  The other subroutines scale at most as $\mathcal{O}(N^2)$ (see the appendixes).  

\subsubsection{Prefactor for convergence of the wavefunction}

Even though the scaling in terms of the number of basis functions is polynomial, the actual cost is significant due to a prefactor.\cite{poulin2014trotter}

The prefactor of the algorithmic scaling is problem dependent and will be the dominant cost to obtaining a ground state wavefunction.  The prefactor will depend on the time used to prepare $\Psi$, the number of steps that the QAE algorithm must be run, and the  required number of systems to be visited.  Note that the number of steps that must be run to obtain the correct time evolution is dependent on how close the starting state is to the system's solution, the number of electrons, how strongly correlated the electrons are, and many other variables.

Note that if the first system in the RWMP algorithm is exactly equal to the initial wavefunction ({\it e.g.}, the well-separated limit for a neutral molecule where each separated piece is a hydrogen atom with one electron and where Hartree-Fock is exact for one-electron systems), then the time to make the first solution is zero and each subsequent motion of the atoms closer together will be another data point that may not require too many time steps to obtain.

\subsubsection{Convergence of the machine learned model}

In order to aid the convergence of the ML model, it would be useful to have several quantum computers running at once.  This would mean that more than one system can be used to construct a mini-batch update for the cost function of the gradient descent, making the resulting update to the ML model more accurate.  One also can save the output quantities as future training points at the cost of extra registers.  Once an accurate ML model is generated, the problem does not need to be solved again.

\subsubsection{Resources required}

The RWMP algorithm presented in this work requires 4+$\zeta$ qubit registers for $\zeta$ quantities of interest that are not the energy ({\it e.g.}, learning $E$ and $v\s(\br)$ implies $\zeta=1$).  There will also be overhead for the QAE and other parts of the RWMP algorithm One can also divide the parameters in the ML model into separate registers for each additional quantity of interest since the model spaces between them should be sparse.

The number of qubits required is clearly large and dependent on the specific steps of the RWMP algorithm used.  This means the RWMP algorithm is best suited for a fully working quantum computer that is expected in the future.  There are many ways to reduce the steps necessary. For example, in one way, by outputting the QAE to a classical user.  However, we want to maintain the flexibility that a quantum advantage could be realized here as discussed next in Sec.~\ref{qNN}.

\subsection{Structure of the neural network}\label{qNN}

In a neural network, the form of the nonlinear function, $\mathscr{S}$ (see Appendix.~\ref{SGD}), is often chosen so that it is easily differentiable.  On the quantum computer, determining the gradients are accomplished in one oracle query (see Appendix~\ref{QGA}), so the traditionally used functions on a classical computer ({\it e.g.}, sigmoids, etc.) can be swapped out for another function that may be lead to better performance or faster convergence.  It is not clear if this will produce any detectable advantage.

The number of coefficients required to obtain an accurate ML model can be very high, but some guidelines can reduce this number for a given quantity.\cite{langer2020representations} For the problem at hand, note that the training is done in the same basis as the problem is expressed.  In this case, the connectivity of the neural network can be constrained on physical arguments.  Known bounds on the structure of local correlations as proved by Hastings in Ref.~\onlinecite{hastings2004locality} apply to densities that was originally shown by Kohn, {\it et.~al.} in Ref.~\onlinecite{kohn1996density,prodan2005nearsightedness}  This means that perturbations on the density decay exponentially with distance for gapped systems and as a power law for gapless systems.\cite{baker2019m}  Then, the neural network models for the density do not need to account for arbitrarily long ranged connections and can remain local.  In essence, the connectivity of edges on the graph for the neural network would be similar to the structure of a multi-scale entanglement renormalization ansatz\cite{vidal2008class} (scaling as $N\log N$) but in three-dimensions. Note that this estimate is a minimum and it may still be advantageous to add more hidden layers or connectivity.  Note that we will expect that gradient-free methods (Appendix.~\ref{gradfree}) will not necessarily be more useful in training the ML model.

Not all quantities will have a local structure.  The $v\s(\br)$ is fully nonlocal due to the Hartree potential (see Appendix.~\ref{findKS}) and can require all-to-all connectivity between subsequent hidden layers in the neural network.

Note that it is not required to train the ML model on the quantum computer.  One can determine the quantities of interest for each system and output the results with the QAE from the quantum computer to the classical computer.\cite{temme2011quantum} But we give the option here to train the neural network on the quantum computer in case a quantum advantage can be realized.   Using the RWMP algorithm as presented in Sec.~\ref{alg} would still avoid excessive measurement; meanwhile, outputting to classical variables will reduce the amount of overhead needed to store ML parameters.

\subsection{Comment on the universal vs.~exact functional}\label{universality}

ML functionals are very accurate over the training manifold on which they were constructed.\cite{SRMB15}  We have chosen the appropriate adjective (exact or universal) describing the functional very carefully in each use here. The procedure we describe here is not truly universal since the accuracy of the ML functional is limited to the training manifold of potentials that we have explored for the ML model.  For example, trying to solve a model trained for solutions with $N_e$ electrons with a new potential in the training manifold with $N_e+1$ electrons will try to project the solution back onto the $N_e$ solutions.

The resulting ML functional can be numerically exact, however, for a given problem.  This functional will be called the exact functional, implying that it is accurate for some systems but not all possible systems.

To complete the discussion, one can have a universal but not exact functional.  For example, if one estimated the energy to be a fixed value for all systems, this would be universal but not very useful.

\subsection{Comparing the search for the density and Kohn-Sham potential}

To actually compare the true cost of the $\mathcal{O}(N^2)$ operations to find the density with the $\mathcal{O}(\prefact N)$ coefficients for the KS potential, one would need to know the number of times $\prefact$ a gradient must be applied. This depends on the system studied and starting KS potential.  It is therefore not clear as a general guideline if directly finding the coefficients of the density matrix is always better than starting from a KS potential that is close and applying the gradients. The key difference in the two cases is that the KS potential relies on a suitable starting state and the density does not. The assumption for the KS potential--that a good starting state is required--is very similar to the restriction on RTE itself to find the ground state, so all assumptions are consistent.

The fact that the KS potential is competitive with finding the density here is more a comment on the overhead required for the density (which is far greater than on the classical computer) than the KS potential being any easier to find.  There is one advantage to using the minimization for $v\s(\br)$ here in that the QGA is much more efficient than the classical computer. 

In addition to comparing the scaling, the nature of the KS problem requires a large basis set to obtain the proper KS potential and avoid any Gibbs oscillations that would appear in a truncated basis set.\cite{kanungo2019exact}  So, the $N$ required for an accurate density or energy might be smaller than for an accurate KS potential in practice.

Regardless, any effort to find $v\s(\br)$ is worth it since $v\s(\br)$ is far more descriptive and gives access to useful quantities.

\subsection{Use of the functional}

The algorithmic cost in Sec.~\ref{scaling} is paid only once.  When the model is given to a classical user, a separate cost must be paid to solve it for the ground state.  Given the ML model, the classical user can solve the ML functional self-consistently by determining the Euler-Lagrange equation to minimize the functional (see Appendix.~\ref{minKS}).  This requires that a projection back into the training manifold must occur to ensure the functional derivative is accurate.\cite{SRHB13,brockherde2017bypassing}  This projection is estimated to scale as $\mathcal{O}(N^2)$ but can be machine learned separately to speed up computation.\cite{SRHB13,LBWB16,brockherde2017bypassing,brockherde2017bypassing}  

Finding the pure density functional has better scaling for the classical user, scaling like the number of basis functions, $N$. Note that even though we can learn $F[n]$, the external potentials (as well as particle numbers and polarizations) over which the ML model is accurate must also be given to the classical user so it is understood what the training manifold is.

The KS system scales as the cube of the number of basis functions formally, $N^3$, since it is a noninteracting problem.  Note that when solving the KS system with the exact functional, convergence is proven.\cite{wagner2013guaranteed,wagnerPRB14}

If a bifunctional, $E[n,v]$, is used then a self-consistent solution is not necessary.\cite{brockherde2017bypassing}

Once the model is trained sufficiently accurately, one can always re-train for another type of functional ({\it e.g.}, from the KS potential, the density can be obtained and a pure density functional can be trained).

Finding the density functional--and using it to compute quantities--is preferable to generating a library of system properties and machine learning those properties.  From the DFT model, the other quantities can be constructed.

As an example application, if trained over enough external potentials, this method could efficiently evaluate molecular dynamics problems giving accurate results on laptops instead of supercomputers.\cite{li2015molecular}

\subsection{Other types of functional theories}

There are many other types of functional theories that are related to DFT that can be obtained with these methods.  The most easily extended method is density matrix functional theory (DMFT).\cite{gilbert1975hohenberg}  Since our method of obtaining the density was to actually obtain the coefficients of the density matrix (see Apx~\ref{qchem}), the density matrix functional is learned.

Extensions of DFT can also be solved with this method including the motion of nuclei,\cite{gidopoulos2014electronic} time-dependence (TD-DFT),\cite{runge1984density,van1998causality,elliott20093,ullrich2011time} thermal properties,\cite{mermin1965thermal} superconducting functionals,\cite{kohn1989orbital,capelle1997density} quantum electrodynamics-DFT,\cite{ruggenthaler2011time,tokatly2013time,ruggenthaler2014quantum} ensembles,\cite{filatov2015ensemble,jouzdani2019method} and others.\cite{cohen2006hardness,cohen2007foundations,elliott2010partition,yang2004potential,cangi2011electronic,cangi2013potential}  In each case, the solution on the quantum computer is modified in some way.  For example, a superconducting functional theory can be found if the pairing potential\cite{schrieffer1963theory} is also computed and learned.

Two of the methods in the above list (time- and temperature-dependent methods) deserve extended discussion in the subsequent sub-sections since both exact and approximate functionals can be found from both.  Note that the approximate methods only require the ground state functional in both cases.

\subsubsection{Time evolution}

Given an accurate enough ML approximation for the KS potential, one can time evolve the system.\cite{suzuki2020machine}  If the system is adiabatically time evolved with a weak enough perturbation, this would be available immediately from just the KS system at zero time (see Eq.~\eqref{chis} in Appendix.~\ref{tddft}).  The adiabatic approximation is often sufficient for many physical processes.

For time evolution beyond the adiabatic limit, an exchange-correlation kernel denoted as $f_\mathrm{xc}$ (see Appendix.~\ref{tddft}) would need to be found. Full evaluation of $f_\mathrm{xc}$ could be accomplished with the QGA (Appendix.~\ref{funcderivQGA}); however, one can estimate the kernel via back propagation in the neural network (see Appendix.~\ref{SGD}).   An accurate functional derivative will be required in this case.\cite{SRHB13}  One can also time evolve on the quantum computer and use expectation values at various times for other quantities.\cite{chowdhury2017quantum,low2017optimal,low2019hamiltonian,rall2020quantum}

\subsubsection{Finite temperature calculations}

For temperature-dependent processes, one approximation that is available from the ground state KS potential is the Fermi-weighted KS technique.\cite{mattsson2006phase}  In this method, occupied and unoccupied orbitals found at zero-temperature can be weighted with a Fermi-Dirac distribution to obtain a finite temperature density.  

For exact computations at finite temperature, Appendix.~\ref{TDFT} discusses the temperature-dependent KS potential (which is very similar to the ground state case) and can be found if a finite temperature wavefunction is provided.\cite{poulin2009sampling,chowdhury2017quantum}

Other theories may require computation of other quantities, including those not relevant to quantum chemistry.  Yet, they can follow the same strategy as the RWMP algorithm ({\it i.e.}, re-use of the wavefunction and QAE) to find the relevant quantities.

\subsection{Non-density functional quantities}

Note that the RWMP algorithm is not specific to the density, the KS potential, or even a density functional.  Many quantities of interest can be obtained.  To see how the continued fraction representation of the Green's function can be obtained, see Ref.~\onlinecite{bakerGreen20}.

\subsection{Reduced examples for testing}

Throughout, it is implied that many qubits will be required, but we view this as similar to requiring sufficient memory on the classical computer.  However,  there are test systems that are used to provide insight into DFT and may be useful for proofs of principle.\cite{stoudenmire2012one,shulenburger2009spin,wagner2012reference,bakerPRB15,*baker2016erratum,helbig2009exact,helbig2011density,elliott2012universal,fuks2015time,lima2003density}  A simple model that may be within reach of existing quantum computers is the two-site Hubbard model, which has been used to study simplified DFT.\cite{fuks2014challenging,fuks2014charge,cohen2016landscape,carrascal2015hubbard,carrascal2018linear,smith2018thermal,sagredo2018accurate,smith2016exact,herrera2018melting}

Each individual algorithm has been tested in cases of relevant interest for other problems, many of which by the authors. Citations to many of the tests of these algorithms are included near their description.

\section{Limitations of known algorithms}\label{limitations}

This section discusses limitations on the types of algorithms that could have been used to machine learn the functional and how future improvements in algorithms must continue to improve the performance on quantum computers. We hope that a clear and complete discussion of the current hurdles for quantum computation motivate future algorithms and provide a consistent account.

\subsection{Feasibility of obtaining the starting state}\label{quantumspeedup?}

In the previous section, we omitted a discussion of how to prepare the initial state in superposition and how it is converged. The time it takes to obtain a wavefunction is known to be inefficient with current techniques. For some algorithms, the ideal input would have been a superposition of solutions that would include all combinations of particle numbers and spin polarizations for all external potentials. 

In this section, we consider the complications for even preparing a suitable state in superposition and known limits.  We will discuss the RTE algorithm in the context of its scaling, why the large prefactor can prohibit its use on some systems, comparing the implementation of RTE on the classical and quantum computer, and limitations to constructing a superposition of all solutions with any method.  This last point will address the feasibility of finding the original starting state for the QML.

Some other algorithms that were not used as subroutines in the RWMP method are also discussed.

\subsubsection{General considerations for real time evolution}

One of the primary motivations for making the RWMP algorithm for the density functional is to recycle the ground state wavefunction, reducing the cost of RTE.  It is true that RTE scales only polynomially (Appendix.~\ref{basis_sets}) with the number of orbitals.  In comparison, the full configuration interaction (FCI) gives the exact results but scales as $\mathcal{O}(\Ne{}!)$.\cite{vigor2015minimising}  So, RTE has a scaling advantage over FCI; however, the prefactor matters.

The exact same RTE algorithm can be run on a classical computer. The reason that quantum chemistry computations are not run with RTE is that the prefactor equal to the number of time steps is large.\cite{poulin2014trotter,mcclean2014exploiting}   This is especially problematic for large systems and those where a near-degeneracy must be avoided, necessitating a small step size.  But this is exactly where quantum computers are hoped to be applied.

Note also that the Trotterization of the time evolution operator is suited for planar molecules since interactions are comparatively localized, which is the same reason that matrix product states prefer these geometries\cite{schollwock2011density,chan2011density,baker2019m} further limiting the usefulness of RTE. 

Further comment becomes far more complicated because these alternative, smaller wavefunction ansatzs on the classical computer may display all the same features of the true ground state and accurate energy.  These solvers typically have systematic errors that are studied, can solve systems faster, and have led to solutions of large systems\cite{yang2014ab} and new discoveries.\cite{arodola2017quantum} So, it is not clear if a quantum computation can beat all classical representations in terms of efficiency. Note that on the classical computer, another solver can be used ({\it e.g.} tensor networks,\cite{schollwock2011density,baker2019m} quantum Monte Carlo,\cite{foulkes2001quantum} random phase approximation,\cite{eshuis2012electron,chen2017random} or many, many other methods\cite{helgaker2014molecular}).  The number of choices here is vast and storied, so we leave more discussion for others.\cite{pople1999nobel,helgaker2014molecular}

In summary, the time complexity of solving the quantum computation should be expected to be larger than classical solvers. Quantum computers can represent a reduction in the amount of space required to store a wavefunction, but it is not clear if this will always beat every classical representation. Improved methods of finding the ground state would be highly valuable.

\subsubsection{Choice of basis sets for algorithms}

It is true that the quantum computer provides a different representation of the wavefunction. In some representations for classical algorithms, the memory will grow considerably with system size.  For example, a coupled cluster calculation can have many coefficients as in the number of operations may become exponentially large.\cite{bartlett2007coupled}  This is due to the need to store more coefficients for a given basis.   The quantum wavefunction on the quantum computer may have some advantage here since coefficients are stored in superposition on each qubit.  However, there are many wavefunction ansatzs to consider in comparison because some methods can find accurate answers with considerably less coefficients, and the time needed for the quantum computer's wavefunction can be lengthy.  

Some recent efforts on the quantum computer have sought to impose specific conditions to reduce the amount of operations required.  The strategies that are pursued are to use different basis sets and generate criteria for removing terms from the Hamiltonian.  For example, plane-waves lead to a sparse Hamiltonian (see also Appendix.~\ref{basis_sets}).\cite{babbush2018low} Such systematic methods for this should be expected to be as difficult as (or more so than) solving the wavefunction outright.

The proper comparison between RTE on a quantum computer with a plane-wave basis set and RTE on a classical computer is that that the classical computer can handle any basis.  So, a comparison with a plane-wave basis set on a quantum computer should be compared to a RTE on the classical computer with some other basis set ({\it e.g.} Gaussians). The difference in the number of orbitals required to accurately simulate matter for plane-waves can be much higher than other basis sets due to cusps in the wavefunction that require large numbers of plane-waves to resolve.\cite{kimball1975short,klahn1984rates,hill1985rates,helgaker2014molecular} So, restricting the basis to plane-waves at best matches the classical equivalent in terms of operations required (see also Appendix.~\ref{basis_sets}).

Wavelets have appeared in some references recently suggesting that these functions can provide a systematic way to compress a Hamiltonian, but in the 30 year existence of these functions (for more information, see the references in Ref.~\onlinecite{bakerPRB18}), they have been demonstrated to scale poorly to large system sizes due to the curse of dimensionality.\cite{beylkin2002numerical} These methods typically give access to 1--3 electrons when not approximating the Fock operator. To use these functions, a compression ratio of nearly 100\% would be required.  Using these functions beyond one-dimension faces significant hurdles in the general, interacting case.

In conclusion, restrictions to a particular basis set or procedures to remove terms in a Hamiltonian is not a cure-all for computation in general.  The task of identifying which terms of the Hamiltonian without first solving the problem is typically complicated, and the resulting answer is not necessarily exact anymore. Approximated calculations are essentially the strategy of classical methods which leave out some effect or terms, and it is not clear how to do this systematically in general, nor if any simple strategy can be expected.  A single method or change in basis is not a panacea to making the solution on a quantum computer less complex.

\subsubsection{Limits on solutions in superposition}

With regards to applying any generic algorithm to find a superposition of solutions, we can place a limit on finding the initial state required for algorithms that need this superposition of solutions.

\begin{theorem}\label{nosuperRTE}
A quantum computer cannot efficiently generate a superposition of solutions for all potentials necessary for $F[n]$ unless BQP=QMA-complete.
\end{theorem}

\begin{proof}
It is also known that at least one of the systems contained in the universal functional is in the computational class QMA-complete to solve.\cite{oliveira2005complexity} This is not efficient on the quantum computer to compute objects in this complexity class.  So, no algorithm should be able to obtain all solutions efficiently.
\end{proof}
Thus, some elements should be expected to be unconverged in a superposition unless the algorithm is run for an impractically long time. Note that algorithms requiring a superposition generally require more than one solution, so the superposition is not run just once.

\subsubsection{Alternative algorithms to real time evolution}

We do not rule out that some improvement may allow RTE (or some other method) to receive a quantum advantage. If a superior algorithm is developed ({\it e.g.}, the tools exist to make imaginary time evolution,\cite{mcardle2019variational} perturbation theory,\cite{hackl2008real} preparation of projected entangled pair states,\cite{schwarz2012preparing} or a very expensive version of the density matrix renormalization group\cite{schollwock2011density,gilyen2019quantum,yanofsky2008quantum} at the present time), it is likely that it will rely on converging from some initial state.  Also, any algorithm like exact diagonalization for general systems is not expected to be efficient since this problem is not contained in the BQP complexity class.\cite{oliveira2005complexity,liu2007quantum,schuch2009computational} This creates more motivation to focus on algorithms that converge. 

We do note that progress through the decades on quantum chemistry has been difficult to find a unifying principle that would help in algorithm design,\cite{pople1999nobel} but perhaps quantum computing may motivate a new way to look at the problem for cases of interest since the prospects of finding a general algorithm are prohibitive.

\subsection{Other methods and quantum machine learning}\label{rationale}

In regards to other methods to train the ML model, we had investigated using alternative subroutines using a superposition of solutions.  Algorithms that we considered included Grover's algorithm,\cite{grover1997quantum,grover2001schrodinger} quantum walks,\cite{szegedy2004quantum,lemieux2020efficient} and others.  Each of these requires a solution of a superposition of systems, some means of identifying the correct solutions through a phase kickback, undoing the superposition of systems, and then repeating the process until the error is low enough to ensure the correct solution is determined to some high probability.  

The limitations described in Theorem~\ref{nosuperRTE} are one hurdle, but in our view, these strategies of uninformed search were too lengthy even on the quantum computer. Since the improvement in the number of steps required is only by a square root factor, searching the exponentially sized database causes the algorithm to run for far longer than other classical algorithms for electronic structure.

The variational quantum eigensolver (VQE)\cite{peruzzo2014variational} might be adapted to finding, for example, the Kohn-Sham potential with a method from Refs.~\onlinecite{gidopoulos2011progress,wagnerPRB14,jensen2016numerical,jensen2018numerical,callow2018optimal,kanungo2019exact,callow2020density,kumar2020general}, but it is not clear if errors can be kept small in a reasonable amount of time- and wish to keep focus on finding the exact KS potential.  The VQE would also require the exact density from another method, but it could in principle be adapted.  The main question is how well this suggestion would perform on a quantum computer based on current hardware.

\subsubsection{Quantum machine learning}

Quantum machine learning (QML) algorithms were also not suitable since known bounds on the number of oracle queries imply that there is only a polynomial speedup for QML algorithms. See Ref.~\onlinecite{servedio2004equivalences} and experimental proof on a simple case in Ref.~\onlinecite{riste2017demonstration}. A confirming statement is also found in in Ref.~\onlinecite{arunachalam2018optimal}.  

Lacking an exponential speedup, it is likely too expensive for the quantum computer to run in a reasonable amount of time here.

Existing statements in the literature on the hardness of determining the universal functional can lead to limits on the types of QML algorithms we expect can exist. We formalize the relevant limits in some statements here.\footnote{Theorems listed here are so heavily dependent on pre-existing proofs that they are probably more accurately called a corollary of those theorems, but this naming convention is used in several physics works and here to match.}

\begin{theorem}\label{QMLthm}
No QML algorithm can discover the universal functional in polynomial time unless QMA-complete reduces to the complexity class BQP.
\end{theorem}

\begin{proof}
The functional is proven to be QMA-complete to learn.\cite{schuch2009computational}  If an algorithm determined the functional in polynomial time, then BQP=QMA-complete.
\end{proof}

Note that this says nothing about whether QML can do slightly better than the classical algorithm, but typically a step in a QML algorithm is to re-prepare the wavefunction and this is one of the main issues we want to avoid here. Because the functional is known to be QMA-complete to learn, finding the exact and universal functional with the RWMP algorithm would require an exponential amount of time to visit all systems, as expected.\cite{oliveira2005complexity,liu2007quantum,schuch2009computational,whitfield2013computational}   We continue on to make a connection with some statements about learnability.

Theorem~\ref{QMLthm} also implies limits on how learnable the universal functional is by any method.  In Ref.~\onlinecite{servedio2004equivalences}, under the probably approximately correct (PAC)\cite{valiant1984theory} model of ML, only a polynomial reduction in oracle queries (training points) is possible with QML.

\begin{lemma}
The assumed limitations on the number of oracle queries (quantum and classical learning differ by only polynomial factors) required to discover $F[n]$ with QML are the same as those under the PAC model.\cite{servedio2004equivalences}
\end{lemma}

This is a consequence of the hardness of finding the functional.  If this were not true, we could find a QML algorithm that could discover $F[n]$ in exponentially fewer steps, which is a violation of both Theorem~\ref{QMLthm} and Ref.~\onlinecite{servedio2004equivalences}.  So, there can be no exponential speedup for QML under the PAC model here.

While there may be cases that lie outside of the PAC model,\cite{amsterdam1988valiant,bergadano1989error,haussler1990probably,buntine1990theory,pazzani1992framework} we have we no evidence that the functional is not subject to the PAC assumptions.  There is also good evidence to suggest that simple systems--at least--obey the limits of the PAC model.\cite{riste2017demonstration}

In summary, the difficulty of finding the universal functional on a quantum computer places limits on the ability for many-body solvers on the quantum computer and QML as well. Note the generality of the statements here for all QML algorithms. Recent works have shown that some existing QML algorithms are not as efficient as once thought,\cite{tang2019quantum,tang2018quantum} but the arguments here apply to QML in general for this problem. 

These general arguments do not prohibit a quantum advantage if another algorithm can be found to solve systems in a more specific case or restricted class of systems. Recent progress on finding classical algorithms that are superior to quantum algorithms illustrate the need for caution when proposing an efficient QML algorithm, however.\cite{chia2020sampling}  Still, QML does not seem to be a feasible way forward for the problem of interest here.

\subsection{Summary}

The main take-away from these statements is that an exponentially more efficient algorithm for the most general case is ruled out when discussing the solution on a quantum computer.  Reducing the prefactor, therefore, becomes highly beneficial.  This does not preclude algorithms on more restricted systems, but there is no hint of how exactly to construct such an algorithm or that this is any easier than a straight-forward solution. 

The RWMP method avoids repeated measurement, reduces the prefactor to solve each system iteratively and allows for more systems to be solved with RTE or some other method.

\section{Conclusion}

It has been demonstrated that a combination of algorithms applied to a wavefunction on the quantum computer can yield the Kohn-Sham potential, energy, and density matrix coefficients without completely re-preparing the ground state wavefunction each time. The determined quantities can be used to train a machine learned model using gradient-based methods either on the quantum computer or classically.  The ground state wavefunction was used as the starting point for the next system, reducing the prefactor and avoiding an expensive computation of the ground state at each step.  This efficiency was also used for the Kohn-Sham potential with a minimization condition. 

Once a model is created, a classical user can extract the relevant quantities from the machine learned model and use it for ground state, time-, and temperature-dependent calculations.  Finding the Kohn-Sham potential is especially useful here since it gives access to many properties of the ground state; in addition, there was some indication that the Kohn-Sham potential might scale better in some cases as opposed to finding the density. Known limitations on the complexity of finding the universal functional and quantum machine learning have constrained the choice of subroutines in the algorithm here.    A better method to solve for the ground state on the quantum computer must be a focus of future research to make quantum chemistry studies feasible, but this algorithm will allow for solutions to exported to many users.

\section{Acknowledgements}

\begin{CJK*}{UTF8}{gbsn}
T.E.B. acknowledges funding provided by the postdoctoral fellowship from Institut quantique and Institut Transdisciplinaire d'Information Quantique (INTRIQ). This research was undertaken thanks in part to funding from the Canada First Research Excellence Fund (CFREF).  We thank useful discussions at the 2018 New Trends in Quantum Error Correction workshop at Universit\'e de Sherbrooke.  The authors are thankful for useful discussions with Frank Verstraete, Guillaume Duclos-Cianci, Anirban Narayan Chowdhury, Jonathan A.~Gross, Colin Trout, Li Li (李力), Yehua Liu, Vamsee Voora, Shane Parker, Raphael Ribeiro, Agustin Di Paolo, Anirudh Krishna, Maxime Tremblay, Jessica Lemieux, Benjamin Bourassa, Stuart Clark, Rex Godby, and Nikitas Gidopoulos.
\end{CJK*}

\begin{appendix}


\section{Quantum chemistry}\label{qchem}

In this section, we review some background information on quantum chemistry.

\subsection{Many-body problems}\label{manybody}

The problem of interest is to solve the many-body problem expressed by the Hamiltonian\cite{fetter2012quantum}
\begin{equation}\label{secondQHam}
\mathcal{H}=\sum_{ij\sigma}\left(t_{ij}\hat c^\dagger_{i\sigma}\hat c_{j\sigma}+\sum_{k\ell\sigma'}\left(V_{ijk\ell}\hat c^\dagger_{i\sigma}\hat c^\dagger_{j\sigma'}\hat c_{\ell\sigma'}\hat c_{k\sigma}\right)\right)
\end{equation}
with fermionic operators $\hat c$ on discretized lattice sites (or basis functions) indexed by $i,j,k,\ell\in\{1,\ldots,\Nb{}\}$ (for $\Nb{}$ basis functions) with spin $\sigma$. Note the order of indices.\cite{fetter2012quantum,raimes1972many}  The one-electron integral is
\begin{equation}\label{oneEl}
t_{ij}=\int\varphi_i^*(\mathbf{r})\left(-\frac12\mathbf{\nabla}^2+v(\mathbf{r})\right)\varphi_j(\mathbf{r})\; \dbr
\end{equation}
which is the kinetic plus external potential terms. The two-electron integral is
\begin{equation}\label{twoEl}
V_{ijk\ell}=\frac12\iint\varphi_i^*(\mathbf{r})\varphi_j^*(\mathbf{r}')v_\mathrm{ee}(\mathbf{r}-\mathbf{r}')\varphi_k(\mathbf{r})\varphi_\ell(\mathbf{r}')\; \dbr\,\dbrp
\end{equation}
where $v_\mathrm{ee}(\mathbf{r}-\mathbf{r}')=1/|\br{}-\br{}'|$ for the case of a Coulomb interaction and that this expression assumes the orbitals for both spin-up and spin-down electrons are the same.  Note that a Hubbard model is an approximation with only the most diagonally dominant terms of the Coulomb operator, $V_{ijk\ell}=U\delta_{ij}\delta_{jk}\delta_{k\ell}$ for Hubbard interaction $U$.\cite{hubbard1963electron} We have restricted our consideration to the Born-Oppenheimber approximation,\cite{born1927quantentheorie} even though the discussion can be generalized to the motion of nuclei.\cite{gidopoulos2014electronic}

Solving the entire many-body problem is known to be difficult if not impossible.  However, approximate methods can yield results that are accurate to what is known as chemical accuracy (1 mHa) or a stricter limit applies in some cases.\cite{helgaker2014molecular}

\subsection{Basis sets}\label{basis_sets}

We can note that Eq.~\eqref{manybody} has been written in the second quantized form since we expect to need a basis to truncate the problem to a more manageable size. One may, for example, choose Gaussian orbitals\cite{boys1950electronic} so that Eqs.~\eqref{oneEl} and \eqref{twoEl} can be evaluated analytically and chemical accuracy can be obtained with only a few functions.  Other basis functions can also be chosen.\cite{helgaker2014molecular}

It is known that Eq.~\eqref{twoEl}, when represented in a local basis, reduces to\cite{helgaker2014molecular}
\begin{equation}\label{localscaling}
\underset{\mathrm{loc.}}
\lim \;V_{ijk\ell}=\frac12\iint\left|\gamma_{ik}(\mathbf{r},\mathbf{r}')\right|^2v_\mathrm{ee}(\br{}-\br{}')\; \dbr\,\dbrp\sim\mathcal{O}(N^2)
\end{equation}
for a density matrix $\gamma$ where the limit is taken for well separated, local orbitals at large distances.  This reduces the computational complexity from $\mathcal{O}(N^4)$ to $\mathcal{O}(N^2)$ in the asymptotic limit, although the true scaling lies somewhere in-between depending on the details of the system.\cite{poulin2014trotter}  This argument only applies to orbitals that drop off sufficiently quickly with distance from the origin. 

\subsubsection{The curse of dimensionality and other limitations in point-like basis sets}

We note that the reduction in Eq.~\eqref{localscaling} happens immediately when using purely local basis sets, with no spatial extent, is used. In that case, we would reduce to a sum over only the diagonal elements of the two-electron integral, $V_{ijk\ell}$, if the orbitals were point-like.  However, using only these localized orbitals ({\it e.g.}, grid points, plane-waves, wavelets, etc.) comes at a steep price.  

In particular, note that wavelets are very expensive for large scale quantum chemistry problems. It has been known for some time now that a curse of dimensionality shows that the number of functions in one dimension scales as $N_\mathrm{1D}^d$ for $d$ dimensions with a number of functions in one-dimension, $N_\mathrm{1D}$.\cite{beylkin2002numerical}  Due to the large number of basis functions, wavelet based functions have only been able to solve 2 and 3 electron systems maximum in the general case,\cite{bischoff2012computing} although these functions can be efficient for larger noninteracting or single Slater-determinant theories or other cases of very particular interest.\cite{harrison2016madness} Wavelets are simply not expected to be efficient for real three-dimensional systems of any meaningful size based on pre-existing works unless the problem is converted to a noninteracting theory or a special geometry is chosen. For more information, see the references in Ref.~\onlinecite{bakerPRB18}.  

For plane-wave functions, many thousands of functions are required to resolve the electron-electron cusp ({\it e.g.}, the behavior of the wavefunction at the nucleus in a hydrogen 1s orbital).\cite{kimball1975short,klahn1984rates,hill1985rates,helgaker2014molecular}    We will not consider point-like basis functions further here to concentrate on the general case, although plane-waves can be useful for periodic systems. With respect to the density matrix (which is a highly important quantity in Sec.~\ref{DFT}), the full double sum will be taken (see Eq.~\eqref{densrho}) and not just diagonal elements. 

In summary, even though the scaling of the two-electron operator is better for point-like basis sets, many more basis functions will be required to obtain accurate results except in special cases.  So, choosing a point-like basis function will not represent a general strategy for all types of quantum chemical problems that we may wish to solve.

\section{Density functional theory}\label{DFT}

The foundations of density functional theory (DFT), including the Kohn-Sham system and other variants, are introduced here.

A compact representation of the quantum ground state is the one-body density.  In DFT, the ground state wavefunction is replaced with the density. It was proven in Ref.~\onlinecite{hohenberg1964inhomogeneous} that the one-body density, defined as
\begin{equation}
n(\br)=\int\ldots\int\left|\Psi(\br,\br_2,\ldots,\br_{\Ne{}})\right|^2 \dbr_2\ldots \dbr_{\Ne{}}
\end{equation}
is sufficient to characterize the ground state. Note that in order to obtain this quantity on the lattice, the one-body reduced density matrix must be obtained for $N_e$ electrons,
\begin{eqnarray}
\hat\rho(\br,\br')&=&\int\ldots\int\Psi^*(\br,\br_2,\ldots,\br_{N_e})\\
&&\hspace{1.5cm}\times\Psi(\br',\br_2,\ldots,\br_{N_e})\; \dbr_2\ldots \dbr_{N_e}\nonumber
\end{eqnarray}
and is related to the density in the limit where $\br\rightarrow\br'$
\begin{equation}\label{densrho}
n(\br)=\frac12\sum_{ij}\varphi_i^*(\br)\rho_{ij}\varphi_j(\br)
\end{equation}
where $\rho_{ij}=\langle\Psi| \hat c^\dagger_i\hat c_j|\Psi\rangle$.  A spin index has been suppressed, signifying a spin degenerate ground state.  However, extensions to ground states without spin degeneracy are also available.\cite{gross2013density}

Having replaced the wavefunction for the more compact density, the Hamiltonian must be substituted for another mathematical object that acts on the density.  In general, an object that maps a function to a scalar value is known as a functional.\cite{reed2012methods} In DFT, a functional maps the one-body density to a scalar energy value. 

In order to find the ground state energy, we can use a minimization over all densities\cite{levy1985constrained}
\begin{equation}\label{EDFT}
E=\underset{n}{\mathrm{min}}\left(F[n]+\int n(\br)v(\br)\; \dbr\right)
\end{equation}
although it is impractical to search for the ground state density with this formulation.
The second term in Eq.~\eqref{EDFT} is the external potential functional (often denoted as $V[n]$) and has a known form.  Contrastingly, the universal functional, $F[n]$, is defined as the search over all wavefunctions $\Psi$ constrained to give the density, \cite{engel2011density,gross2013density}
\begin{equation}\label{FDFT}
F[n]=\underset{\Psi\rightarrow n}{\mathrm{min}}\langle\Psi[n]|\hat T+\hat V_\mathrm{ee}|\Psi[n]\rangle.
\end{equation}
and is common to all systems since it does not depend on the external potential.  Clearly, the minimization is not an efficient way to find the functional, but it is useful as a mathematical tool.

Because $F[n]$ is unknown explicitly (its existence is proven by contradiction), it requires approximation to use.\cite{kohn1965self}  Some limiting cases are known, such as one- or two-electron cases, the uniform gas via a fitting procedure, and one-dimension.\cite{fermi1927metodo,thomas1927calculation,weizsacker1935theorie,ribeiro2015corrections}  Many exact properties of the functional are known from rigorous mathematical statements\cite{lieb1981improved,pittalis2011exact} or limited test cases.\cite{wagner2013guaranteed,fuks2015time}  One common way to design new functionals is to build in exact conditions.\cite{sun2015strongly,mori2008localization,cohen2008insights}

Note that to solve a problem with $F[n]$, the functional derivative can be used and is defined as\cite{engel2011density,gross2013density}
\begin{align}\label{funcderivdef}
\int\frac{\delta\Z[g(\x)]}{\delta g(\x)}\force(\x)\;\mathrm{d}\mathbf{x}&\equiv\underset{\eta\rightarrow0}\lim\Big(\frac{\Z[g(\x)+\eta\force(\x)]-\Z[g(\x)]}\eta\Big)\\
&=\left.\left(\frac{\mathrm{d}}{\mathrm{d}\eta}\Z[g(\x)+\eta\force(\x)]\right)\right|_{\eta=0}\label{epsderiv}
\end{align}
where $\force$ is an arbitrary test function, $g$ is the function we wish to evaluate $F$ around, and $\eta$ is a small parameter.  The first functional derivative is most well-known from classical physics where it is used to minimize the Lagrangian via Euler-Lagrange minimization.\cite{goldstein2014classical}

In order to find the minimal density, a functional derivative can be taken.  This is synonymous with the Euler-Lagrange equations in this case\cite{gross2012introduction}
\begin{equation}\label{ELDFT}
\frac{\delta F[n]}{\delta n(\br)}+v(\br)=\mu
\end{equation}
where a constant chemical potential $\mu$ was added as a Lagrange multiplier for the total particle number.  This equation is then used to solve for orbital-free DFT.

\subsection{Kohn-Sham density functional theory}\label{kohnsham}

One useful alternative formulation of $F[n]$ is KS-DFT.  This reformulation of DFT proposes an external potential whose solution resulting one-body density is equivalent to obtaining the one-body density of the fully interacting system.  The original goal of DFT was to propose a purely wavefunction-free method to characterize the quantum ground state, but it is difficult to find suitable approximations that are accurate enough.  

It can be noted that approximating $F[n]$ is a large approximation on the total energy.  In this alternative formulation, one introduces an easy to solve, noninteracting, auxiliary system to make the required approximation a smaller fraction of the overall energy.  Obtaining the KS potential gives insight to many more physical quantities than just the density, and the orbitals of the noninteracting system can be used in a variety of other contexts.

\subsubsection{Finding the Kohn-Sham potential}\label{findKS}

To formalize the KS system, what is known as the adiabatic connection can be used to transform from the original problem to the final noninteracting problem.\cite{seidl1999strong}  Equation~\eqref{secondQHam} can be rewritten as\cite{langreth1975exchange} 
\begin{equation}
\mathcal{H}^\lambda=\hat T+ \lambda \hat{V}_\mathrm{ee}+\hat{V}^{(\lambda)}
\end{equation}
where an express dependence on the coupling constant $\lambda$ has been introduced. The tuning parameter $\lambda$ can vary between the KS system ($\lambda=0$, where $\hat V^{(\lambda=0)}=\hat V\s$) and the original system ($\lambda=1$). Note that the external potential operator ($\hat V=v(\br)$) has received an implicit coupling constant dependence, but there is no simple analytic form for $\hat{V}^{(\lambda)}$.  The constraint given in this problem is that the density must be the same for any $\lambda$,
\begin{equation}\label{KSmatching}
n(\br)\equiv n^{(\lambda=1)}(\br)\overset!=n^{(\lambda=0)}(\br)
\end{equation}
which is difficult to construct in practice.  Note that the other limit of $\lambda\rightarrow\infty$ can also be used to base a functional theory.\cite{seidl1999strictly}

\subsubsection{Components of the functional in the Kohn-Sham system}

The form of the universal functional for the KS case is
\begin{equation}\label{FKS}
F[n]=T\s[n]+U[n]+E_\mathrm{xc}[n]
\end{equation}
This form is known from perturbative expansions of the many-body system.\cite{kohn1965self,fetter2012quantum} Note that the subscripted "s" on the kinetic energy is to signify that $T\s$ is evaluated over noninteracting wavefunctions, $\phi$, but has the same form as the same kinetic energy operator in Eq.~\eqref{oneEl}.  This term shows that the KS scheme is not a pure density functional but one that relies the auxiliary noninteracting orbitals.  The cost to solve the noninteracting system is larger than pure-DFT, but still smaller than many other approximations.  

In addition to the kinetic energy, another known energy in Eq.~\eqref{FKS} is the Hartree energy,
\begin{equation}
U[n]=\frac12\iint\frac{n(\br)n(\br')}{|\br-\br'|} \dbr\, \dbrp
\end{equation}
which is fully nonlocal.

The unknown term in Eq.~\eqref{FKS} is the exchange-correlation energy, $E_\mathrm{xc}$, which is not known as a density functional and requires approximation in practice.  If the exact $E_\mathrm{xc}$ is used, then the theory is exact.  The usefulness of defining the KS system is that the approximation to the total energy is small for many systems of practical interest. 

\subsubsection{Kohn-Sham potential by functional derivatives}

The KS potential is explicitly
\begin{eqnarray}\label{vks}
v\s(\br)&=&\frac{\delta}{\delta n} \Big(U[n]+E_\mathrm{xc}[n]+\int n(\br)v(\br)\;\dbr\Big)\\
&=&v_\mathrm{H}[n](\br)+v_\mathrm{xc}[n](\br)+v(\br)\nonumber
\end{eqnarray}
where a functional derivative is taken over the relevant energy terms, for example,
\begin{equation}
v_\mathrm{H}[n](\br)=\frac{\delta U[n]}{\delta n}=\int\frac{n(\br')}{|\br-\br'|}\dbr'
\end{equation}
for the Hartree potential, and the form of $v_\mathrm{xc}[n](\br)$ is not known explicitly. In summary, by re-grouping the non-kinetic energy terms in the Hamiltonian, $U[n]+V[n]+E_\mathrm{xc}[n]$, the resulting system will appear as noninteracting.  The electron-electron term is contained in the resulting potential of the noninteracting system.

\subsubsection{Variational principle for the Kohn-Sham potential}\label{variationalKSdetails}

The Kohn-Sham potential also satisfies the minimization of the quantity\cite{gidopoulos2011progress}
\begin{equation}\label{KSmin}
T_\Psi[v\s]=\langle\Psi[v]|\hat T+\hat{V}\s|\Psi[v]\rangle-\langle\Phi[v\s]|\hat T+\hat{V}\s|\Phi[v\s]\rangle
\end{equation}
where we follow Refs.~\onlinecite{gidopoulos2011progress,callow2018optimal,callow2020density} closely. Note that $\Psi(\br,\br_2,\ldots,\br_{N_e})$ is not an eigenstate of $\hat T+\hat{V}\s$ but that $\Phi(\br)$ is. So,
\begin{equation}
\langle\Psi[v]|\hat T+\hat{V}\s|\Psi[v]\rangle>\langle\Phi[v\s]|\hat T+\hat{V}\s|\Phi[v\s]\rangle
\end{equation}
by the variational principle.\cite{townsend2000modern} The functional derivative of Eq.~\eqref{KSmin} with respect to $v\s(\br)$ is\cite{gidopoulos2011progress}
\begin{equation}
\frac{\delta T_\Psi[v\s]}{\delta v\s}=n_\Psi(\br)-n_\Phi(\br)
\end{equation}
and equals the difference in the densities of the two systems, one computed from $\Psi$ ($n_\Psi$) and the other density from $\Phi$ ($n_\Phi$).\cite{gidopoulos2011progress} When this difference is zero, the condition for the Kohn-Sham potential is found given in Eq.~\eqref{KSmatching}.

\subsubsection{$v$-representability}

The KS scheme is exactly defined provided that $v$-representability is satisfied.  In common practice, this is not a concern since it was proven on a grid that the system must be $v$-representable since the kinetic energy is regularized.\cite{kohn1983v,chayes1985density} So, we always expect $v$-representability here.

\subsubsection{Minimization of the Kohn-Sham functional}\label{minKS}

Note that the Euler-Lagrange minimization of the functional yields the KS equations
\begin{equation}
\left(-\frac{\nabla^2}2+v\s(\br)\right)\phi_j(\br)=\epsilon_j\phi_j(\br)
\end{equation}
for some KS energy eigenvalues $\epsilon_j$ and KS orbitals $\phi_j(\br)$.  The density is then the sum over occupied orbitals equivalent to 
\begin{equation}
n(\br)=\sum_{j\in\mathrm{occ.}}|\phi_j(\br)|^2
\end{equation}
which can be found from Eq.~\eqref{densrho} by noting that the excitations are orthogonal.  One recovers Eq.~\eqref{densrho} with an additional index for the excitations when $\phi$ is decomposed into a chosen basis.

\subsubsection{Relationship between the energies of the Kohn-Sham and the fully interacting system}

The adiabatic connection from Sec.~\ref{findKS} does not conserve energy.  The relation between the ground state energy of the interacting system, $E$, and the energy of the KS system (the sum of eigenvalues of the noninteracting system, $\sum_{j\in\mathrm{occ.}}\epsilon_j$) is\cite{engel2011density}
\begin{equation}\label{KSEdiff}
E=\sum_{j\in\mathrm{occ.}}\epsilon_j-U[n]+E_\mathrm{xc}[n]-\int n(\br)v_\mathrm{xc}(\br) \dbr
\end{equation}
for Hartree energy $U$, exchange correlation energy $E_\mathrm{xc}$, and exchange-correlation potential $v_\mathrm{xc}$, Note that Eq.~\eqref{KSEdiff} shows it is not sufficient to have only the KS potential to find $E$, although perturbation theory on the density can be used.\cite{gorling1994exact}

\subsection{Potential functional theory}\label{PFT}

When examining Eq.~\eqref{EDFT}, it is natural to ask if a dual theory can be formulated based on $v(\br)$ instead of $n(\br)$ since both are one-body quantities.  This question stems from noticing that functional derivatives of $n(\br)$ yield equations that can be solved for the density, resulting in the Euler-Lagrange minimization for the density functional from Eq.~\eqref{ELDFT}.

To the question: can we instead take a functional derivative with respect to $v(\br)$ instead?  The answer is yes.  It was proven in Ref.~\onlinecite{yang2004potential} that the dependence on the functional in terms of the external potential was sufficient to describe the ground state.  In this theory, the density must be determined from $v(\br)$ directly as $n(\br)\rightarrow n[v](\br)$.  The resulting energy becomes
\begin{equation}\label{EPFT}
E=\underset {n[v]}{\mathrm{min}}\left(F[v]+\int n[v](\br)v(\br)\; \dbr\right)
\end{equation}
where 
\begin{equation}
F[v]=\underset{\Psi\rightarrow n[v]}{\mathrm{min}}\langle\Psi[v]|\hat T+\hat V_{ee}|\Psi[v]\rangle
\end{equation}
which is similar to Eq.~\eqref{FDFT}.

\subsubsection{Why the functional derivative is also necessary}

Very importantly, one cannot determine the entire character of the ground state ({\it i.e.}, find $n(\br)$ or equivalent) with only $E[v]$.  To see this, note that the Euler-Lagrange equation for potential functional theory is\cite{cangi2013potential}
\begin{equation}\label{PFTeq}
\frac{\delta F[v]}{\delta v}+\int\frac{\delta n[v](\br)}{\delta v(\br')}v(\br')\;\mathrm{d}\br'=0
\end{equation}
where a derivative of the density with respect to the external potential in the second term of the left-hand side must be determined to solve this equation and find the density. In summary, one can formulate potential functionals, $E[v]$, provided that $n[v](\br)$ is known.\cite{gross2009adiabatic,cangi2011electronic} In other words, the functional derivative is also necessary to perform self-consistent calculations if only the energy is known for a given potential.

\subsection{Time-dependent density functional theory}\label{tddft}

In a time-dependent DFT (TD-DFT), one may simply propagate the KS potential according to Schr\"odinger's equation for time evolution\cite{gross2012introduction,ullrich2011time}
\begin{equation}
i\frac\partial{\partial t}\phi_j(\br{},t)=\left(-\frac{\nabla^2}2+v\s[n,\Phi_0](\br{},t)\right)\phi_j(\br{},t)
\end{equation}
with an initial starting state $\Phi_0$. A formal justification for the existence of TD-DFT is available.\cite{runge1984density,van1998causality}

Computing response functions is also necessary if a perturbation to $v\s(\br)$ is applied.  Knowing just the KS orbitals is sufficient to determine the response function for the KS system,\cite{ullrich2011time}
\begin{align}\label{chis}
\chi\s(\br{},\br{}',\omega)=&\underset{\eta\rightarrow0^+}\lim\sum_{k,j=1}^\infty(\xi_k-\xi_j)\frac{\phi_k^*(\br)\phi_j(\br)\phi_j^*(\br')\phi_k(\br')}{\omega-(\epsilon_j-\epsilon_k)+i\eta}
\end{align}
with occupation numbers $\xi_j$, eigenvalues $\epsilon_j$, frequency $\omega$, and small parameter $\eta$.  Hence, knowing all eigenvalues of the $v\s(\br)$ at $t=0$ gives the KS response function.  One can also relate $\chi\s$ to the interacting response function $\chi$ via a kernel ($n_\mathrm{g.s.}$ is the ground state density)
\begin{equation}
f_\mathrm{xc}[n](\br{},t,\br{}',t')=\left.\frac{\delta v_\mathrm{xc}[n](\br,t)}{\delta n(\br',t')}\right|_{n=n_\mathrm{g.s.}}
\end{equation}
and the relation
\begin{equation}
f_\mathrm{xc}(\br{},\br{}',\omega)=\chi\s^{-1}(\br,\br',\omega)-\chi^{-1}(\br,\br',\omega)-v_\mathrm{ee}(\mathbf{r}-\mathbf{r}')
\end{equation}
which is similar to a Dyson's equation. Many cases of interest obtain sufficiently accurate answers with only the adiabatic approximation, however.  TD-DFT can be used to find excited states.\cite{elliott20093}

\subsection{Density functional theory at finite temperature}\label{TDFT}

In order to incorporate finite temperature effects into the density functional, an entropy term can be added following the original treatment by Mermin,\cite{mermin1965thermal} one can write the grand canonical free energy as
\begin{equation}
\hat\Omega=\mathcal{H}-\tau\hat S-\mu\hat N
\end{equation}
for temperature $\tau$, chemical potential $\mu$, number operator $\hat N$, and entropy operator
\begin{equation}
\hat S=-k_B\ln\hat \Gamma
\end{equation}
where \cite{pribram2014thermal}
\begin{equation}
\hat \Gamma=\sum_{N_e,i}p_{N_e,i}|\psi_{N_e,i}\rangle\langle\psi_{N_e,i}|
\end{equation}
with $\sum_{\{N_e,i\}}p_{N_e,i}=1$, $0\leq p_{N_e,i}\leq1$, and states $\psi_{N_e,i}$ indexing excitations over a particular number of particles $N_e$.

The minimum of $\hat \Omega$ is (and adding descriptive indices)
\begin{equation}
\hat\Omega^\tau_{v,\mu}=\underset{n}{\mathrm{min}}\left\{F^\tau[n]+\int n(\br)(v(\br)-\mu)\dr\right\}
\end{equation}
and
\begin{equation}
F^\tau[n]\equiv\underset{\hat\Gamma\rightarrow n}{\mathrm{min}}\left\{T[\hat\Gamma]+V_\mathrm{ee}[\hat\Gamma]-\tau S[\hat\Gamma]\right\}
\end{equation}

In summary, if a system is solved at a given temperature, one can solve for the KS potential analogously to the ground state with an extra term $-\tau\hat S$ (and a term for the particle number) representing the entropy in the functional.  Note that the ground state density is replaced by the $\hat\Gamma$ object which is akin to a density matrix and that extra weights must be solved.  In order to find this, several states $\psi_{N_e,i}$ must be used.

\subsection{Comment on density functional approximations}

There are many functionals that can be used to approximate $E_\mathrm{xc}$. Each performs with its own set of systematic deficiencies.

The most pertinent approximations for this paper are the ML functionals.\cite{SRHM12,LBWB16,brockherde2017bypassing,bogojeski2019density,bogojeski2018efficient,nagai2018neural,li2016understanding,denner2020active,nagai2020completing,manzhos2020machine,suzuki2020machine,wetherell2020insights} The general strategy of fitting a functional may be unpalatable,\cite{medvedev2017density} but the generic strategy of fitting exact data is not unique to ML functionals.  The simplest approximation to the functional--known as the local density approximation--is a fit of highly accurate quantum Monte Carlo data.\cite{ceperley1980ground,vosko1980accurate} Further, coefficients present in hybrid functionals are also fit to existing data,\cite{becke1993new} among other examples.  The ML functionals simply represent a more robust approximation that can interpolate well provided the system solved is close to the training manifold.  This strategy would capture the exact conditions of the exact functional.\cite{LBWB16,hollingsworth2018can,perdew1982density,engel2011density,gross2013density}


\section{Training a machine learning model with stochastic gradient descent}\label{SGD}

Machine learning methods rely on the minimization of some cost function.  Here we describe the most basic version of this, the stochastic gradient descent (SGD).  If a function must be minimized, some procedure for the minimization is necessary.  We can define a cost function, $\MLcostfct{}$, that could take the form
\begin{equation}\label{costfct}
\MLcostfct{}=\sum_{\iota,\mathbf{x}}\left(n_\mathbf{\hat w}(\mathbf{x}_{(\iota)})-n_{(\iota)}\right)^2
\end{equation}
for some observable $n$ with known values of $n_{(\iota)}$ (called a training set) indexed by $\iota$ with some coefficients for the weight $\mathbf{\hat w}$ and bias $\mathbf{b}$ in the form
\begin{equation}\label{costform}
\mathbf{x}^{(i+1)}= \mathscr{S}\left(\mathbf{\hat w}^{(i)}\cdot\mathbf{x}^{(i)}+\mathbf{b}^{(i)}\right)
\end{equation}
for a level $i$ of the neural network with a nonlinear function $\mathscr{S}$ with a given input $\mathbf{x}_{(\iota)}$. The final level of the neural network will be the final quantity of interest, $n_\mathbf{\hat w}(\mathbf{x}_{(\iota)})$.

In order to minimize $\MLcostfct{}$ and therefore construct the best approximation to the known values, a gradient descent can be performed. The basic idea is to ensure that any evolution of the coefficients $\mathbf{\hat w}$ occur along the steepest negative gradient in the system.  In order to ensure that the gradient is negative, we can start from a consequence of Gauss's law which states that a gradient of a scalar (here, $\MLcostfct{}$) along the direction of changing $\mathbf{\hat w}$ (denoted $\delta\mathbf{\hat w}$) is equivalent to the Laplacian of $\MLcostfct{}$, $\Delta \MLcostfct{}={\boldsymbol{\nabla}}_{\mathbf{\hat w}}\MLcostfct{}\cdot\delta\mathbf{\hat w}$ where ${\boldsymbol{\nabla}}_{\mathbf{\hat w}}\MLcostfct{}=\left(\partial_{w_1} \MLcostfct{},\partial_{w_2}\MLcostfct{},\ldots\right)$.

If the form of the update from iteration over a time $\delta t$ (which can also be expressed in the discrete case as earlier) is chosen as
\begin{equation}\label{paramupdate}
\mathbf{\hat w}_{t+\delta t}=\mathbf{\hat w}_t-\eta{\boldsymbol{\nabla}}_{\mathbf{\hat w}}\MLcostfct{}
\end{equation}
for some small parameter $\eta$ (called a learning rate), then then applying ${\boldsymbol{\nabla}}_{\mathbf{\hat w}}\MLcostfct{}$ to each side with a dot product gives $\Delta \MLcostfct{}=-\eta\left|{\boldsymbol{\nabla}}_{\mathbf{\hat w}}\MLcostfct{}\right|^2<0$. In other words, the Laplacian of $\MLcostfct{}$ is guaranteed to be negative for a small enough $\eta$ such that $\MLcostfct{}$ is linear on a small neighborhood. The argument was shown for the weights of the neural network, but the same argument applies to the biases.

Evaluating Eq.~\eqref{costfct} for all $\iota$ provided can be very costly since the resulting gradient is very noisy.  To speedup the gradient descent, a randomly sampled subset of the provided training data, called a mini-batch.  This can be one or more selected points.  Since the points are randomly selected at each step, the gradient descent is taken stochastically.  

Even though the cost function evaluated over the mini-batch must be negative, the entire cost function evaluated over all training points does not need to decrease.  However, on average, the observable's value is lowered over several SGD steps.  The condition is $\left\langle\partial_t\MLcostfct{}\right\rangle<0$ for a mini-batch in the general case but would be $\partial_t\MLcostfct{}<0$ if all points are used.

There are two steps required when training a ML model with SGD: forward and backward propagation.  In forward propagation, Eq.~\eqref{costform} is applied straightforwardly from the input to the output layers of the neural network.  The backward propagation requires that a derivative of Eq.~\eqref{costform} from the output to input layers.  Repeating this constructs the cost function and then applies the gradient to all weights and biases in the network until the model is more converged.

More advanced algorithms can also be used to converge gradient descents as well.\cite{ruder2016overview}

\subsection{Gradient-free methods}\label{gradfree}

This section has been focused on training ML models with gradients.  Another class of method, one that uses random walks,\cite{vcerny1985thermodynamical,bertsimas1993simulated,tierney1994markov,vishwanathan2006fast,neal2012bayesian,rios2013derivative,pillai2014noisy,perozzi2016walklets,tejedor2018dimensional,ghosh2018journey} is also available.  However, these methods of training that involve a random walk typically only perform well on small numbers of parameters.  In fact, these methods perform better than gradients in some cases for these small systems.  If the problem is too large, then a gradient-based method is generally better.\cite{rios2013derivative}  There may also be opportunities to combine the two,\cite{tejedor2012optimizing,maclaurin2015gradient} although this may require that the wavefunction is re-solved to implement on the quantum computer.  We were not able to find any evidence that random walks can reliably compete with gradient-based training methods; however, if one could use such an algorithm, a quadratic speedup is available on the quantum computer for training with the random walk.\cite{szegedy2004quantum}

We also note that kernel based methods can be used to train the neural network but that they are not as ``choice-free'' since a functional form of the kernel must be selected. The minimization of the coefficients with a kernel does not require gradients. \cite{vu2015understanding}

\section{Quantum algorithms}\label{qalgs}

On a quantum computer, both the method of manipulating and storing information is different from a classical computer.  On a classical computer, electrons are moved around and represent different information based on where they are placed.

In a quantum computer, the quantity that we manipulate is the spin of an electron or some other quantity that is allowed to exist in a superposition of states.\cite{nielsen2010quantum}  We will refer to qubits in this work but note that one can extend the ideas to qudits where more than two states are possible.  It is not possible to determine all coefficients of a wavefunction in a superposition without an exponential number of measurements to find them ({\it i.e.}, measuring the spin for the both $|0\rangle$ and $|1\rangle$ states for each qubit individually).  To manipulate the state of a quantum wavefunction, we can apply operators, specifically operators that are unitary.  Often, the operators are applied to a specified qubit and also another auxiliary qubit to keep required operations unitary.

Sequences of these unitary operators can be cast as a tensor network diagram.  Each line of the diagram represents a single qubit or a group of qubits called a register.  Each block is a unitary operation that manipulates the state of one or more qubits.  Note that the allowed operations on the quantum computer maintain the number of lines and do not involve truncations of the space at any point.  This is qualitatively different from other uses of tensor network diagrams used for ``tensor network methods'' which can refer to a class of algorithms that solve for the ground state of a quantum system on classical computer.\cite{baker2019m} So, there is a subtle distinction between the tensor network methods and the tensor network diagrams we draw here even though the general properties of the diagrams is the same in both.  The main difference here is that no form of truncation or a connection with a renormalization group is taken explicitly.

One example of a useful gate that will appear in several places is the Hadamard gate,
\begin{equation}\label{hadamard}
H=\frac1{\sqrt2}\left(\begin{array}{cc}
1 & 1\\
1 & -1
\end{array}\right)
\end{equation}
which is written in the $\{|0\rangle,|1\rangle\}$ qubit basis states.  Applying this gate on a $|0\rangle$ state will give an equal superposition over the $|0\rangle$ and $|1\rangle$ states, 
\begin{equation}
H|0\rangle=\left(|0\rangle+|1\rangle\right)/\sqrt2
\end{equation}  
which can be applied identically to more qubits as denoted by the $\otimes$ operator.

We present some common algorithms in the context of solving quantum chemistry problems.\cite{blais2003algorithmes}  The quantum gradient algorithm (QGA) is used to find derivatives of a function, generically.  Meanwhile, the quantum phase estimation (QPE) determines the phase of a given state.  An important sub-algorithm in the QGA and QPE is the quantum Fourier transform (QFT) which is detailed first.\cite{nielsen2010quantum}  Real time evolution (RTE) in addition to quantum amplitude estimation (QAE) are also discussed.

Note that everywhere we use the symbol $N$ for the number of qubits in this section.  Many times in the literature, the number of qubits may also  implicitly mean $N$ registers with $r$ qubits each for $r$ digits of precision.  If this is the case, the same concepts would apply, regardless.

\subsection{Quantum Fourier transform}\label{QFT}

The QFT begins with a set of input data recorded on an initial set of registers, which we denote as $|y\rangle$.  The end result of the QFT is to change the data $|y\rangle$ into $|x\rangle$ according to the discrete Fourier transform\cite{kittel1987quantum}
\begin{equation}\label{qft}
|x\rangle=\frac1{2^{N/2}}\sum_{y=0}^{2^N-1}\exp(i2\pi xy/2^N)|y\rangle
\end{equation}
for $N$ qubits with a normalization factor $2^{-N/2}$ coming from the prefactor in Eq.~\eqref{hadamard}. The key first step is to understand how Eq.~\eqref{qft} can be re-expressed in binary instead of integers integer $y$.  A number $y$ can be re-expressed in base 2 with the digits (assuming value either 0 or 1, corresponding to the states of a qubit) $y_0y_1\ldots y_{N-1}$ where $N$ is the maximum number of (qu)bits we allow for the binary number.  In full, $y$ relates to the binary digits as
\begin{equation}\label{ybinary}
y=2^{N-1}y_0+2^{N-2}y_1+\ldots+2y_{N-2}+y_{N-1}
\end{equation}
A similar form can be written for using a qudit where a different basis is used ({\it e.g.}, trinary).

To express the binary representation of $|y\rangle$ as qubits, let us express the state as
\begin{equation}\label{ybits}
|y\rangle=|y_0\ldots y_{N-1}\rangle=|y_0\rangle\otimes\ldots\otimes|y_{N-1}\rangle\equiv\bigotimes_{\ell=0}^{N-1}|y_\ell\rangle
\end{equation}
with each of the sub-indices representing another digit of $y$'s binary representation ($y_\ell\in\{0,1\}$).

\begin{figure}[t]
\includegraphics[width=\columnwidth]{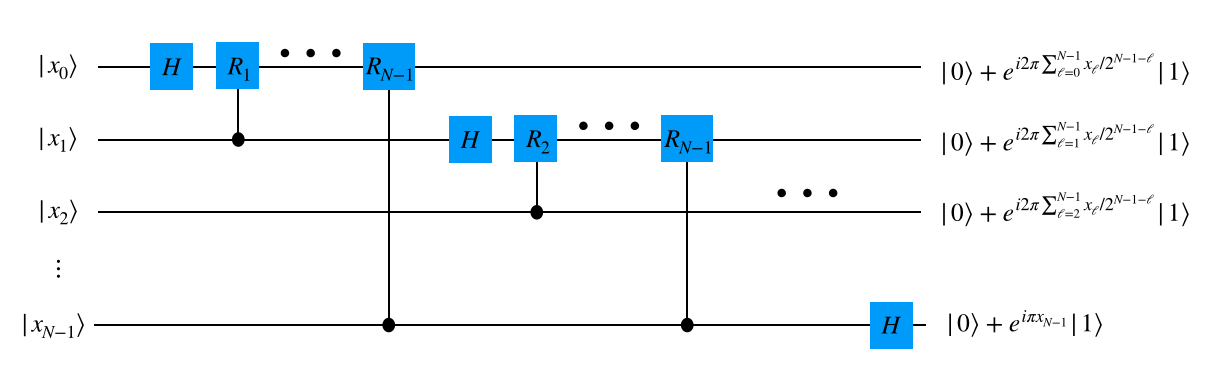}
\caption{Circuit diagram of the quantum Fourier transform. Note that the rotation operator is about the $z$ axis, and the Hadamard is about the $x$ axis.  So, they cannot commute nor can they be placed in a different order.\label{QFTcircuit}}
\end{figure}

We can substitute $y$ for its binary representation as using the previously defined expressions in Eq.~\eqref{ybinary} to find \cite{nielsen2010quantum}
\begin{align}\label{qft_binary}
|x\rangle&=\frac1{2^{N/2}}\bigotimes_{\ell=0}^{N-1}\sum_{y_\ell=0}^1\exp\left(i\pi xy_\ell/2^\ell\right)|y_\ell\rangle
\end{align}
where the sums in Eq.~\eqref{qft} are rewritten for the binary representation as
\begin{equation}\label{binarysum}
\sum_{y=0}^{2^N-1}\equiv\bigotimes_{\ell=0}^{N-1}\sum_{y_\ell=0}^1=\sum_{y_0=0}^1\sum_{y_1=0}^1\ldots\sum_{y_{N-1}=0}^1
\end{equation}
for simplicity.  Noticing that each term is factorizable, this becomes
\begin{align}
|x\rangle=\frac1{2^{N/2}}\bigotimes_{\ell=0}^{N-1}\Bigg\{|0\rangle+\exp\left(i\pi x/2^\ell\right)|1\rangle\Bigg\}\label{qftbit}
\end{align}
where for a given $\ell$, the division by a power of 2 in the argument of the exponential will reveal one digit of the binary representation of the integer $x$.\cite{nielsen2010quantum}

The operators necessary to perform a QFT on a quantum computer can be identified in Eq.~\eqref{qftbit}.  Note that applying the Hadamard gate on a qubit written in an arbitrary binary form is
\begin{equation}
H|x_k\rangle=\left(|0\rangle+\exp(i\pi x_k)|1\rangle\right)/\sqrt2
\end{equation}
as can be verified by noting that $x_k$ is either 0 or 1.  So, a Hadamard will rotate each bit of $x$ into a basis that is only different from an individual term in Eq.~\eqref{qftbit} by a phase on the $|1\rangle$ state.  If we apply the phase rotation gate of the form
\begin{equation}\label{rotation}
R_\ell=\left(\begin{array}{cc}
1 & 0\\
0 & e^{i\pi/2^\ell}
\end{array}\right)
\end{equation}
after the Hadamard, then this will constitute the most basic operation in the QFT. That is, applying $H$ and then $R_\ell$ a certain number of times will generate terms required in Eq.~\eqref{qftbit}.  More than one $R_\ell$ gate may need to be applied depending on which bit of $x$ we are acting on.  The structure of gates is shown in Fig.~\ref{QFTcircuit}. The next qubit has all but the first register's rotation matrix applied.  We then continue to the next qubit, applying one less set of gates, and continue until we apply only the Hadamard on the last qubit.  If a swap operation is applied, the qubits will appear in order and this completes the QFT.

Note that the overall cost of this algorithm scales as $\mathcal{O}(N^2)$ but that a cheaper $\mathcal{O}(N\log N)$ is also available.\cite{hales2000improved}  Note also that since we cannot efficiently obtain all coefficients from the superposition, the algorithm does not provide a useful speedup over the classical algorithm which scales exponentially, not polynomially.  However, the time to measure all elements is exponentially long, so this quantum advantage is not truly advantageous.  Still, the QFT can be useful as a tool in other subroutines.

\subsection{Quantum phase estimation}\label{QPE}

Given an input state, $|\psi\rangle$, we want to determine the associated eigenvalue of the form $\exp(i2\pi\varphi)$ for some real $\varphi$ denoting a phase. For this algorithm, we must first have the initial state, $|\psi\rangle$, and a number of auxiliary qubits at least equal to the number of digits that the phase must be accurate to. 

\begin{figure}[b]
\includegraphics[width=\columnwidth]{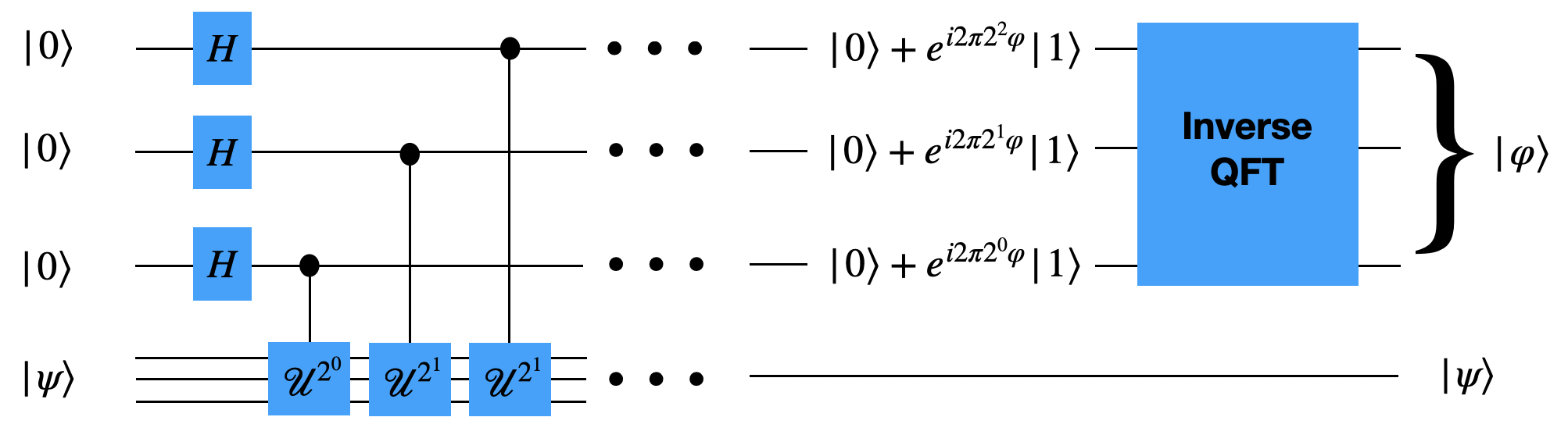}
\caption{Circuit diagram of the quantum phase estimation algorithm The input is an auxiliary register and the initial wavefunction.  The output is the phase corresponding to the eigenvalue and the original waevfunction.  Essentially, one obtains the associated energy for an input wavefunction.\label{QPEcircuit}}
\end{figure}

The strategy will be to generate, from an initial wavefunction $\psi$, the binary digits of the phase of $\varphi$--as defined in Appendix.~\ref{QFT}--and then perform an inverse QFT to obtain the phase on an auxiliary register.

Let a gate $\mathcal{U}[=R_1$ from Eq.~\eqref{rotation}] be one such that when applied to $\psi$, and controlled on one of the auxiliary qubits, it produces a phase that is the $j^\mathrm{th}$ value of the binary representation,
\begin{equation}\label{cU}
\mathcal{U}^{2^j}H|0\rangle_j|\psi\rangle=\frac1{\sqrt2}\left(|0\rangle+e^{i2\pi2^{t-j}\varphi}|1\rangle\right)_j|\psi\rangle
\end{equation}
where $|0\rangle_j$ is the $j^\mathrm{th}$ qubit and $t$ is the total number of qubits that this operator will be applied to.   

The gate $\mathcal{U}$ is applied as in Eq.~\eqref{cU} to a register of auxiliary qubits and to $\psi$ as in Fig.~\ref{QPEcircuit}.  The resulting state is then
\begin{equation}
\frac1{2^{\t{}/2}}\bigotimes_{j=1}^\t{}\left(|0\rangle+e^{2\pi i2^{t-j}\varphi}|1\rangle\right)=\frac1{2^{t/2}}\sum_{y=0}^{2^t-1}e^{i2\pi\varphi y}|y\rangle
\end{equation}
where the same conversion to and from a binary representation in Appendix.~\ref{QFT} is used. Once an inverse Fourier transform is applied on the last step, we obtain $|\varphi\rangle$ and have then represented the phase on a set of qubits.

In order to construct the unitary operator $\mathcal{U}$, we simply must obtain the representation of the exponentiated Hamiltonian as $\exp(-i\mathcal{H}t)$.\cite{whitfield2011simulation}  The energy of the wavefunction ($\varphi$ related to $E$) is then related to the time applied and may involve other pre-determined constants.\cite{whitfield2011simulation}

In order to apply the exponentiated Hamiltonian, one option is to use the Trotter-Suzuki decomposition\cite{suzuki1993improved} to decompose the exponential into a product of exponentials with fewer terms.\cite{poulin2014trotter}

By expanding the number of qubits used in the auxiliary register, we can increase the accuracy of the final result.  We will defer detailed analysis of this point to Ref.~\onlinecite{nielsen2010quantum}, but it is worth noting that some error in this algorithm can be reduced with more resources.  The total error can be reduced arbitrarily to 1-$\eta$ for some small number $\eta$.

Note that improvements can be applied to the algorithm to generate more methods of QPE.\cite{poulin2018quantum} A recent improvement known as qubitization can bring down the gate-count for the determining the phase.  We will defer to the discussion in Ref.~\onlinecite{low2019hamiltonian}.  

\subsection{Quantum gradient algorithm}\label{QGA}

Given an oracle for a function $f$, its gradients can be computed in one query of the oracle instead of $m+1$ classically for $m$ grid points.\cite{jordan2005fast}

We start with three multi-qubit registers.  One is used for the computation of $f$.  The other contains an equal superposition over all states (representing the infinitesimal directions that $f$ can be shifted for an eventual gradient).  The last receives a QFT (on initial register set to 1 while the others are set to zero). This final register will be used for a phase kickback.  The purpose of the phase kickback is to modify the phases of the QFT.  When the inverse QFT is applied, we obtain the gradient similar to how the phase was obtained for the QPE.

The steps of this algorithm are shown in Fig.~\ref{QGAcircuit}.  The Hadamard gate produces
\begin{equation}\label{hadamardaction}
H^{\otimes N}|0\rangle^{\otimes N}=\frac1{\sqrt{2^{N}}}\bigotimes_{\ell=1}^{N}\sum_{\delta_\ell=0}^1|\delta_\ell\rangle\equiv\frac1{\sqrt{2^{N}}}\sum_{\boldsymbol{\delta}}|\boldsymbol{\delta}\rangle
\end{equation}
where $N$ is the number of qubits contained in the second register.  

When calling the oracle query as controlled on the equal superposition in Eq.~\eqref{hadamardaction}, $f$ is evaluated on all arguments $\delta_\ell$ giving $f(\boldsymbol{c}+\boldsymbol{\delta})$, having perturbed an initial coordinate by $\boldsymbol{\delta}$ or some similar function.\cite{gilyen2019optimizing}  

\begin{figure}[b]
\includegraphics[width=\columnwidth]{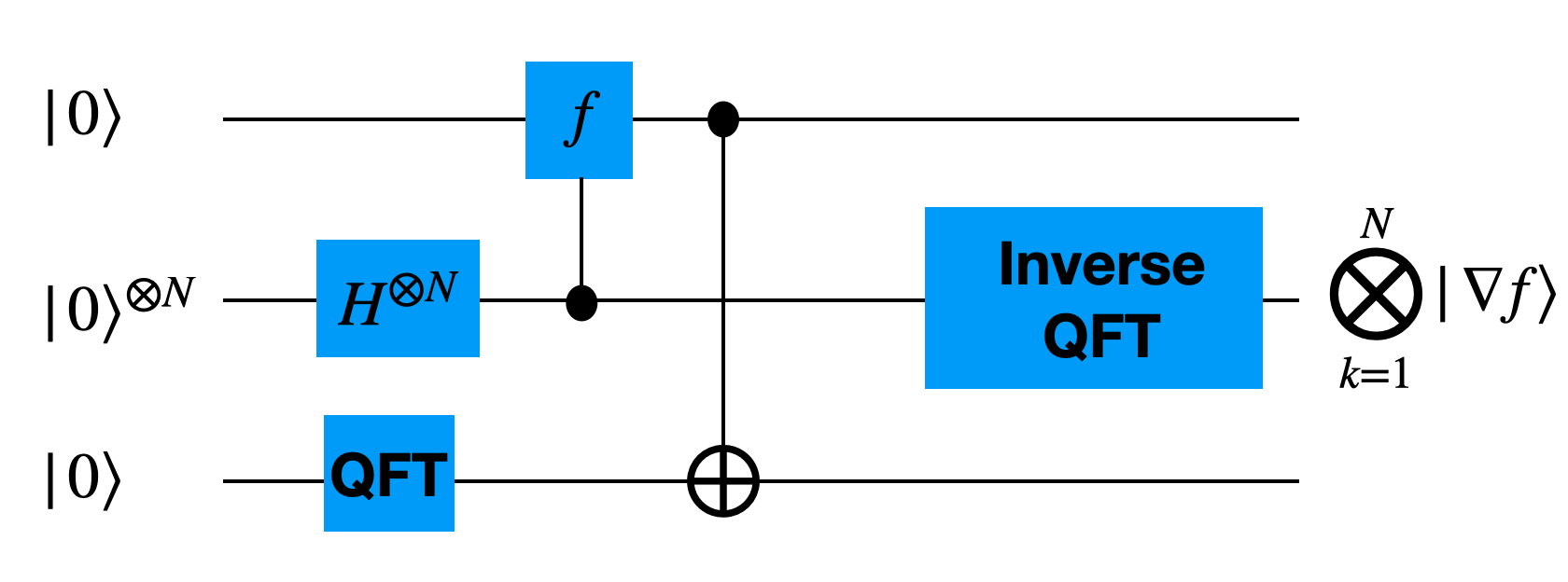}
\caption{Circuit diagram of the quantum gradient algorithm  One could write the function $f$ with a control to a fourth register with input $c$ for the point that $f$ is evaluated on.\label{QGAcircuit}}
\end{figure}

On the third register in Fig.~\ref{QGAcircuit}, a QFT has been applied on an initial value of 1, giving an output of 
\begin{equation}
\frac1{2^{\Np/2}}\sum_w e^{i2\pi w/2^{\Np}}|w\rangle
\end{equation}
which is similar to Eq.~\eqref{qft}. The number of qubits in this register, $\Np$, are enough to allow a pending bitwise addition to be carried out properly.

The state of the full quantum wavefunction is then
\begin{equation}
|\psi\rangle=\frac1{\sqrt{2^{N}2^{\Np}}}\sum_{\boldsymbol{\delta}}\sum_w e^{i2\pi w/2^{\Np}}|w\rangle|f(\boldsymbol{c}+\boldsymbol{\tilde\delta})\rangle|\boldsymbol{\delta}\rangle
\end{equation}
where
\begin{equation}
\boldsymbol{\tilde\delta}=L(\boldsymbol{\delta}-\mathbf{N}/2)/2^N
\end{equation}
and $\boldsymbol{N}$ is a vector of $(2^N,2^N,\ldots)$ and provides the offset factor to convert integers to real numbers. $M$ and $L$ are assigned meaning in the following. The factor $M$ scales the maximum amount of $\boldsymbol{\nabla} f$ to keep this quantity expressed as an integer.  The parameter $L$ is a neighborhood over which the derivative is accurate to first order, for example in one-dimension
\begin{equation}\label{derivative}
\partial_x f\approx\frac{f(x+L/2)-f(x-L/2)}L
\end{equation}
is the decomposition in one-dimension.

The next step is to add (bitwise addition denoted by $\oplus$) the first register to the third under the addition $w\rightarrow w\oplus(2^N2^{\Np} f)/(ML)\;\mathrm{mod}\;2^{\Np}$ so that the register and phase are shifted as
\begin{align}
|\psi\rangle&=\frac1{\sqrt{2^{N}2^{\Np}}}\sum_{\boldsymbol{\delta}}\sum_w e^{i2\pi \left(w+ \frac{2^N2^{\Np}}{ML}f(\boldsymbol{c}+\boldsymbol{\tilde\delta})\right)/2^{\Np}}\nonumber\\
&\hspace{4cm}\times|w\rangle|f(\boldsymbol{c}+\boldsymbol{\tilde\delta})\rangle|\boldsymbol{\delta}\rangle
\end{align}
where this trick is often called a phase kickback. In a small neighborhood, $L$, around the central point $\boldsymbol{c}$ of the oracle query the vectors $\boldsymbol{\delta}$ can be thought of as perturbations on this point.  Expanding the function according to a Taylor expansion gives
\begin{align}
\sum_{\boldsymbol{\delta}} e^{i2\pi\frac{2^N}{ML}\left(f(\boldsymbol{c})+\frac{\lrange{}}{2^N}\left(\boldsymbol{\delta}-\frac{\mathbf{N}}2\right)\cdot\boldsymbol{\nabla} f\right)}|f(\boldsymbol{c}+\boldsymbol{\tilde\delta})\rangle|\boldsymbol{\delta}\rangle\\
=e^{i2\pi\frac{2^N}{ML} f(\boldsymbol{c})}
e^{i\frac{2\pi}M\frac{\mathbf{N}}2\cdot\boldsymbol{\nabla} f}
\sum_{\boldsymbol{\delta}} e^{i\frac{2\pi}M\boldsymbol{\delta}\cdot\boldsymbol{\nabla} f}|f(\boldsymbol{c}+\boldsymbol{\tilde\delta})\rangle|\boldsymbol{\delta}\rangle\nonumber
\end{align}

Recall that, by inspection from Eq.~\eqref{qft}, applying the inverse Fourier transform will give
\begin{equation}
\bigotimes_k\left|\frac{2^N}{\m{}}(\nabla_{x_k}f)\Big|_c\right\rangle
\end{equation}
So, given a continuous function $f$, we obtain a gradient $\boldsymbol{\nabla} f$.  The representation on the qubits for both $f$ and $\boldsymbol{\nabla} f$ is in the binary representation of their continuous values. Improvements to this algorithm have been noted.\cite{gilyen2019optimizing}

\subsubsection{Functional derivatives with quantum gradient algorithms}\label{funcderivQGA}

The QGA can be used to evaluate the functional derivative, Eq.~\eqref{funcderivdef}.  The test field $\Upsilon(\x)$ can be provided by hand or by Hadamard transformation with each resulting state giving the same result.  Evaluating the derivative as before on $\eta$ produces the functional derivative. This could be applied to a variety of functionals, such as the density functional or the partition function (although this last object would be very difficult to compute before taking the functional derivative due again to the curse of dimensionality).\cite{novak2008tractability}  There are other ways to get the functional derivative, such as with a chain rule if an alternative form is required.

\subsection{Real time evolution}\label{rte}

While it is possible in theory to implement a more advanced classical algorithm on the quantum computer, there may be sizable overhead. It is generally accepted in the literature to use the real time evolution (RTE) method which we choose to introduce here.\cite{oh2008quantum}

The initial state of qubits can be initialized into a state consisting of a single-particle Hamiltonian's, $\mathcal{H}_0$, eigenstate.  This could be the Hartree-Fock solution for some number of electrons, $N_e$.\cite{whitfield2011simulation} The Hamiltonian is then given a time-dependence such that $t=0$ is $\mathcal{H}_0$ and $t=t_\mathrm{max.}$ is the full Hamiltonian,
\begin{equation}
\mathcal{H}(t)=\mathcal{H}_0+\lambda(t)\mathcal{H}_1+\mathcal{C}
\end{equation}
for some time-dependent function $\lambda(t)$ and interaction term $\mathcal{H}_1$.  By tuning the time parameter slowly enough, the new ground state can be found.  A constant $\mathcal{C}$ is added in anticipation of the QPE and is simply taken into account when converting the output phase to the energy.\cite{oh2008quantum}

The final state of the RTE must be a close approximation to the true ground state for QPE to work properly.\cite{nielsen2010quantum,kassal2011simulating}  In order to evolve the initial state to the ground state, the error in the Trotter step must be no more than the allowed accuracy for a computation.\cite{wecker2014gate}  For the case of molecular systems, this is 1 mHa (although this can be even lower for some applications).  A variable number of steps is required to fully evolve the initial wavefunction to the ground state, but this may be on the order of a number of thousands and the entire process can take months or much, much longer.\cite{poulin2014trotter}

There are also other algorithms that could be used,\cite{mcardle2019variational,hackl2008real,schwarz2012preparing,schollwock2011density,gilyen2019quantum,yanofsky2008quantum} but these may have a large overhead.  Just as with phase estimation, we present the most widely known algorithm here for ease of presentation.  The RWMP method of the main text can interchange subroutines for the best algorithm

\subsection{Quantum amplitude estimation (quantum counting)}\label{qcount}

A way to obtain useful quantities from a wavefunction without measuring it is to use QAE\footnote{This algorithm can go by other names, notably quantum counting, but we use `QAE' instead of the abbreviated `QC' for quantum counting since this would overlap with `quantum chemistry' and `quantum computing' if we chose to also use abbreviations for those.} as described by Ref.~\onlinecite{brassard2002quantum} (although we follow the state-preserving quantum counting algorithm used in Ref.~\onlinecite{temme2011quantum} and also note Refs.~\onlinecite{knill2007optimal,bakerGreen20}).  One can envision the application of an operator $\Op$ onto the wavefunction ({\it e.g.}, $\hat c^\dagger_i\hat c_j$ onto $|\Psi\rangle$) as being represented in a superposition of the original function $\Psi$ and all other states that are perpendicular, $\Psi^\perp$, as
\begin{equation}
\Op|\Psi\rangle=\alpha_0|\Psi\rangle+\alpha^\perp|\Psi^\perp\rangle.
\end{equation}
We will simply assume some process exists to apply the operator of interest is available. The goal is to estimate $\alpha_0$ which is the expectation value.  A series of steps is required to find this coefficient.  The fraction of times that the following is successful will give $\alpha_0$.  Before performing any of the subsequent steps, we assume that the energy of $\Psi$ has been obtained via phase estimation beforehand on a separate register.  In total, four registers are required: one for $\Psi$, one for the saved ground state energy, one for the check energy, and one for the pointer qubit.

One iteration of the full algorithm is the following:
\begin{enumerate}
\item An operator is applied to $\Psi$
\item The energy of the resulting state is determined; the energy is in a superposition over all states
\item The newly found energy is compared to the saved energy of the original wavefunction
\item The difference is represented as a single bit (called a pointer qubit)
\item (Accept) The pointer qubit is measured.  The algorithm then branches: if the measurement of the single pointer qubit gives a success, implying the energies match and the original $\Psi$ was recovered, we repeat the steps starting from step 1 here after returning the state to the original configuration.  A separate counter is incremented by one each time this step is reached.
\item (Reject) If the pointer measurement results in a failure, then the wavefunction found is not the original.  We must recover the original wavefunction by undoing the QPE, undoing the operator applied ($\Op^\dagger$), applying the operator again, and again finding the energy difference. We return to step 1.
\end{enumerate}
More details and diagrams can be found in Refs.~\onlinecite{brassard2002quantum,marriott2005quantum,temme2011quantum}.  Intuitively, we are counting the number of times the starting wavefunction is recovered when applying an operator.  The ratio of accepted counts to the total number of times the operator is applied is related to the coefficient on the ground state, $\alpha_0$.  While the description here involves a measurement and therefore gives classical data, the algorithm can be used as an oracle query for the QGA (as stated in Ref.~\onlinecite{temme2011quantum} at the cost of additional auxiliary qubits).

In order to see how the algorithm will converge when the reject step is activated, we can analogize with the half-life of radioactive isotopes to envision when the rejection procedure must eventually recover the correct ground state.  A detailed analysis of the convergence of the rejection step shows the number of steps required to find the original $\Psi$ is related to $1/\varepsilon$ for some probability of failure, $\varepsilon$ and is found in Refs.~\onlinecite{temme2011quantum,brassard2002quantum}.

\end{appendix}

\bibliography{MLQCDFT}

\begin{thebibliography}{230}%
\makeatletter
\providecommand \@ifxundefined [1]{%
 \@ifx{#1\undefined}
}%
\providecommand \@ifnum [1]{%
 \ifnum #1\expandafter \@firstoftwo
 \else \expandafter \@secondoftwo
 \fi
}%
\providecommand \@ifx [1]{%
 \ifx #1\expandafter \@firstoftwo
 \else \expandafter \@secondoftwo
 \fi
}%
\providecommand \natexlab [1]{#1}%
\providecommand \enquote  [1]{``#1''}%
\providecommand \bibnamefont  [1]{#1}%
\providecommand \bibfnamefont [1]{#1}%
\providecommand \citenamefont [1]{#1}%
\providecommand \href@noop [0]{\@secondoftwo}%
\providecommand \href [0]{\begingroup \@sanitize@url \@href}%
\providecommand \@href[1]{\@@startlink{#1}\@@href}%
\providecommand \@@href[1]{\endgroup#1\@@endlink}%
\providecommand \@sanitize@url [0]{\catcode `\\12\catcode `\$12\catcode
  `\&12\catcode `\#12\catcode `\^12\catcode `\_12\catcode `\%12\relax}%
\providecommand \@@startlink[1]{}%
\providecommand \@@endlink[0]{}%
\providecommand \url  [0]{\begingroup\@sanitize@url \@url }%
\providecommand \@url [1]{\endgroup\@href {#1}{\urlprefix }}%
\providecommand \urlprefix  [0]{URL }%
\providecommand \Eprint [0]{\href }%
\providecommand \doibase [0]{http://dx.doi.org/}%
\providecommand \selectlanguage [0]{\@gobble}%
\providecommand \bibinfo  [0]{\@secondoftwo}%
\providecommand \bibfield  [0]{\@secondoftwo}%
\providecommand \translation [1]{[#1]}%
\providecommand \BibitemOpen [0]{}%
\providecommand \bibitemStop [0]{}%
\providecommand \bibitemNoStop [0]{.\EOS\space}%
\providecommand \EOS [0]{\spacefactor3000\relax}%
\providecommand \BibitemShut  [1]{\csname bibitem#1\endcsname}%
\let\auto@bib@innerbib\@empty
\bibitem [{\citenamefont {Nielsen}\ and\ \citenamefont
  {Chuang}(2010)}]{nielsen2010quantum}%
  \BibitemOpen
  \bibfield  {author} {\bibinfo {author} {\bibfnamefont {Michael~A}\
  \bibnamefont {Nielsen}}\ and\ \bibinfo {author} {\bibfnamefont {Isaac~L}\
  \bibnamefont {Chuang}},\ }\href@noop {} {\emph {\bibinfo {title} {Quantum
  Computation and Quantum Information}}}\ (\bibinfo  {publisher} {Cambridge
  University Press},\ \bibinfo {year} {2010})\BibitemShut {NoStop}%
\bibitem [{\citenamefont {Deutsch}\ and\ \citenamefont
  {Jozsa}(1992)}]{deutsch1992rapid}%
  \BibitemOpen
  \bibfield  {author} {\bibinfo {author} {\bibfnamefont {David}\ \bibnamefont
  {Deutsch}}\ and\ \bibinfo {author} {\bibfnamefont {Richard}\ \bibnamefont
  {Jozsa}},\ }\bibfield  {title} {\enquote {\bibinfo {title} {Rapid solution of
  problems by quantum computation},}\ }\href@noop {} {\bibfield  {journal}
  {\bibinfo  {journal} {Proceedings of the Royal Society of London. Series A:
  Mathematical and Physical Sciences}\ }\textbf {\bibinfo {volume} {439}},\
  \bibinfo {pages} {553--558} (\bibinfo {year} {1992})}\BibitemShut {NoStop}%
\bibitem [{\citenamefont {Shor}(1994)}]{shor1994proceedings}%
  \BibitemOpen
  \bibfield  {author} {\bibinfo {author} {\bibfnamefont {Peter~W}\ \bibnamefont
  {Shor}},\ }\bibfield  {title} {\enquote {\bibinfo {title} {Proceedings of the
  35th annual symposium on foundations of computer science},}\ }\href@noop {}
  {\bibfield  {journal} {\bibinfo  {journal} {IEE Computer society press, Santa
  Fe, NM}\ } (\bibinfo {year} {1994})}\BibitemShut {NoStop}%
\bibitem [{\citenamefont {Grover}(1997)}]{grover1997quantum}%
  \BibitemOpen
  \bibfield  {author} {\bibinfo {author} {\bibfnamefont {Lov~K}\ \bibnamefont
  {Grover}},\ }\bibfield  {title} {\enquote {\bibinfo {title} {Quantum
  mechanics helps in searching for a needle in a haystack},}\ }\href@noop {}
  {\bibfield  {journal} {\bibinfo  {journal} {Phys.~Rev.~Lett.}\ }\textbf
  {\bibinfo {volume} {79}},\ \bibinfo {pages} {325} (\bibinfo {year}
  {1997})}\BibitemShut {NoStop}%
\bibitem [{\citenamefont {Brassard}\ \emph {et~al.}(2002)\citenamefont
  {Brassard}, \citenamefont {Hoyer}, \citenamefont {Mosca},\ and\ \citenamefont
  {Tapp}}]{brassard2002quantum}%
  \BibitemOpen
  \bibfield  {author} {\bibinfo {author} {\bibfnamefont {Gilles}\ \bibnamefont
  {Brassard}}, \bibinfo {author} {\bibfnamefont {Peter}\ \bibnamefont {Hoyer}},
  \bibinfo {author} {\bibfnamefont {Michele}\ \bibnamefont {Mosca}}, \ and\
  \bibinfo {author} {\bibfnamefont {Alain}\ \bibnamefont {Tapp}},\ }\bibfield
  {title} {\enquote {\bibinfo {title} {Quantum amplitude amplification and
  estimation},}\ }\href@noop {} {\bibfield  {journal} {\bibinfo  {journal}
  {Contemporary Mathematics}\ }\textbf {\bibinfo {volume} {305}},\ \bibinfo
  {pages} {53--74} (\bibinfo {year} {2002})}\BibitemShut {NoStop}%
\bibitem [{\citenamefont {Shor}(1999)}]{shor1999polynomial}%
  \BibitemOpen
  \bibfield  {author} {\bibinfo {author} {\bibfnamefont {Peter~W}\ \bibnamefont
  {Shor}},\ }\bibfield  {title} {\enquote {\bibinfo {title} {Polynomial-time
  algorithms for prime factorization and discrete logarithms on a quantum
  computer},}\ }\href@noop {} {\bibfield  {journal} {\bibinfo  {journal} {SIAM
  review}\ }\textbf {\bibinfo {volume} {41}},\ \bibinfo {pages} {303--332}
  (\bibinfo {year} {1999})}\BibitemShut {NoStop}%
\bibitem [{\citenamefont {Grover}(2001)}]{grover2001schrodinger}%
  \BibitemOpen
  \bibfield  {author} {\bibinfo {author} {\bibfnamefont {Lov~K}\ \bibnamefont
  {Grover}},\ }\bibfield  {title} {\enquote {\bibinfo {title} {{From
  Schr{\"o}dinger's equation to the quantum search algorithm}},}\ }\href@noop
  {} {\bibfield  {journal} {\bibinfo  {journal} {Pramana}\ }\textbf {\bibinfo
  {volume} {56}},\ \bibinfo {pages} {333--348} (\bibinfo {year}
  {2001})}\BibitemShut {NoStop}%
\bibitem [{\citenamefont {Aspuru-Guzik}\ \emph {et~al.}(2005)\citenamefont
  {Aspuru-Guzik}, \citenamefont {Dutoi}, \citenamefont {Love},\ and\
  \citenamefont {Head-Gordon}}]{aspuru2005simulated}%
  \BibitemOpen
  \bibfield  {author} {\bibinfo {author} {\bibfnamefont {Al{\'a}n}\
  \bibnamefont {Aspuru-Guzik}}, \bibinfo {author} {\bibfnamefont {Anthony~D}\
  \bibnamefont {Dutoi}}, \bibinfo {author} {\bibfnamefont {Peter~J}\
  \bibnamefont {Love}}, \ and\ \bibinfo {author} {\bibfnamefont {Martin}\
  \bibnamefont {Head-Gordon}},\ }\bibfield  {title} {\enquote {\bibinfo {title}
  {Simulated quantum computation of molecular energies},}\ }\href@noop {}
  {\bibfield  {journal} {\bibinfo  {journal} {Science}\ }\textbf {\bibinfo
  {volume} {309}},\ \bibinfo {pages} {1704--1707} (\bibinfo {year}
  {2005})}\BibitemShut {NoStop}%
\bibitem [{\citenamefont {Brown}\ \emph {et~al.}(2010)\citenamefont {Brown},
  \citenamefont {Munro},\ and\ \citenamefont {Kendon}}]{brown2010using}%
  \BibitemOpen
  \bibfield  {author} {\bibinfo {author} {\bibfnamefont {Katherine~L}\
  \bibnamefont {Brown}}, \bibinfo {author} {\bibfnamefont {William~J}\
  \bibnamefont {Munro}}, \ and\ \bibinfo {author} {\bibfnamefont {Vivien~M}\
  \bibnamefont {Kendon}},\ }\bibfield  {title} {\enquote {\bibinfo {title}
  {Using quantum computers for quantum simulation},}\ }\href@noop {} {\bibfield
   {journal} {\bibinfo  {journal} {Entropy}\ }\textbf {\bibinfo {volume}
  {12}},\ \bibinfo {pages} {2268--2307} (\bibinfo {year} {2010})}\BibitemShut
  {NoStop}%
\bibitem [{\citenamefont {Lanyon}\ \emph {et~al.}(2010)\citenamefont {Lanyon},
  \citenamefont {Whitfield}, \citenamefont {Gillett}, \citenamefont {Goggin},
  \citenamefont {Almeida}, \citenamefont {Kassal}, \citenamefont {Biamonte},
  \citenamefont {Mohseni}, \citenamefont {Powell}, \citenamefont {Barbieri}
  \emph {et~al.}}]{lanyon2010towards}%
  \BibitemOpen
  \bibfield  {author} {\bibinfo {author} {\bibfnamefont {Benjamin~P}\
  \bibnamefont {Lanyon}}, \bibinfo {author} {\bibfnamefont {James~D}\
  \bibnamefont {Whitfield}}, \bibinfo {author} {\bibfnamefont {Geoff~G}\
  \bibnamefont {Gillett}}, \bibinfo {author} {\bibfnamefont {Michael~E}\
  \bibnamefont {Goggin}}, \bibinfo {author} {\bibfnamefont {Marcelo~P}\
  \bibnamefont {Almeida}}, \bibinfo {author} {\bibfnamefont {Ivan}\
  \bibnamefont {Kassal}}, \bibinfo {author} {\bibfnamefont {Jacob~D}\
  \bibnamefont {Biamonte}}, \bibinfo {author} {\bibfnamefont {Masoud}\
  \bibnamefont {Mohseni}}, \bibinfo {author} {\bibfnamefont {Ben~J}\
  \bibnamefont {Powell}}, \bibinfo {author} {\bibfnamefont {Marco}\
  \bibnamefont {Barbieri}},  \emph {et~al.},\ }\bibfield  {title} {\enquote
  {\bibinfo {title} {Towards quantum chemistry on a quantum computer},}\
  }\href@noop {} {\bibfield  {journal} {\bibinfo  {journal} {Nature Chemistry}\
  }\textbf {\bibinfo {volume} {2}},\ \bibinfo {pages} {106--111} (\bibinfo
  {year} {2010})}\BibitemShut {NoStop}%
\bibitem [{\citenamefont {Whitfield}\ \emph {et~al.}(2011)\citenamefont
  {Whitfield}, \citenamefont {Biamonte},\ and\ \citenamefont
  {Aspuru-Guzik}}]{whitfield2011simulation}%
  \BibitemOpen
  \bibfield  {author} {\bibinfo {author} {\bibfnamefont {James~D}\ \bibnamefont
  {Whitfield}}, \bibinfo {author} {\bibfnamefont {Jacob}\ \bibnamefont
  {Biamonte}}, \ and\ \bibinfo {author} {\bibfnamefont {Al{\'a}n}\ \bibnamefont
  {Aspuru-Guzik}},\ }\bibfield  {title} {\enquote {\bibinfo {title} {Simulation
  of electronic structure hamiltonians using quantum computers},}\ }\href@noop
  {} {\bibfield  {journal} {\bibinfo  {journal} {Mol.~Phys.}\ }\textbf
  {\bibinfo {volume} {109}},\ \bibinfo {pages} {735--750} (\bibinfo {year}
  {2011})}\BibitemShut {NoStop}%
\bibitem [{\citenamefont {Cao}\ \emph {et~al.}(2019)\citenamefont {Cao},
  \citenamefont {Romero}, \citenamefont {Olson}, \citenamefont {Degroote},
  \citenamefont {Johnson}, \citenamefont {Kieferov{\'a}}, \citenamefont
  {Kivlichan}, \citenamefont {Menke}, \citenamefont {Peropadre}, \citenamefont
  {Sawaya} \emph {et~al.}}]{cao2019quantum}%
  \BibitemOpen
  \bibfield  {author} {\bibinfo {author} {\bibfnamefont {Yudong}\ \bibnamefont
  {Cao}}, \bibinfo {author} {\bibfnamefont {Jonathan}\ \bibnamefont {Romero}},
  \bibinfo {author} {\bibfnamefont {Jonathan~P}\ \bibnamefont {Olson}},
  \bibinfo {author} {\bibfnamefont {Matthias}\ \bibnamefont {Degroote}},
  \bibinfo {author} {\bibfnamefont {Peter~D}\ \bibnamefont {Johnson}}, \bibinfo
  {author} {\bibfnamefont {M{\'a}ria}\ \bibnamefont {Kieferov{\'a}}}, \bibinfo
  {author} {\bibfnamefont {Ian~D}\ \bibnamefont {Kivlichan}}, \bibinfo {author}
  {\bibfnamefont {Tim}\ \bibnamefont {Menke}}, \bibinfo {author} {\bibfnamefont
  {Borja}\ \bibnamefont {Peropadre}}, \bibinfo {author} {\bibfnamefont
  {Nicolas~PD}\ \bibnamefont {Sawaya}},  \emph {et~al.},\ }\bibfield  {title}
  {\enquote {\bibinfo {title} {Quantum chemistry in the age of quantum
  computing},}\ }\href@noop {} {\bibfield  {journal} {\bibinfo  {journal}
  {Chemical reviews}\ }\textbf {\bibinfo {volume} {119}},\ \bibinfo {pages}
  {10856--10915} (\bibinfo {year} {2019})}\BibitemShut {NoStop}%
\bibitem [{\citenamefont {McArdle}\ \emph {et~al.}(2020)\citenamefont
  {McArdle}, \citenamefont {Endo}, \citenamefont {Aspuru-Guzik}, \citenamefont
  {Benjamin},\ and\ \citenamefont {Yuan}}]{mcardle2020quantum}%
  \BibitemOpen
  \bibfield  {author} {\bibinfo {author} {\bibfnamefont {Sam}\ \bibnamefont
  {McArdle}}, \bibinfo {author} {\bibfnamefont {Suguru}\ \bibnamefont {Endo}},
  \bibinfo {author} {\bibfnamefont {Alan}\ \bibnamefont {Aspuru-Guzik}},
  \bibinfo {author} {\bibfnamefont {Simon~C}\ \bibnamefont {Benjamin}}, \ and\
  \bibinfo {author} {\bibfnamefont {Xiao}\ \bibnamefont {Yuan}},\ }\bibfield
  {title} {\enquote {\bibinfo {title} {Quantum computational chemistry},}\
  }\href@noop {} {\bibfield  {journal} {\bibinfo  {journal} {Reviews of Modern
  Physics}\ }\textbf {\bibinfo {volume} {92}},\ \bibinfo {pages} {015003}
  (\bibinfo {year} {2020})}\BibitemShut {NoStop}%
\bibitem [{\citenamefont {Kang}\ \emph {et~al.}(2006)\citenamefont {Kang},
  \citenamefont {Meng}, \citenamefont {Br{\'e}ger}, \citenamefont {Grey},\ and\
  \citenamefont {Ceder}}]{kang2006electrodes}%
  \BibitemOpen
  \bibfield  {author} {\bibinfo {author} {\bibfnamefont {Kisuk}\ \bibnamefont
  {Kang}}, \bibinfo {author} {\bibfnamefont {Ying~Shirley}\ \bibnamefont
  {Meng}}, \bibinfo {author} {\bibfnamefont {Julien}\ \bibnamefont
  {Br{\'e}ger}}, \bibinfo {author} {\bibfnamefont {Clare~P}\ \bibnamefont
  {Grey}}, \ and\ \bibinfo {author} {\bibfnamefont {Gerbrand}\ \bibnamefont
  {Ceder}},\ }\bibfield  {title} {\enquote {\bibinfo {title} {Electrodes with
  high power and high capacity for rechargeable lithium batteries},}\
  }\href@noop {} {\bibfield  {journal} {\bibinfo  {journal} {Science}\ }\textbf
  {\bibinfo {volume} {311}},\ \bibinfo {pages} {977--980} (\bibinfo {year}
  {2006})}\BibitemShut {NoStop}%
\bibitem [{\citenamefont {Mavros}\ \emph {et~al.}(2014)\citenamefont {Mavros},
  \citenamefont {Tsuchimochi}, \citenamefont {Kowalczyk}, \citenamefont
  {McIsaac}, \citenamefont {Wang},\ and\ \citenamefont
  {Voorhis}}]{mavros2014can}%
  \BibitemOpen
  \bibfield  {author} {\bibinfo {author} {\bibfnamefont {Michael~G}\
  \bibnamefont {Mavros}}, \bibinfo {author} {\bibfnamefont {Takashi}\
  \bibnamefont {Tsuchimochi}}, \bibinfo {author} {\bibfnamefont {Tim}\
  \bibnamefont {Kowalczyk}}, \bibinfo {author} {\bibfnamefont {Alexandra}\
  \bibnamefont {McIsaac}}, \bibinfo {author} {\bibfnamefont {Lee-Ping}\
  \bibnamefont {Wang}}, \ and\ \bibinfo {author} {\bibfnamefont {Troy~Van}\
  \bibnamefont {Voorhis}},\ }\bibfield  {title} {\enquote {\bibinfo {title}
  {What can density functional theory tell us about artificial catalytic water
  splitting?}}\ }\href@noop {} {\bibfield  {journal} {\bibinfo  {journal}
  {Inorganic chemistry}\ }\textbf {\bibinfo {volume} {53}},\ \bibinfo {pages}
  {6386--6397} (\bibinfo {year} {2014})}\BibitemShut {NoStop}%
\bibitem [{\citenamefont {Cudazzo}\ \emph {et~al.}(2008)\citenamefont
  {Cudazzo}, \citenamefont {Profeta}, \citenamefont {Sanna}, \citenamefont
  {Floris}, \citenamefont {Continenza}, \citenamefont {Massidda},\ and\
  \citenamefont {Gross}}]{cudazzo2008ab}%
  \BibitemOpen
  \bibfield  {author} {\bibinfo {author} {\bibfnamefont {P}~\bibnamefont
  {Cudazzo}}, \bibinfo {author} {\bibfnamefont {Gianni}\ \bibnamefont
  {Profeta}}, \bibinfo {author} {\bibfnamefont {Antonio}\ \bibnamefont
  {Sanna}}, \bibinfo {author} {\bibfnamefont {A}~\bibnamefont {Floris}},
  \bibinfo {author} {\bibfnamefont {A}~\bibnamefont {Continenza}}, \bibinfo
  {author} {\bibfnamefont {S}~\bibnamefont {Massidda}}, \ and\ \bibinfo
  {author} {\bibfnamefont {EKU}\ \bibnamefont {Gross}},\ }\bibfield  {title}
  {\enquote {\bibinfo {title} {Ab initio description of high-temperature
  superconductivity in dense molecular hydrogen},}\ }\href@noop {} {\bibfield
  {journal} {\bibinfo  {journal} {Phys.~Rev.~Lett.}\ }\textbf {\bibinfo
  {volume} {100}},\ \bibinfo {pages} {257001} (\bibinfo {year}
  {2008})}\BibitemShut {NoStop}%
\bibitem [{\citenamefont {Rod}\ \emph {et~al.}(2000)\citenamefont {Rod},
  \citenamefont {Logadottir},\ and\ \citenamefont
  {N{\o}rskov}}]{rod2000ammonia}%
  \BibitemOpen
  \bibfield  {author} {\bibinfo {author} {\bibfnamefont {Thomas~Holm}\
  \bibnamefont {Rod}}, \bibinfo {author} {\bibfnamefont {Ashildur}\
  \bibnamefont {Logadottir}}, \ and\ \bibinfo {author} {\bibfnamefont
  {Jens~Kehlet}\ \bibnamefont {N{\o}rskov}},\ }\bibfield  {title} {\enquote
  {\bibinfo {title} {Ammonia synthesis at low temperatures},}\ }\href@noop {}
  {\bibfield  {journal} {\bibinfo  {journal} {J.~Chem.~Phys.}\ }\textbf
  {\bibinfo {volume} {112}},\ \bibinfo {pages} {5343--5347} (\bibinfo {year}
  {2000})}\BibitemShut {NoStop}%
\bibitem [{\citenamefont {N{\o}rskov}\ \emph {et~al.}(2006)\citenamefont
  {N{\o}rskov}, \citenamefont {Scheffler},\ and\ \citenamefont
  {Toulhoat}}]{norskov2006density}%
  \BibitemOpen
  \bibfield  {author} {\bibinfo {author} {\bibfnamefont {Jens~Kehlet}\
  \bibnamefont {N{\o}rskov}}, \bibinfo {author} {\bibfnamefont {Matthias}\
  \bibnamefont {Scheffler}}, \ and\ \bibinfo {author} {\bibfnamefont
  {Herv{\'e}}\ \bibnamefont {Toulhoat}},\ }\bibfield  {title} {\enquote
  {\bibinfo {title} {Density functional theory in surface science and
  heterogeneous catalysis},}\ }\href@noop {} {\bibfield  {journal} {\bibinfo
  {journal} {MRS Bulletin}\ }\textbf {\bibinfo {volume} {31}},\ \bibinfo
  {pages} {669--674} (\bibinfo {year} {2006})}\BibitemShut {NoStop}%
\bibitem [{\citenamefont {Flores-Livas}\ \emph {et~al.}(2016)\citenamefont
  {Flores-Livas}, \citenamefont {Sanna},\ and\ \citenamefont
  {Gross}}]{flores2016high}%
  \BibitemOpen
  \bibfield  {author} {\bibinfo {author} {\bibfnamefont {Jos{\'e}~A}\
  \bibnamefont {Flores-Livas}}, \bibinfo {author} {\bibfnamefont {Antonio}\
  \bibnamefont {Sanna}}, \ and\ \bibinfo {author} {\bibfnamefont {EKU}\
  \bibnamefont {Gross}},\ }\bibfield  {title} {\enquote {\bibinfo {title} {High
  temperature superconductivity in sulfur and selenium hydrides at high
  pressure},}\ }\href@noop {} {\bibfield  {journal} {\bibinfo  {journal} {The
  European Physical Journal B}\ }\textbf {\bibinfo {volume} {89}},\ \bibinfo
  {pages} {63} (\bibinfo {year} {2016})}\BibitemShut {NoStop}%
\bibitem [{\citenamefont {Rydberg}\ \emph {et~al.}(2014)\citenamefont
  {Rydberg}, \citenamefont {J{\o}rgensen},\ and\ \citenamefont
  {Olsen}}]{rydberg2014use}%
  \BibitemOpen
  \bibfield  {author} {\bibinfo {author} {\bibfnamefont {Patrik}\ \bibnamefont
  {Rydberg}}, \bibinfo {author} {\bibfnamefont {Flemming~Steen}\ \bibnamefont
  {J{\o}rgensen}}, \ and\ \bibinfo {author} {\bibfnamefont {Lars}\ \bibnamefont
  {Olsen}},\ }\bibfield  {title} {\enquote {\bibinfo {title} {Use of density
  functional theory in drug metabolism studies},}\ }\href@noop {} {\bibfield
  {journal} {\bibinfo  {journal} {Expert opinion on drug metabolism \&
  toxicology}\ }\textbf {\bibinfo {volume} {10}},\ \bibinfo {pages} {215--227}
  (\bibinfo {year} {2014})}\BibitemShut {NoStop}%
\bibitem [{\citenamefont {Ashcroft}\ and\ \citenamefont
  {Mermin}(1976)}]{ashcroft1976introduction}%
  \BibitemOpen
  \bibfield  {author} {\bibinfo {author} {\bibfnamefont {N.W.}\ \bibnamefont
  {Ashcroft}}\ and\ \bibinfo {author} {\bibfnamefont {N.D.}\ \bibnamefont
  {Mermin}},\ }\href@noop {} {\emph {\bibinfo {title} {{Solid State
  Physics}}}}\ (\bibinfo  {publisher} {Saunders College},\ \bibinfo {address}
  {Philadelphia},\ \bibinfo {year} {1976})\BibitemShut {NoStop}%
\bibitem [{\citenamefont {Jansen}\ \emph {et~al.}(2007)\citenamefont {Jansen},
  \citenamefont {Ruskai},\ and\ \citenamefont {Seiler}}]{jansen2007bounds}%
  \BibitemOpen
  \bibfield  {author} {\bibinfo {author} {\bibfnamefont {Sabine}\ \bibnamefont
  {Jansen}}, \bibinfo {author} {\bibfnamefont {Mary-Beth}\ \bibnamefont
  {Ruskai}}, \ and\ \bibinfo {author} {\bibfnamefont {Ruedi}\ \bibnamefont
  {Seiler}},\ }\bibfield  {title} {\enquote {\bibinfo {title} {Bounds for the
  adiabatic approximation with applications to quantum computation},}\
  }\href@noop {} {\bibfield  {journal} {\bibinfo  {journal} {J.~Math.~Phys.}\
  }\textbf {\bibinfo {volume} {48}},\ \bibinfo {pages} {102111} (\bibinfo
  {year} {2007})}\BibitemShut {NoStop}%
\bibitem [{\citenamefont {Wecker}\ \emph {et~al.}(2014)\citenamefont {Wecker},
  \citenamefont {Bauer}, \citenamefont {Clark}, \citenamefont {Hastings},\ and\
  \citenamefont {Troyer}}]{wecker2014gate}%
  \BibitemOpen
  \bibfield  {author} {\bibinfo {author} {\bibfnamefont {Dave}\ \bibnamefont
  {Wecker}}, \bibinfo {author} {\bibfnamefont {Bela}\ \bibnamefont {Bauer}},
  \bibinfo {author} {\bibfnamefont {Bryan~K}\ \bibnamefont {Clark}}, \bibinfo
  {author} {\bibfnamefont {Matthew~B}\ \bibnamefont {Hastings}}, \ and\
  \bibinfo {author} {\bibfnamefont {Matthias}\ \bibnamefont {Troyer}},\
  }\bibfield  {title} {\enquote {\bibinfo {title} {Gate-count estimates for
  performing quantum chemistry on small quantum computers},}\ }\href@noop {}
  {\bibfield  {journal} {\bibinfo  {journal} {Phys.~Rev.~A}\ }\textbf {\bibinfo
  {volume} {90}},\ \bibinfo {pages} {022305} (\bibinfo {year}
  {2014})}\BibitemShut {NoStop}%
\bibitem [{\citenamefont {Poulin}\ \emph {et~al.}(2015)\citenamefont {Poulin},
  \citenamefont {Hastings}, \citenamefont {Wecker}, \citenamefont {Wiebe},
  \citenamefont {Doherty},\ and\ \citenamefont {Troyer}}]{poulin2014trotter}%
  \BibitemOpen
  \bibfield  {author} {\bibinfo {author} {\bibfnamefont {David}\ \bibnamefont
  {Poulin}}, \bibinfo {author} {\bibfnamefont {Matthew~B}\ \bibnamefont
  {Hastings}}, \bibinfo {author} {\bibfnamefont {Dave}\ \bibnamefont {Wecker}},
  \bibinfo {author} {\bibfnamefont {Nathan}\ \bibnamefont {Wiebe}}, \bibinfo
  {author} {\bibfnamefont {Andrew~C}\ \bibnamefont {Doherty}}, \ and\ \bibinfo
  {author} {\bibfnamefont {Matthias}\ \bibnamefont {Troyer}},\ }\bibfield
  {title} {\enquote {\bibinfo {title} {The trotter step size required for
  accurate quantum simulation of quantum chemistry},}\ }\href@noop {}
  {\bibfield  {journal} {\bibinfo  {journal} {Quantum Information and
  Computation}\ }\textbf {\bibinfo {volume} {15}},\ \bibinfo {pages}
  {0361--0384} (\bibinfo {year} {2015})}\BibitemShut {NoStop}%
\bibitem [{\citenamefont {Lemieux}\ \emph
  {et~al.}(2020{\natexlab{a}})\citenamefont {Lemieux}, \citenamefont
  {Duclos-Cianci}, \citenamefont {S{\'e}n{\'e}chal},\ and\ \citenamefont
  {Poulin}}]{lemieux2020resource}%
  \BibitemOpen
  \bibfield  {author} {\bibinfo {author} {\bibfnamefont {Jessica}\ \bibnamefont
  {Lemieux}}, \bibinfo {author} {\bibfnamefont {Guillaume}\ \bibnamefont
  {Duclos-Cianci}}, \bibinfo {author} {\bibfnamefont {David}\ \bibnamefont
  {S{\'e}n{\'e}chal}}, \ and\ \bibinfo {author} {\bibfnamefont {David}\
  \bibnamefont {Poulin}},\ }\bibfield  {title} {\enquote {\bibinfo {title}
  {Resource estimate for quantum many-body ground state preparation on a
  quantum computer},}\ }\href@noop {} {\bibfield  {journal} {\bibinfo
  {journal} {arXiv preprint arXiv:2006.04650}\ } (\bibinfo {year}
  {2020}{\natexlab{a}})}\BibitemShut {NoStop}%
\bibitem [{\citenamefont {Park}(1970)}]{park1970concept}%
  \BibitemOpen
  \bibfield  {author} {\bibinfo {author} {\bibfnamefont {James~L}\ \bibnamefont
  {Park}},\ }\bibfield  {title} {\enquote {\bibinfo {title} {The concept of
  transition in quantum mechanics},}\ }\href@noop {} {\bibfield  {journal}
  {\bibinfo  {journal} {Foundations of Physics}\ }\textbf {\bibinfo {volume}
  {1}},\ \bibinfo {pages} {23--33} (\bibinfo {year} {1970})}\BibitemShut
  {NoStop}%
\bibitem [{\citenamefont {Aaronson}(2018)}]{aaronson2018shadow}%
  \BibitemOpen
  \bibfield  {author} {\bibinfo {author} {\bibfnamefont {Scott}\ \bibnamefont
  {Aaronson}},\ }\bibfield  {title} {\enquote {\bibinfo {title} {Shadow
  tomography of quantum states},}\ }in\ \href@noop {} {\emph {\bibinfo
  {booktitle} {Proceedings of the 50th Annual ACM SIGACT Symposium on Theory of
  Computing}}}\ (\bibinfo {year} {2018})\ pp.\ \bibinfo {pages}
  {325--338}\BibitemShut {NoStop}%
\bibitem [{\citenamefont {Carleo}\ \emph {et~al.}(2019)\citenamefont {Carleo},
  \citenamefont {Cirac}, \citenamefont {Cranmer}, \citenamefont {Daudet},
  \citenamefont {Schuld}, \citenamefont {Tishby}, \citenamefont
  {Vogt-Maranto},\ and\ \citenamefont {Zdeborov{\'a}}}]{carleo2019machine}%
  \BibitemOpen
  \bibfield  {author} {\bibinfo {author} {\bibfnamefont {Giuseppe}\
  \bibnamefont {Carleo}}, \bibinfo {author} {\bibfnamefont {Ignacio}\
  \bibnamefont {Cirac}}, \bibinfo {author} {\bibfnamefont {Kyle}\ \bibnamefont
  {Cranmer}}, \bibinfo {author} {\bibfnamefont {Laurent}\ \bibnamefont
  {Daudet}}, \bibinfo {author} {\bibfnamefont {Maria}\ \bibnamefont {Schuld}},
  \bibinfo {author} {\bibfnamefont {Naftali}\ \bibnamefont {Tishby}}, \bibinfo
  {author} {\bibfnamefont {Leslie}\ \bibnamefont {Vogt-Maranto}}, \ and\
  \bibinfo {author} {\bibfnamefont {Lenka}\ \bibnamefont {Zdeborov{\'a}}},\
  }\bibfield  {title} {\enquote {\bibinfo {title} {Machine learning and the
  physical sciences},}\ }\href@noop {} {\bibfield  {journal} {\bibinfo
  {journal} {Rev.~Mod.~Phys.}\ }\textbf {\bibinfo {volume} {91}},\ \bibinfo
  {pages} {045002} (\bibinfo {year} {2019})}\BibitemShut {NoStop}%
\bibitem [{\citenamefont {Hohenberg}\ and\ \citenamefont
  {Kohn}(1964)}]{hohenberg1964inhomogeneous}%
  \BibitemOpen
  \bibfield  {author} {\bibinfo {author} {\bibfnamefont {Pierre}\ \bibnamefont
  {Hohenberg}}\ and\ \bibinfo {author} {\bibfnamefont {Walter}\ \bibnamefont
  {Kohn}},\ }\bibfield  {title} {\enquote {\bibinfo {title} {Inhomogeneous
  electron gas},}\ }\href@noop {} {\bibfield  {journal} {\bibinfo  {journal}
  {Physical Review}\ }\textbf {\bibinfo {volume} {136}},\ \bibinfo {pages}
  {B864} (\bibinfo {year} {1964})}\BibitemShut {NoStop}%
\bibitem [{\citenamefont {Kohn}\ and\ \citenamefont
  {Sham}(1965)}]{kohn1965self}%
  \BibitemOpen
  \bibfield  {author} {\bibinfo {author} {\bibfnamefont {Walter}\ \bibnamefont
  {Kohn}}\ and\ \bibinfo {author} {\bibfnamefont {Lu~Jeu}\ \bibnamefont
  {Sham}},\ }\bibfield  {title} {\enquote {\bibinfo {title} {Self-consistent
  equations including exchange and correlation effects},}\ }\href@noop {}
  {\bibfield  {journal} {\bibinfo  {journal} {Physical Review}\ }\textbf
  {\bibinfo {volume} {140}},\ \bibinfo {pages} {A1133} (\bibinfo {year}
  {1965})}\BibitemShut {NoStop}%
\bibitem [{\citenamefont {Snyder}\ \emph {et~al.}(2012)\citenamefont {Snyder},
  \citenamefont {Rupp}, \citenamefont {Hansen}, \citenamefont {M{\"u}ller},\
  and\ \citenamefont {Burke}}]{SRHM12}%
  \BibitemOpen
  \bibfield  {author} {\bibinfo {author} {\bibfnamefont {John~C.}\ \bibnamefont
  {Snyder}}, \bibinfo {author} {\bibfnamefont {Matthias}\ \bibnamefont {Rupp}},
  \bibinfo {author} {\bibfnamefont {Katja}\ \bibnamefont {Hansen}}, \bibinfo
  {author} {\bibfnamefont {Klaus-Robert}\ \bibnamefont {M{\"u}ller}}, \ and\
  \bibinfo {author} {\bibfnamefont {Kieron}\ \bibnamefont {Burke}},\ }\bibfield
   {title} {\enquote {\bibinfo {title} {Finding density functionals with
  machine learning},}\ }\href {\doibase 10.1103/PhysRevLett.108.253002}
  {\bibfield  {journal} {\bibinfo  {journal} {Phys. Rev. Lett.}\ }\textbf
  {\bibinfo {volume} {108}},\ \bibinfo {pages} {253002} (\bibinfo {year}
  {2012})}\BibitemShut {NoStop}%
\bibitem [{\citenamefont {Behler}\ and\ \citenamefont
  {Parrinello}(2007)}]{behler2007generalized}%
  \BibitemOpen
  \bibfield  {author} {\bibinfo {author} {\bibfnamefont {J{\"o}rg}\
  \bibnamefont {Behler}}\ and\ \bibinfo {author} {\bibfnamefont {Michele}\
  \bibnamefont {Parrinello}},\ }\bibfield  {title} {\enquote {\bibinfo {title}
  {Generalized neural-network representation of high-dimensional
  potential-energy surfaces},}\ }\href@noop {} {\bibfield  {journal} {\bibinfo
  {journal} {Physical review letters}\ }\textbf {\bibinfo {volume} {98}},\
  \bibinfo {pages} {146401} (\bibinfo {year} {2007})}\BibitemShut {NoStop}%
\bibitem [{\citenamefont {Li}\ \emph {et~al.}(2016{\natexlab{a}})\citenamefont
  {Li}, \citenamefont {Baker}, \citenamefont {White},\ and\ \citenamefont
  {Burke}}]{LBWB16}%
  \BibitemOpen
  \bibfield  {author} {\bibinfo {author} {\bibfnamefont {Li}~\bibnamefont
  {Li}}, \bibinfo {author} {\bibfnamefont {Thomas~E.}\ \bibnamefont {Baker}},
  \bibinfo {author} {\bibfnamefont {Steven~R.}\ \bibnamefont {White}}, \ and\
  \bibinfo {author} {\bibfnamefont {Kieron}\ \bibnamefont {Burke}},\ }\bibfield
   {title} {\enquote {\bibinfo {title} {Pure density functional for strong
  correlation and the thermodynamic limit from machine learning},}\ }\href
  {\doibase 10.1103/PhysRevB.94.245129} {\bibfield  {journal} {\bibinfo
  {journal} {Phys. Rev. B}\ }\textbf {\bibinfo {volume} {94}},\ \bibinfo
  {pages} {245129} (\bibinfo {year} {2016}{\natexlab{a}})}\BibitemShut
  {NoStop}%
\bibitem [{\citenamefont {Brockherde}\ \emph {et~al.}(2017)\citenamefont
  {Brockherde}, \citenamefont {Vogt}, \citenamefont {Li}, \citenamefont
  {Tuckerman}, \citenamefont {Burke},\ and\ \citenamefont
  {M{\"u}ller}}]{brockherde2017bypassing}%
  \BibitemOpen
  \bibfield  {author} {\bibinfo {author} {\bibfnamefont {Felix}\ \bibnamefont
  {Brockherde}}, \bibinfo {author} {\bibfnamefont {Leslie}\ \bibnamefont
  {Vogt}}, \bibinfo {author} {\bibfnamefont {Li}~\bibnamefont {Li}}, \bibinfo
  {author} {\bibfnamefont {Mark~E}\ \bibnamefont {Tuckerman}}, \bibinfo
  {author} {\bibfnamefont {Kieron}\ \bibnamefont {Burke}}, \ and\ \bibinfo
  {author} {\bibfnamefont {Klaus-Robert}\ \bibnamefont {M{\"u}ller}},\
  }\bibfield  {title} {\enquote {\bibinfo {title} {{Bypassing the Kohn-Sham
  equations with machine learning}},}\ }\href@noop {} {\bibfield  {journal}
  {\bibinfo  {journal} {Nature Communications}\ }\textbf {\bibinfo {volume}
  {8}},\ \bibinfo {pages} {872} (\bibinfo {year} {2017})}\BibitemShut {NoStop}%
\bibitem [{\citenamefont {Bogojeski}\ \emph {et~al.}(2019)\citenamefont
  {Bogojeski}, \citenamefont {Vogt-Maranto}, \citenamefont {Tuckerman},
  \citenamefont {Mueller},\ and\ \citenamefont {Burke}}]{bogojeski2019density}%
  \BibitemOpen
  \bibfield  {author} {\bibinfo {author} {\bibfnamefont {Mihail}\ \bibnamefont
  {Bogojeski}}, \bibinfo {author} {\bibfnamefont {Leslie}\ \bibnamefont
  {Vogt-Maranto}}, \bibinfo {author} {\bibfnamefont {Mark~E}\ \bibnamefont
  {Tuckerman}}, \bibinfo {author} {\bibfnamefont {Klaus-Robert}\ \bibnamefont
  {Mueller}}, \ and\ \bibinfo {author} {\bibfnamefont {Kieron}\ \bibnamefont
  {Burke}},\ }\bibfield  {title} {\enquote {\bibinfo {title} {Density
  functionals with quantum chemical accuracy: From machine learning to
  molecular dynamics},}\ }\href@noop {} {\bibfield  {journal} {\bibinfo
  {journal} {Preprint at ChemRxiv https://doi. org/10.26434/chemrxiv}\ }\textbf
  {\bibinfo {volume} {8079917}},\ \bibinfo {pages} {v1} (\bibinfo {year}
  {2019})}\BibitemShut {NoStop}%
\bibitem [{\citenamefont {Bogojeski}\ \emph {et~al.}(2018)\citenamefont
  {Bogojeski}, \citenamefont {Brockherde}, \citenamefont {Vogt-Maranto},
  \citenamefont {Li}, \citenamefont {Tuckerman}, \citenamefont {Burke},\ and\
  \citenamefont {M{\"u}ller}}]{bogojeski2018efficient}%
  \BibitemOpen
  \bibfield  {author} {\bibinfo {author} {\bibfnamefont {Mihail}\ \bibnamefont
  {Bogojeski}}, \bibinfo {author} {\bibfnamefont {Felix}\ \bibnamefont
  {Brockherde}}, \bibinfo {author} {\bibfnamefont {Leslie}\ \bibnamefont
  {Vogt-Maranto}}, \bibinfo {author} {\bibfnamefont {Li}~\bibnamefont {Li}},
  \bibinfo {author} {\bibfnamefont {Mark~E}\ \bibnamefont {Tuckerman}},
  \bibinfo {author} {\bibfnamefont {Kieron}\ \bibnamefont {Burke}}, \ and\
  \bibinfo {author} {\bibfnamefont {Klaus-Robert}\ \bibnamefont {M{\"u}ller}},\
  }\bibfield  {title} {\enquote {\bibinfo {title} {{Efficient prediction of 3D
  electron densities using machine learning}},}\ }\href@noop {} {\bibfield
  {journal} {\bibinfo  {journal} {arXiv preprint arXiv:1811.06255}\ } (\bibinfo
  {year} {2018})}\BibitemShut {NoStop}%
\bibitem [{\citenamefont {Nagai}\ \emph {et~al.}(2018)\citenamefont {Nagai},
  \citenamefont {Akashi}, \citenamefont {Sasaki},\ and\ \citenamefont
  {Tsuneyuki}}]{nagai2018neural}%
  \BibitemOpen
  \bibfield  {author} {\bibinfo {author} {\bibfnamefont {Ryo}\ \bibnamefont
  {Nagai}}, \bibinfo {author} {\bibfnamefont {Ryosuke}\ \bibnamefont {Akashi}},
  \bibinfo {author} {\bibfnamefont {Shu}\ \bibnamefont {Sasaki}}, \ and\
  \bibinfo {author} {\bibfnamefont {Shinji}\ \bibnamefont {Tsuneyuki}},\
  }\bibfield  {title} {\enquote {\bibinfo {title} {Neural-network kohn-sham
  exchange-correlation potential and its out-of-training transferability},}\
  }\href@noop {} {\bibfield  {journal} {\bibinfo  {journal} {J.~Chem.~Phys.}\
  }\textbf {\bibinfo {volume} {148}},\ \bibinfo {pages} {241737} (\bibinfo
  {year} {2018})}\BibitemShut {NoStop}%
\bibitem [{\citenamefont {Li}\ \emph {et~al.}(2016{\natexlab{b}})\citenamefont
  {Li}, \citenamefont {Snyder}, \citenamefont {Pelaschier}, \citenamefont
  {Huang}, \citenamefont {Niranjan}, \citenamefont {Duncan}, \citenamefont
  {Rupp}, \citenamefont {M{\"u}ller},\ and\ \citenamefont
  {Burke}}]{li2016understanding}%
  \BibitemOpen
  \bibfield  {author} {\bibinfo {author} {\bibfnamefont {Li}~\bibnamefont
  {Li}}, \bibinfo {author} {\bibfnamefont {John~C}\ \bibnamefont {Snyder}},
  \bibinfo {author} {\bibfnamefont {Isabelle~M}\ \bibnamefont {Pelaschier}},
  \bibinfo {author} {\bibfnamefont {Jessica}\ \bibnamefont {Huang}}, \bibinfo
  {author} {\bibfnamefont {Uma-Naresh}\ \bibnamefont {Niranjan}}, \bibinfo
  {author} {\bibfnamefont {Paul}\ \bibnamefont {Duncan}}, \bibinfo {author}
  {\bibfnamefont {Matthias}\ \bibnamefont {Rupp}}, \bibinfo {author}
  {\bibfnamefont {Klaus-Robert}\ \bibnamefont {M{\"u}ller}}, \ and\ \bibinfo
  {author} {\bibfnamefont {Kieron}\ \bibnamefont {Burke}},\ }\bibfield  {title}
  {\enquote {\bibinfo {title} {Understanding machine-learned density
  functionals},}\ }\href@noop {} {\bibfield  {journal} {\bibinfo  {journal}
  {International Journal of Quantum Chemistry}\ }\textbf {\bibinfo {volume}
  {116}},\ \bibinfo {pages} {819--833} (\bibinfo {year}
  {2016}{\natexlab{b}})}\BibitemShut {NoStop}%
\bibitem [{\citenamefont {Grisafi}\ \emph {et~al.}(2018)\citenamefont
  {Grisafi}, \citenamefont {Fabrizio}, \citenamefont {Meyer}, \citenamefont
  {Wilkins}, \citenamefont {Corminboeuf},\ and\ \citenamefont
  {Ceriotti}}]{grisafi2018transferable}%
  \BibitemOpen
  \bibfield  {author} {\bibinfo {author} {\bibfnamefont {Andrea}\ \bibnamefont
  {Grisafi}}, \bibinfo {author} {\bibfnamefont {Alberto}\ \bibnamefont
  {Fabrizio}}, \bibinfo {author} {\bibfnamefont {Benjamin}\ \bibnamefont
  {Meyer}}, \bibinfo {author} {\bibfnamefont {David~M}\ \bibnamefont
  {Wilkins}}, \bibinfo {author} {\bibfnamefont {Clemence}\ \bibnamefont
  {Corminboeuf}}, \ and\ \bibinfo {author} {\bibfnamefont {Michele}\
  \bibnamefont {Ceriotti}},\ }\bibfield  {title} {\enquote {\bibinfo {title}
  {Transferable machine-learning model of the electron density},}\ }\href@noop
  {} {\bibfield  {journal} {\bibinfo  {journal} {ACS central science}\ }\textbf
  {\bibinfo {volume} {5}},\ \bibinfo {pages} {57--64} (\bibinfo {year}
  {2018})}\BibitemShut {NoStop}%
\bibitem [{\citenamefont {Fabrizio}\ \emph {et~al.}(2019)\citenamefont
  {Fabrizio}, \citenamefont {Grisafi}, \citenamefont {Meyer}, \citenamefont
  {Ceriotti},\ and\ \citenamefont {Corminboeuf}}]{fabrizio2019electron}%
  \BibitemOpen
  \bibfield  {author} {\bibinfo {author} {\bibfnamefont {Alberto}\ \bibnamefont
  {Fabrizio}}, \bibinfo {author} {\bibfnamefont {Andrea}\ \bibnamefont
  {Grisafi}}, \bibinfo {author} {\bibfnamefont {Benjamin}\ \bibnamefont
  {Meyer}}, \bibinfo {author} {\bibfnamefont {Michele}\ \bibnamefont
  {Ceriotti}}, \ and\ \bibinfo {author} {\bibfnamefont {Clemence}\ \bibnamefont
  {Corminboeuf}},\ }\bibfield  {title} {\enquote {\bibinfo {title} {Electron
  density learning of non-covalent systems},}\ }\href@noop {} {\bibfield
  {journal} {\bibinfo  {journal} {Chemical science}\ }\textbf {\bibinfo
  {volume} {10}},\ \bibinfo {pages} {9424--9432} (\bibinfo {year}
  {2019})}\BibitemShut {NoStop}%
\bibitem [{\citenamefont {Denner}\ \emph {et~al.}(2020)\citenamefont {Denner},
  \citenamefont {Fischer},\ and\ \citenamefont {Neupert}}]{denner2020active}%
  \BibitemOpen
  \bibfield  {author} {\bibinfo {author} {\bibfnamefont {M~Michael}\
  \bibnamefont {Denner}}, \bibinfo {author} {\bibfnamefont {Mark~H}\
  \bibnamefont {Fischer}}, \ and\ \bibinfo {author} {\bibfnamefont {Titus}\
  \bibnamefont {Neupert}},\ }\bibfield  {title} {\enquote {\bibinfo {title}
  {Active learning a one-dimensional density functional theory},}\ }\href@noop
  {} {\bibfield  {journal} {\bibinfo  {journal} {arXiv preprint
  arXiv:2005.03014}\ } (\bibinfo {year} {2020})}\BibitemShut {NoStop}%
\bibitem [{\citenamefont {Nagai}\ \emph {et~al.}(2020)\citenamefont {Nagai},
  \citenamefont {Akashi},\ and\ \citenamefont {Sugino}}]{nagai2020completing}%
  \BibitemOpen
  \bibfield  {author} {\bibinfo {author} {\bibfnamefont {Ryo}\ \bibnamefont
  {Nagai}}, \bibinfo {author} {\bibfnamefont {Ryosuke}\ \bibnamefont {Akashi}},
  \ and\ \bibinfo {author} {\bibfnamefont {Osamu}\ \bibnamefont {Sugino}},\
  }\bibfield  {title} {\enquote {\bibinfo {title} {Completing density
  functional theory by machine learning hidden messages from molecules},}\
  }\href@noop {} {\bibfield  {journal} {\bibinfo  {journal} {npj Computational
  Materials}\ }\textbf {\bibinfo {volume} {6}},\ \bibinfo {pages} {1--8}
  (\bibinfo {year} {2020})}\BibitemShut {NoStop}%
\bibitem [{\citenamefont {Manzhos}(2020)}]{manzhos2020machine}%
  \BibitemOpen
  \bibfield  {author} {\bibinfo {author} {\bibfnamefont {Sergei}\ \bibnamefont
  {Manzhos}},\ }\bibfield  {title} {\enquote {\bibinfo {title} {Machine
  learning for the solution of the schr{\"o}dinger equation},}\ }\href@noop {}
  {\bibfield  {journal} {\bibinfo  {journal} {Machine Learning: Science and
  Technology}\ }\textbf {\bibinfo {volume} {1}},\ \bibinfo {pages} {013002}
  (\bibinfo {year} {2020})}\BibitemShut {NoStop}%
\bibitem [{\citenamefont {Suzuki}\ \emph {et~al.}(2020)\citenamefont {Suzuki},
  \citenamefont {Nagai},\ and\ \citenamefont {Haruyama}}]{suzuki2020machine}%
  \BibitemOpen
  \bibfield  {author} {\bibinfo {author} {\bibfnamefont {Yasumitsu}\
  \bibnamefont {Suzuki}}, \bibinfo {author} {\bibfnamefont {Ryo}\ \bibnamefont
  {Nagai}}, \ and\ \bibinfo {author} {\bibfnamefont {Jun}\ \bibnamefont
  {Haruyama}},\ }\bibfield  {title} {\enquote {\bibinfo {title} {Machine
  learning exchange-correlation potential in time-dependent density functional
  theory},}\ }\href@noop {} {\bibfield  {journal} {\bibinfo  {journal} {arXiv
  preprint arXiv:2002.06542}\ } (\bibinfo {year} {2020})}\BibitemShut {NoStop}%
\bibitem [{\citenamefont {Wetherell}\ \emph {et~al.}(2020)\citenamefont
  {Wetherell}, \citenamefont {Costamagna}, \citenamefont {Gatti},\ and\
  \citenamefont {Reining}}]{wetherell2020insights}%
  \BibitemOpen
  \bibfield  {author} {\bibinfo {author} {\bibfnamefont {Jack}\ \bibnamefont
  {Wetherell}}, \bibinfo {author} {\bibfnamefont {Andrea}\ \bibnamefont
  {Costamagna}}, \bibinfo {author} {\bibfnamefont {Matteo}\ \bibnamefont
  {Gatti}}, \ and\ \bibinfo {author} {\bibfnamefont {Lucia}\ \bibnamefont
  {Reining}},\ }\bibfield  {title} {\enquote {\bibinfo {title} {Insights into
  one-body density matrices using deep learning},}\ }\href@noop {} {\bibfield
  {journal} {\bibinfo  {journal} {Faraday Discussions}\ } (\bibinfo {year}
  {2020})}\BibitemShut {NoStop}%
\bibitem [{\citenamefont {Snyder}\ \emph {et~al.}(2013)\citenamefont {Snyder},
  \citenamefont {Rupp}, \citenamefont {Hansen}, \citenamefont {Blooston},
  \citenamefont {M{\"u}ller},\ and\ \citenamefont {Burke}}]{SRHB13}%
  \BibitemOpen
  \bibfield  {author} {\bibinfo {author} {\bibfnamefont {John~C.}\ \bibnamefont
  {Snyder}}, \bibinfo {author} {\bibfnamefont {Matthias}\ \bibnamefont {Rupp}},
  \bibinfo {author} {\bibfnamefont {Katja}\ \bibnamefont {Hansen}}, \bibinfo
  {author} {\bibfnamefont {Leo}\ \bibnamefont {Blooston}}, \bibinfo {author}
  {\bibfnamefont {Klaus-Robert}\ \bibnamefont {M{\"u}ller}}, \ and\ \bibinfo
  {author} {\bibfnamefont {Kieron}\ \bibnamefont {Burke}},\ }\bibfield  {title}
  {\enquote {\bibinfo {title} {Orbital-free bond breaking via machine
  learning},}\ }\href {\doibase 10.1063/1.4834075} {\bibfield  {journal}
  {\bibinfo  {journal} {J. Chem. Phys.}\ }\textbf {\bibinfo {volume} {139}},\
  \bibinfo {pages} {224104} (\bibinfo {year} {2013})}\BibitemShut {NoStop}%
\bibitem [{\citenamefont {Snyder}\ \emph {et~al.}(2015)\citenamefont {Snyder},
  \citenamefont {Rupp}, \citenamefont {Müller},\ and\ \citenamefont
  {Burke}}]{SRMB15}%
  \BibitemOpen
  \bibfield  {author} {\bibinfo {author} {\bibfnamefont {John~C.}\ \bibnamefont
  {Snyder}}, \bibinfo {author} {\bibfnamefont {Matthias}\ \bibnamefont {Rupp}},
  \bibinfo {author} {\bibfnamefont {Klaus-Robert}\ \bibnamefont {Müller}}, \
  and\ \bibinfo {author} {\bibfnamefont {Kieron}\ \bibnamefont {Burke}},\
  }\bibfield  {title} {\enquote {\bibinfo {title} {Nonlinear gradient
  denoising: Finding accurate extrema from inaccurate functional
  derivatives},}\ }\href {\doibase 10.1002/qua.24937} {\bibfield  {journal}
  {\bibinfo  {journal} {Int.~J.~Quant.~Chem.}\ }\textbf {\bibinfo {volume}
  {115}},\ \bibinfo {pages} {1102--1114} (\bibinfo {year} {2015})}\BibitemShut
  {NoStop}%
\bibitem [{\citenamefont {Vu}\ \emph {et~al.}(2015)\citenamefont {Vu},
  \citenamefont {Snyder}, \citenamefont {Li}, \citenamefont {Rupp},
  \citenamefont {Chen}, \citenamefont {Khelif}, \citenamefont {M{\"u}ller},\
  and\ \citenamefont {Burke}}]{vu2015understanding}%
  \BibitemOpen
  \bibfield  {author} {\bibinfo {author} {\bibfnamefont {Kevin}\ \bibnamefont
  {Vu}}, \bibinfo {author} {\bibfnamefont {John~C}\ \bibnamefont {Snyder}},
  \bibinfo {author} {\bibfnamefont {Li}~\bibnamefont {Li}}, \bibinfo {author}
  {\bibfnamefont {Matthias}\ \bibnamefont {Rupp}}, \bibinfo {author}
  {\bibfnamefont {Brandon~F}\ \bibnamefont {Chen}}, \bibinfo {author}
  {\bibfnamefont {Tarek}\ \bibnamefont {Khelif}}, \bibinfo {author}
  {\bibfnamefont {Klaus-Robert}\ \bibnamefont {M{\"u}ller}}, \ and\ \bibinfo
  {author} {\bibfnamefont {Kieron}\ \bibnamefont {Burke}},\ }\bibfield  {title}
  {\enquote {\bibinfo {title} {Understanding kernel ridge regression: Common
  behaviors from simple functions to density functionals},}\ }\href@noop {}
  {\bibfield  {journal} {\bibinfo  {journal} {Int.~J.~Quant.~Chem.}\ }\textbf
  {\bibinfo {volume} {115}},\ \bibinfo {pages} {1115--1128} (\bibinfo {year}
  {2015})}\BibitemShut {NoStop}%
\bibitem [{\citenamefont {Hollingsworth}\ \emph {et~al.}(2018)\citenamefont
  {Hollingsworth}, \citenamefont {Li}, \citenamefont {Baker},\ and\
  \citenamefont {Burke}}]{hollingsworth2018can}%
  \BibitemOpen
  \bibfield  {author} {\bibinfo {author} {\bibfnamefont {Jacob}\ \bibnamefont
  {Hollingsworth}}, \bibinfo {author} {\bibfnamefont {Li}~\bibnamefont {Li}},
  \bibinfo {author} {\bibfnamefont {Thomas~E}\ \bibnamefont {Baker}}, \ and\
  \bibinfo {author} {\bibfnamefont {Kieron}\ \bibnamefont {Burke}},\ }\bibfield
   {title} {\enquote {\bibinfo {title} {Can exact conditions improve
  machine-learned density functionals?}}\ }\href@noop {} {\bibfield  {journal}
  {\bibinfo  {journal} {J.~Chem.~Phys.}\ }\textbf {\bibinfo {volume} {148}},\
  \bibinfo {pages} {241743} (\bibinfo {year} {2018})}\BibitemShut {NoStop}%
\bibitem [{\citenamefont {Engel}\ and\ \citenamefont
  {Dreizler}(2011)}]{engel2011density}%
  \BibitemOpen
  \bibfield  {author} {\bibinfo {author} {\bibfnamefont {Eberhard}\
  \bibnamefont {Engel}}\ and\ \bibinfo {author} {\bibfnamefont {Reiner~M}\
  \bibnamefont {Dreizler}},\ }\href@noop {} {\emph {\bibinfo {title} {Density
  functional theory: an advanced course}}}\ (\bibinfo  {publisher} {Springer
  Science \& Business Media},\ \bibinfo {year} {2011})\BibitemShut {NoStop}%
\bibitem [{\citenamefont {Gross}\ and\ \citenamefont
  {Dreizler}(2013)}]{gross2013density}%
  \BibitemOpen
  \bibfield  {author} {\bibinfo {author} {\bibfnamefont {Eberhard~KU}\
  \bibnamefont {Gross}}\ and\ \bibinfo {author} {\bibfnamefont {Reiner~M}\
  \bibnamefont {Dreizler}},\ }\href@noop {} {\emph {\bibinfo {title} {Density
  functional theory}}},\ \bibinfo {series} {NATO ASI Series}, Vol.\ \bibinfo
  {volume} {337}\ (\bibinfo  {publisher} {Springer Science \& Business Media},\
  \bibinfo {year} {2013})\BibitemShut {NoStop}%
\bibitem [{\citenamefont {Hatcher}\ \emph {et~al.}(2019)\citenamefont
  {Hatcher}, \citenamefont {Kittl},\ and\ \citenamefont
  {Bowen}}]{hatcher2019method}%
  \BibitemOpen
  \bibfield  {author} {\bibinfo {author} {\bibfnamefont {Ryan}\ \bibnamefont
  {Hatcher}}, \bibinfo {author} {\bibfnamefont {Jorge~A}\ \bibnamefont
  {Kittl}}, \ and\ \bibinfo {author} {\bibfnamefont {Christopher}\ \bibnamefont
  {Bowen}},\ }\bibfield  {title} {\enquote {\bibinfo {title} {A method to
  calculate correlation for density functional theory on a quantum
  processor},}\ }\href@noop {} {\bibfield  {journal} {\bibinfo  {journal}
  {arXiv preprint arXiv:1903.05550}\ } (\bibinfo {year} {2019})}\BibitemShut
  {NoStop}%
\bibitem [{\citenamefont {Whitfield}\ \emph {et~al.}(2014)\citenamefont
  {Whitfield}, \citenamefont {Yung}, \citenamefont {Tempel}, \citenamefont
  {Boixo},\ and\ \citenamefont {Aspuru-Guzik}}]{whitfield2014computational}%
  \BibitemOpen
  \bibfield  {author} {\bibinfo {author} {\bibfnamefont {James~D}\ \bibnamefont
  {Whitfield}}, \bibinfo {author} {\bibfnamefont {MH}~\bibnamefont {Yung}},
  \bibinfo {author} {\bibfnamefont {David~Gabriel}\ \bibnamefont {Tempel}},
  \bibinfo {author} {\bibfnamefont {S}~\bibnamefont {Boixo}}, \ and\ \bibinfo
  {author} {\bibfnamefont {Al{\'a}n}\ \bibnamefont {Aspuru-Guzik}},\ }\bibfield
   {title} {\enquote {\bibinfo {title} {Computational complexity of
  time-dependent density functional theory},}\ }\href@noop {} {\bibfield
  {journal} {\bibinfo  {journal} {New.~J.~Phys.}\ }\textbf {\bibinfo {volume}
  {16}},\ \bibinfo {pages} {083035} (\bibinfo {year} {2014})}\BibitemShut
  {NoStop}%
\bibitem [{\citenamefont {Brown}\ \emph {et~al.}(2019)\citenamefont {Brown},
  \citenamefont {Yang},\ and\ \citenamefont {Whitfield}}]{brown2019solver}%
  \BibitemOpen
  \bibfield  {author} {\bibinfo {author} {\bibfnamefont {James}\ \bibnamefont
  {Brown}}, \bibinfo {author} {\bibfnamefont {Jun}\ \bibnamefont {Yang}}, \
  and\ \bibinfo {author} {\bibfnamefont {James~D}\ \bibnamefont {Whitfield}},\
  }\bibfield  {title} {\enquote {\bibinfo {title} {{Solver for the electronic
  V-representation problem of time-dependent density functional theory}},}\
  }\href@noop {} {\bibfield  {journal} {\bibinfo  {journal} {arXiv preprint
  arXiv:1904.10958}\ } (\bibinfo {year} {2019})}\BibitemShut {NoStop}%
\bibitem [{\citenamefont {Yang}\ \emph {et~al.}(2019)\citenamefont {Yang},
  \citenamefont {Brown},\ and\ \citenamefont
  {Whitfield}}]{yang2019measurement}%
  \BibitemOpen
  \bibfield  {author} {\bibinfo {author} {\bibfnamefont {Jun}\ \bibnamefont
  {Yang}}, \bibinfo {author} {\bibfnamefont {James}\ \bibnamefont {Brown}}, \
  and\ \bibinfo {author} {\bibfnamefont {James~Daniel}\ \bibnamefont
  {Whitfield}},\ }\bibfield  {title} {\enquote {\bibinfo {title} {Measurement
  on quantum devices with applications to time-dependent density functional
  theory},}\ }\href@noop {} {\bibfield  {journal} {\bibinfo  {journal} {arXiv
  preprint arXiv:1909.03078}\ } (\bibinfo {year} {2019})}\BibitemShut {NoStop}%
\bibitem [{\citenamefont {Rall}(2020)}]{rall2020quantum}%
  \BibitemOpen
  \bibfield  {author} {\bibinfo {author} {\bibfnamefont {Patrick}\ \bibnamefont
  {Rall}},\ }\bibfield  {title} {\enquote {\bibinfo {title} {Quantum algorithms
  for estimating physical quantities using block-encodings},}\ }\href@noop {}
  {\bibfield  {journal} {\bibinfo  {journal} {arXiv preprint arXiv:2004.06832}\
  } (\bibinfo {year} {2020})}\BibitemShut {NoStop}%
\bibitem [{\citenamefont {Ullrich}(2011)}]{ullrich2011time}%
  \BibitemOpen
  \bibfield  {author} {\bibinfo {author} {\bibfnamefont {Carsten~A}\
  \bibnamefont {Ullrich}},\ }\href@noop {} {\emph {\bibinfo {title}
  {Time-dependent density-functional theory: concepts and applications}}}\
  (\bibinfo  {publisher} {OUP Oxford},\ \bibinfo {year} {2011})\BibitemShut
  {NoStop}%
\bibitem [{\citenamefont {Temme}\ \emph {et~al.}(2011)\citenamefont {Temme},
  \citenamefont {Osborne}, \citenamefont {Vollbrecht}, \citenamefont {Poulin},\
  and\ \citenamefont {Verstraete}}]{temme2011quantum}%
  \BibitemOpen
  \bibfield  {author} {\bibinfo {author} {\bibfnamefont {Kristan}\ \bibnamefont
  {Temme}}, \bibinfo {author} {\bibfnamefont {Tobias~J}\ \bibnamefont
  {Osborne}}, \bibinfo {author} {\bibfnamefont {Karl~G}\ \bibnamefont
  {Vollbrecht}}, \bibinfo {author} {\bibfnamefont {David}\ \bibnamefont
  {Poulin}}, \ and\ \bibinfo {author} {\bibfnamefont {Frank}\ \bibnamefont
  {Verstraete}},\ }\bibfield  {title} {\enquote {\bibinfo {title} {Quantum
  metropolis sampling},}\ }\href@noop {} {\bibfield  {journal} {\bibinfo
  {journal} {Nature}\ }\textbf {\bibinfo {volume} {471}},\ \bibinfo {pages}
  {87} (\bibinfo {year} {2011})}\BibitemShut {NoStop}%
\bibitem [{\citenamefont {Baker}(2020)}]{bakerGreen20}%
  \BibitemOpen
  \bibfield  {author} {\bibinfo {author} {\bibfnamefont {Thomas~E}\
  \bibnamefont {Baker}},\ }\bibfield  {title} {\enquote {\bibinfo {title}
  {Lanczos recursion on a quantum computer for the {G}reen's function},}\
  }\href@noop {} {\bibfield  {journal} {\bibinfo  {journal} {arXiv preprint
  arXiv: 2008.05593}\ } (\bibinfo {year} {2020})}\BibitemShut {NoStop}%
\bibitem [{\citenamefont {Servedio}\ and\ \citenamefont
  {Gortler}(2004)}]{servedio2004equivalences}%
  \BibitemOpen
  \bibfield  {author} {\bibinfo {author} {\bibfnamefont {Rocco~A}\ \bibnamefont
  {Servedio}}\ and\ \bibinfo {author} {\bibfnamefont {Steven~J}\ \bibnamefont
  {Gortler}},\ }\bibfield  {title} {\enquote {\bibinfo {title} {Equivalences
  and separations between quantum and classical learnability},}\ }\href@noop {}
  {\bibfield  {journal} {\bibinfo  {journal} {SIAM Journal on Computing}\
  }\textbf {\bibinfo {volume} {33}},\ \bibinfo {pages} {1067--1092} (\bibinfo
  {year} {2004})}\BibitemShut {NoStop}%
\bibitem [{\citenamefont {Tang}(2019)}]{tang2019quantum}%
  \BibitemOpen
  \bibfield  {author} {\bibinfo {author} {\bibfnamefont {Ewin}\ \bibnamefont
  {Tang}},\ }\bibfield  {title} {\enquote {\bibinfo {title} {A quantum-inspired
  classical algorithm for recommendation systems},}\ }in\ \href@noop {} {\emph
  {\bibinfo {booktitle} {Proceedings of the 51st Annual ACM SIGACT Symposium on
  Theory of Computing}}}\ (\bibinfo {year} {2019})\ pp.\ \bibinfo {pages}
  {217--228}\BibitemShut {NoStop}%
\bibitem [{\citenamefont {Tang}(2018)}]{tang2018quantum}%
  \BibitemOpen
  \bibfield  {author} {\bibinfo {author} {\bibfnamefont {Ewin}\ \bibnamefont
  {Tang}},\ }\bibfield  {title} {\enquote {\bibinfo {title} {Quantum-inspired
  classical algorithms for principal component analysis and supervised
  clustering},}\ }\href@noop {} {\bibfield  {journal} {\bibinfo  {journal}
  {arXiv preprint arXiv:1811.00414}\ } (\bibinfo {year} {2018})}\BibitemShut
  {NoStop}%
\bibitem [{\citenamefont {Gily{\'e}n}\ \emph {et~al.}(2018)\citenamefont
  {Gily{\'e}n}, \citenamefont {Lloyd},\ and\ \citenamefont
  {Tang}}]{gilyen2018quantum}%
  \BibitemOpen
  \bibfield  {author} {\bibinfo {author} {\bibfnamefont {Andr{\'a}s}\
  \bibnamefont {Gily{\'e}n}}, \bibinfo {author} {\bibfnamefont {Seth}\
  \bibnamefont {Lloyd}}, \ and\ \bibinfo {author} {\bibfnamefont {Ewin}\
  \bibnamefont {Tang}},\ }\bibfield  {title} {\enquote {\bibinfo {title}
  {Quantum-inspired low-rank stochastic regression with logarithmic dependence
  on the dimension},}\ }\href@noop {} {\bibfield  {journal} {\bibinfo
  {journal} {arXiv preprint arXiv:1811.04909}\ } (\bibinfo {year}
  {2018})}\BibitemShut {NoStop}%
\bibitem [{\citenamefont {Chia}\ \emph {et~al.}(2020)\citenamefont {Chia},
  \citenamefont {Gily{\'e}n}, \citenamefont {Li}, \citenamefont {Lin},
  \citenamefont {Tang},\ and\ \citenamefont {Wang}}]{chia2020sampling}%
  \BibitemOpen
  \bibfield  {author} {\bibinfo {author} {\bibfnamefont {Nai-Hui}\ \bibnamefont
  {Chia}}, \bibinfo {author} {\bibfnamefont {Andr{\'a}s}\ \bibnamefont
  {Gily{\'e}n}}, \bibinfo {author} {\bibfnamefont {Tongyang}\ \bibnamefont
  {Li}}, \bibinfo {author} {\bibfnamefont {Han-Hsuan}\ \bibnamefont {Lin}},
  \bibinfo {author} {\bibfnamefont {Ewin}\ \bibnamefont {Tang}}, \ and\
  \bibinfo {author} {\bibfnamefont {Chunhao}\ \bibnamefont {Wang}},\ }\bibfield
   {title} {\enquote {\bibinfo {title} {Sampling-based sublinear low-rank
  matrix arithmetic framework for dequantizing quantum machine learning},}\
  }in\ \href@noop {} {\emph {\bibinfo {booktitle} {Proceedings of the 52nd
  Annual ACM SIGACT Symposium on Theory of Computing}}}\ (\bibinfo {year}
  {2020})\ pp.\ \bibinfo {pages} {387--400}\BibitemShut {NoStop}%
\bibitem [{\citenamefont {Low}\ and\ \citenamefont
  {Chuang}(2019)}]{low2019hamiltonian}%
  \BibitemOpen
  \bibfield  {author} {\bibinfo {author} {\bibfnamefont {Guang~Hao}\
  \bibnamefont {Low}}\ and\ \bibinfo {author} {\bibfnamefont {Isaac~L}\
  \bibnamefont {Chuang}},\ }\bibfield  {title} {\enquote {\bibinfo {title}
  {Hamiltonian simulation by qubitization},}\ }\href@noop {} {\bibfield
  {journal} {\bibinfo  {journal} {Quantum}\ }\textbf {\bibinfo {volume} {3}},\
  \bibinfo {pages} {163} (\bibinfo {year} {2019})}\BibitemShut {NoStop}%
\bibitem [{\citenamefont {Yang}\ \emph {et~al.}(2004)\citenamefont {Yang},
  \citenamefont {Ayers},\ and\ \citenamefont {Wu}}]{yang2004potential}%
  \BibitemOpen
  \bibfield  {author} {\bibinfo {author} {\bibfnamefont {Weitao}\ \bibnamefont
  {Yang}}, \bibinfo {author} {\bibfnamefont {Paul~W}\ \bibnamefont {Ayers}}, \
  and\ \bibinfo {author} {\bibfnamefont {Qin}\ \bibnamefont {Wu}},\ }\bibfield
  {title} {\enquote {\bibinfo {title} {Potential functionals: dual to density
  functionals and solution to the v-representability problem},}\ }\href@noop {}
  {\bibfield  {journal} {\bibinfo  {journal} {Phys.~Rev.~Lett.}\ }\textbf
  {\bibinfo {volume} {92}},\ \bibinfo {pages} {146404} (\bibinfo {year}
  {2004})}\BibitemShut {NoStop}%
\bibitem [{\citenamefont {Perdew}(1985)}]{perdew1985kohn}%
  \BibitemOpen
  \bibfield  {author} {\bibinfo {author} {\bibfnamefont {John~P}\ \bibnamefont
  {Perdew}},\ }\bibfield  {title} {\enquote {\bibinfo {title} {What do the
  kohn-sham orbital energies mean? how do atoms dissociate?}}\ }in\ \href@noop
  {} {\emph {\bibinfo {booktitle} {Density Functional Methods in Physics}}},\
  \bibinfo {series and number} {NATO ASI Series}\ (\bibinfo  {publisher}
  {Springer},\ \bibinfo {year} {1985})\ pp.\ \bibinfo {pages}
  {265--308}\BibitemShut {NoStop}%
\bibitem [{\citenamefont {Perdew}\ and\ \citenamefont
  {Levy}(1997)}]{perdew1997comment}%
  \BibitemOpen
  \bibfield  {author} {\bibinfo {author} {\bibfnamefont {John~P}\ \bibnamefont
  {Perdew}}\ and\ \bibinfo {author} {\bibfnamefont {Mel}\ \bibnamefont
  {Levy}},\ }\bibfield  {title} {\enquote {\bibinfo {title} {Comment on
  "significance of the highest occupied kohn-sham eigenvalue"},}\ }\href@noop
  {} {\bibfield  {journal} {\bibinfo  {journal} {Phys.~Rev.~B}\ }\textbf
  {\bibinfo {volume} {56}},\ \bibinfo {pages} {16021} (\bibinfo {year}
  {1997})}\BibitemShut {NoStop}%
\bibitem [{\citenamefont {Levy}(1979)}]{levy1979universal}%
  \BibitemOpen
  \bibfield  {author} {\bibinfo {author} {\bibfnamefont {Mel}\ \bibnamefont
  {Levy}},\ }\bibfield  {title} {\enquote {\bibinfo {title} {Universal
  variational functionals of electron densities, first-order density matrices,
  and natural spin-orbitals and solution of the v-representability problem},}\
  }\href@noop {} {\bibfield  {journal} {\bibinfo  {journal} {Proceedings of the
  National Academy of Sciences}\ }\textbf {\bibinfo {volume} {76}},\ \bibinfo
  {pages} {6062--6065} (\bibinfo {year} {1979})}\BibitemShut {NoStop}%
\bibitem [{\citenamefont {Kohn}(1983)}]{kohn1983v}%
  \BibitemOpen
  \bibfield  {author} {\bibinfo {author} {\bibfnamefont {Walter}\ \bibnamefont
  {Kohn}},\ }\bibfield  {title} {\enquote {\bibinfo {title}
  {{v-Representability and density functional theory}},}\ }\href@noop {}
  {\bibfield  {journal} {\bibinfo  {journal} {Phys.~Rev.~Lett.}\ }\textbf
  {\bibinfo {volume} {51}},\ \bibinfo {pages} {1596} (\bibinfo {year}
  {1983})}\BibitemShut {NoStop}%
\bibitem [{\citenamefont {Chayes}\ \emph {et~al.}(1985)\citenamefont {Chayes},
  \citenamefont {Chayes},\ and\ \citenamefont {Ruskai}}]{chayes1985density}%
  \BibitemOpen
  \bibfield  {author} {\bibinfo {author} {\bibfnamefont {JT}~\bibnamefont
  {Chayes}}, \bibinfo {author} {\bibfnamefont {L}~\bibnamefont {Chayes}}, \
  and\ \bibinfo {author} {\bibfnamefont {Mary~Beth}\ \bibnamefont {Ruskai}},\
  }\bibfield  {title} {\enquote {\bibinfo {title} {Density functional approach
  to quantum lattice systems},}\ }\href@noop {} {\bibfield  {journal} {\bibinfo
   {journal} {J.~Stat.~Phys.}\ }\textbf {\bibinfo {volume} {38}},\ \bibinfo
  {pages} {497--518} (\bibinfo {year} {1985})}\BibitemShut {NoStop}%
\bibitem [{\citenamefont {Wagner}\ \emph {et~al.}(2014)\citenamefont {Wagner},
  \citenamefont {Baker}, \citenamefont {Stoudenmire}, \citenamefont {Burke},\
  and\ \citenamefont {White}}]{wagnerPRB14}%
  \BibitemOpen
  \bibfield  {author} {\bibinfo {author} {\bibfnamefont {Lucas~O}\ \bibnamefont
  {Wagner}}, \bibinfo {author} {\bibfnamefont {Thomas~E}\ \bibnamefont
  {Baker}}, \bibinfo {author} {\bibfnamefont {EM}~\bibnamefont {Stoudenmire}},
  \bibinfo {author} {\bibfnamefont {Kieron}\ \bibnamefont {Burke}}, \ and\
  \bibinfo {author} {\bibfnamefont {Steven~R}\ \bibnamefont {White}},\
  }\bibfield  {title} {\enquote {\bibinfo {title} {{Kohn-Sham calculations with
  the exact functional}},}\ }\href@noop {} {\bibfield  {journal} {\bibinfo
  {journal} {Phys.~Rev.~B}\ }\textbf {\bibinfo {volume} {90}},\ \bibinfo
  {pages} {045109} (\bibinfo {year} {2014})}\BibitemShut {NoStop}%
\bibitem [{\citenamefont {Gidopoulos}(2011)}]{gidopoulos2011progress}%
  \BibitemOpen
  \bibfield  {author} {\bibinfo {author} {\bibfnamefont {N.I.}\ \bibnamefont
  {Gidopoulos}},\ }\bibfield  {title} {\enquote {\bibinfo {title} {Progress at
  the interface of wave-function and density-functional theories},}\
  }\href@noop {} {\bibfield  {journal} {\bibinfo  {journal} {Phys.~Rev.~A}\
  }\textbf {\bibinfo {volume} {83}},\ \bibinfo {pages} {040502(R)} (\bibinfo
  {year} {2011})}\BibitemShut {NoStop}%
\bibitem [{\citenamefont {Callow}\ and\ \citenamefont
  {Gidopoulos}(2018)}]{callow2018optimal}%
  \BibitemOpen
  \bibfield  {author} {\bibinfo {author} {\bibfnamefont {Timothy~J}\
  \bibnamefont {Callow}}\ and\ \bibinfo {author} {\bibfnamefont {Nikitas~I}\
  \bibnamefont {Gidopoulos}},\ }\bibfield  {title} {\enquote {\bibinfo {title}
  {{Optimal power series expansions of the Kohn--Sham potential}},}\
  }\href@noop {} {\bibfield  {journal} {\bibinfo  {journal} {The European
  Physical Journal B}\ }\textbf {\bibinfo {volume} {91}},\ \bibinfo {pages}
  {209} (\bibinfo {year} {2018})}\BibitemShut {NoStop}%
\bibitem [{\citenamefont {Callow}\ \emph {et~al.}(2020)\citenamefont {Callow},
  \citenamefont {Lathiotakis},\ and\ \citenamefont
  {Gidopoulos}}]{callow2020density}%
  \BibitemOpen
  \bibfield  {author} {\bibinfo {author} {\bibfnamefont {Timothy~J}\
  \bibnamefont {Callow}}, \bibinfo {author} {\bibfnamefont {Nektarios~N}\
  \bibnamefont {Lathiotakis}}, \ and\ \bibinfo {author} {\bibfnamefont
  {Nikitas~I}\ \bibnamefont {Gidopoulos}},\ }\bibfield  {title} {\enquote
  {\bibinfo {title} {Density-inversion method for the kohn--sham potential:
  Role of the screening density},}\ }\href@noop {} {\bibfield  {journal}
  {\bibinfo  {journal} {The Journal of Chemical Physics}\ }\textbf {\bibinfo
  {volume} {152}},\ \bibinfo {pages} {164114} (\bibinfo {year}
  {2020})}\BibitemShut {NoStop}%
\bibitem [{\citenamefont {Jensen}\ and\ \citenamefont
  {Wasserman}(2016)}]{jensen2016numerical}%
  \BibitemOpen
  \bibfield  {author} {\bibinfo {author} {\bibfnamefont {Daniel~S}\
  \bibnamefont {Jensen}}\ and\ \bibinfo {author} {\bibfnamefont {Adam}\
  \bibnamefont {Wasserman}},\ }\bibfield  {title} {\enquote {\bibinfo {title}
  {Numerical density-to-potential inversions in time-dependent density
  functional theory},}\ }\href@noop {} {\bibfield  {journal} {\bibinfo
  {journal} {Phys.~Chem.~Chem.~Phys.}\ }\textbf {\bibinfo {volume} {18}},\
  \bibinfo {pages} {21079--21091} (\bibinfo {year} {2016})}\BibitemShut
  {NoStop}%
\bibitem [{\citenamefont {Jensen}\ and\ \citenamefont
  {Wasserman}(2018)}]{jensen2018numerical}%
  \BibitemOpen
  \bibfield  {author} {\bibinfo {author} {\bibfnamefont {Daniel~S}\
  \bibnamefont {Jensen}}\ and\ \bibinfo {author} {\bibfnamefont {Adam}\
  \bibnamefont {Wasserman}},\ }\bibfield  {title} {\enquote {\bibinfo {title}
  {Numerical methods for the inverse problem of density functional theory},}\
  }\href@noop {} {\bibfield  {journal} {\bibinfo  {journal}
  {Int.~J.~Quant.~Chem.}\ }\textbf {\bibinfo {volume} {118}},\ \bibinfo {pages}
  {e25425} (\bibinfo {year} {2018})}\BibitemShut {NoStop}%
\bibitem [{\citenamefont {Kanungo}\ \emph {et~al.}(2019)\citenamefont
  {Kanungo}, \citenamefont {Zimmerman},\ and\ \citenamefont
  {Gavini}}]{kanungo2019exact}%
  \BibitemOpen
  \bibfield  {author} {\bibinfo {author} {\bibfnamefont {Bikash}\ \bibnamefont
  {Kanungo}}, \bibinfo {author} {\bibfnamefont {Paul~M}\ \bibnamefont
  {Zimmerman}}, \ and\ \bibinfo {author} {\bibfnamefont {Vikram}\ \bibnamefont
  {Gavini}},\ }\bibfield  {title} {\enquote {\bibinfo {title} {Exact
  exchange-correlation potentials from ground-state electron densities},}\
  }\href@noop {} {\bibfield  {journal} {\bibinfo  {journal} {Nature
  communications}\ }\textbf {\bibinfo {volume} {10}},\ \bibinfo {pages} {1--9}
  (\bibinfo {year} {2019})}\BibitemShut {NoStop}%
\bibitem [{\citenamefont {Kumar}\ and\ \citenamefont
  {Harbola}(2020{\natexlab{a}})}]{kumar2020general}%
  \BibitemOpen
  \bibfield  {author} {\bibinfo {author} {\bibfnamefont {Ashish}\ \bibnamefont
  {Kumar}}\ and\ \bibinfo {author} {\bibfnamefont {Manoj~K}\ \bibnamefont
  {Harbola}},\ }\bibfield  {title} {\enquote {\bibinfo {title} {A general
  penalty method for density-to-potential inversion},}\ }\href@noop {}
  {\bibfield  {journal} {\bibinfo  {journal} {arXiv preprint arXiv:2004.03219}\
  } (\bibinfo {year} {2020}{\natexlab{a}})}\BibitemShut {NoStop}%
\bibitem [{\citenamefont {Draper}(2000)}]{draper2000addition}%
  \BibitemOpen
  \bibfield  {author} {\bibinfo {author} {\bibfnamefont {Thomas~G}\
  \bibnamefont {Draper}},\ }\bibfield  {title} {\enquote {\bibinfo {title}
  {Addition on a quantum computer},}\ }\href@noop {} {\bibfield  {journal}
  {\bibinfo  {journal} {arXiv preprint quant-ph/0008033}\ } (\bibinfo {year}
  {2000})}\BibitemShut {NoStop}%
\bibitem [{\citenamefont {Jordan}(2005)}]{jordan2005fast}%
  \BibitemOpen
  \bibfield  {author} {\bibinfo {author} {\bibfnamefont {Stephen~P}\
  \bibnamefont {Jordan}},\ }\bibfield  {title} {\enquote {\bibinfo {title}
  {Fast quantum algorithm for numerical gradient estimation},}\ }\href@noop {}
  {\bibfield  {journal} {\bibinfo  {journal} {Phys.~Rev.~Lett.}\ }\textbf
  {\bibinfo {volume} {95}},\ \bibinfo {pages} {050501} (\bibinfo {year}
  {2005})}\BibitemShut {NoStop}%
\bibitem [{\citenamefont {Gily{\'e}n}\ \emph
  {et~al.}(2019{\natexlab{a}})\citenamefont {Gily{\'e}n}, \citenamefont
  {Arunachalam},\ and\ \citenamefont {Wiebe}}]{gilyen2019optimizing}%
  \BibitemOpen
  \bibfield  {author} {\bibinfo {author} {\bibfnamefont {Andr{\'a}s}\
  \bibnamefont {Gily{\'e}n}}, \bibinfo {author} {\bibfnamefont {Srinivasan}\
  \bibnamefont {Arunachalam}}, \ and\ \bibinfo {author} {\bibfnamefont
  {Nathan}\ \bibnamefont {Wiebe}},\ }\bibfield  {title} {\enquote {\bibinfo
  {title} {Optimizing quantum optimization algorithms via faster quantum
  gradient computation},}\ }in\ \href@noop {} {\emph {\bibinfo {booktitle}
  {Proceedings of the Thirtieth Annual ACM-SIAM Symposium on Discrete
  Algorithms}}}\ (\bibinfo {organization} {Society for Industrial and Applied
  Mathematics},\ \bibinfo {year} {2019})\ pp.\ \bibinfo {pages}
  {1425--1444}\BibitemShut {NoStop}%
\bibitem [{\citenamefont {Medvedev}\ \emph {et~al.}(2017)\citenamefont
  {Medvedev}, \citenamefont {Bushmarinov}, \citenamefont {Sun}, \citenamefont
  {Perdew},\ and\ \citenamefont {Lyssenko}}]{medvedev2017density}%
  \BibitemOpen
  \bibfield  {author} {\bibinfo {author} {\bibfnamefont {Michael~G}\
  \bibnamefont {Medvedev}}, \bibinfo {author} {\bibfnamefont {Ivan~S}\
  \bibnamefont {Bushmarinov}}, \bibinfo {author} {\bibfnamefont {Jianwei}\
  \bibnamefont {Sun}}, \bibinfo {author} {\bibfnamefont {John~P}\ \bibnamefont
  {Perdew}}, \ and\ \bibinfo {author} {\bibfnamefont {Konstantin~A}\
  \bibnamefont {Lyssenko}},\ }\bibfield  {title} {\enquote {\bibinfo {title}
  {Density functional theory is straying from the path toward the exact
  functional},}\ }\href@noop {} {\bibfield  {journal} {\bibinfo  {journal}
  {Science}\ }\textbf {\bibinfo {volume} {355}},\ \bibinfo {pages} {49--52}
  (\bibinfo {year} {2017})}\BibitemShut {NoStop}%
\bibitem [{\citenamefont {Baker}(2017)}]{baker2017methods}%
  \BibitemOpen
  \bibfield  {author} {\bibinfo {author} {\bibfnamefont {Thomas~Edward}\
  \bibnamefont {Baker}},\ }\emph {\bibinfo {title} {Methods of Calculation with
  the Exact Density Functional using the Renormalization Group}},\ \href@noop
  {} {Ph.D. thesis},\ \bibinfo  {school} {University of California, Irvine}
  (\bibinfo {year} {2017})\BibitemShut {NoStop}%
\bibitem [{\citenamefont {Perdew}\ and\ \citenamefont
  {Baker}(2017)}]{perdew2017IPAM}%
  \BibitemOpen
  \bibfield  {author} {\bibinfo {author} {\bibfnamefont {John~P}\ \bibnamefont
  {Perdew}}\ and\ \bibinfo {author} {\bibfnamefont {Thomas~E}\ \bibnamefont
  {Baker}},\ }\enquote {\bibinfo {title} {{IPAM Book of DFT: The Generalized
  Gradient Approximation}},}\ \ (\bibinfo {year} {2017})\ Chap.~\bibinfo
  {chapter} {6}\BibitemShut {NoStop}%
\bibitem [{\citenamefont {Kumar}\ and\ \citenamefont
  {Harbola}(2020{\natexlab{b}})}]{kumar2020using}%
  \BibitemOpen
  \bibfield  {author} {\bibinfo {author} {\bibfnamefont {Ashish}\ \bibnamefont
  {Kumar}}\ and\ \bibinfo {author} {\bibfnamefont {Manoj~K}\ \bibnamefont
  {Harbola}},\ }\bibfield  {title} {\enquote {\bibinfo {title} {Using random
  numbers to obtain kohn-sham potential for a given density},}\ }\href@noop {}
  {\bibfield  {journal} {\bibinfo  {journal} {arXiv preprint arXiv:2006.00324}\
  } (\bibinfo {year} {2020}{\natexlab{b}})}\BibitemShut {NoStop}%
\bibitem [{\citenamefont {Gross}\ and\ \citenamefont
  {Proetto}(2009)}]{gross2009adiabatic}%
  \BibitemOpen
  \bibfield  {author} {\bibinfo {author} {\bibfnamefont {EKU}\ \bibnamefont
  {Gross}}\ and\ \bibinfo {author} {\bibfnamefont {CR}~\bibnamefont
  {Proetto}},\ }\bibfield  {title} {\enquote {\bibinfo {title} {{Adiabatic
  Connection and the Kohn- Sham Variety of Potential- Functional Theory}},}\
  }\href@noop {} {\bibfield  {journal} {\bibinfo  {journal} {Journal of
  chemical theory and computation}\ }\textbf {\bibinfo {volume} {5}},\ \bibinfo
  {pages} {844--849} (\bibinfo {year} {2009})}\BibitemShut {NoStop}%
\bibitem [{\citenamefont {Langer}\ \emph {et~al.}(2020)\citenamefont {Langer},
  \citenamefont {Goe{\ss}mann},\ and\ \citenamefont
  {Rupp}}]{langer2020representations}%
  \BibitemOpen
  \bibfield  {author} {\bibinfo {author} {\bibfnamefont {Marcel~F}\
  \bibnamefont {Langer}}, \bibinfo {author} {\bibfnamefont {Alex}\ \bibnamefont
  {Goe{\ss}mann}}, \ and\ \bibinfo {author} {\bibfnamefont {Matthias}\
  \bibnamefont {Rupp}},\ }\bibfield  {title} {\enquote {\bibinfo {title}
  {Representations of molecules and materials for interpolation of
  quantum-mechanical simulations via machine learning},}\ }\href@noop {}
  {\bibfield  {journal} {\bibinfo  {journal} {arXiv preprint arXiv:2003.12081}\
  } (\bibinfo {year} {2020})}\BibitemShut {NoStop}%
\bibitem [{\citenamefont {Hastings}(2004)}]{hastings2004locality}%
  \BibitemOpen
  \bibfield  {author} {\bibinfo {author} {\bibfnamefont {Matthew~B}\
  \bibnamefont {Hastings}},\ }\bibfield  {title} {\enquote {\bibinfo {title}
  {Locality in quantum and markov dynamics on lattices and networks},}\
  }\href@noop {} {\bibfield  {journal} {\bibinfo  {journal} {Phys.~Rev.~Lett.}\
  }\textbf {\bibinfo {volume} {93}},\ \bibinfo {pages} {140402} (\bibinfo
  {year} {2004})}\BibitemShut {NoStop}%
\bibitem [{\citenamefont {Kohn}(1996)}]{kohn1996density}%
  \BibitemOpen
  \bibfield  {author} {\bibinfo {author} {\bibfnamefont {Walter}\ \bibnamefont
  {Kohn}},\ }\bibfield  {title} {\enquote {\bibinfo {title} {Density functional
  and density matrix method scaling linearly with the number of atoms},}\
  }\href@noop {} {\bibfield  {journal} {\bibinfo  {journal} {Phys.~Rev.~Lett.}\
  }\textbf {\bibinfo {volume} {76}},\ \bibinfo {pages} {3168} (\bibinfo {year}
  {1996})}\BibitemShut {NoStop}%
\bibitem [{\citenamefont {Prodan}\ and\ \citenamefont
  {Kohn}(2005)}]{prodan2005nearsightedness}%
  \BibitemOpen
  \bibfield  {author} {\bibinfo {author} {\bibfnamefont {Emil}\ \bibnamefont
  {Prodan}}\ and\ \bibinfo {author} {\bibfnamefont {Walter}\ \bibnamefont
  {Kohn}},\ }\bibfield  {title} {\enquote {\bibinfo {title} {Nearsightedness of
  electronic matter},}\ }\href@noop {} {\bibfield  {journal} {\bibinfo
  {journal} {Proceedings of the National Academy of Sciences of the United
  States of America}\ }\textbf {\bibinfo {volume} {102}},\ \bibinfo {pages}
  {11635--11638} (\bibinfo {year} {2005})}\BibitemShut {NoStop}%
\bibitem [{\citenamefont {Baker}\ \emph {et~al.}(2019)\citenamefont {Baker},
  \citenamefont {Desrosiers}, \citenamefont {Tremblay},\ and\ \citenamefont
  {Thompson}}]{baker2019m}%
  \BibitemOpen
  \bibfield  {author} {\bibinfo {author} {\bibfnamefont {Thomas~E}\
  \bibnamefont {Baker}}, \bibinfo {author} {\bibfnamefont {Samuel}\
  \bibnamefont {Desrosiers}}, \bibinfo {author} {\bibfnamefont {Maxime}\
  \bibnamefont {Tremblay}}, \ and\ \bibinfo {author} {\bibfnamefont {Martin~P}\
  \bibnamefont {Thompson}},\ }\bibfield  {title} {\enquote {\bibinfo {title}
  {M\'ethodes de calcul avec r\'eseaux de tenseurs en physique (basic tensor
  network computations in physics)},}\ }\href@noop {} {\bibfield  {journal}
  {\bibinfo  {journal} {arXiv preprint arXiv:1911.11566}\ } (\bibinfo {year}
  {2019})}\BibitemShut {NoStop}%
\bibitem [{\citenamefont {Vidal}(2008)}]{vidal2008class}%
  \BibitemOpen
  \bibfield  {author} {\bibinfo {author} {\bibfnamefont {Guifr{\'e}}\
  \bibnamefont {Vidal}},\ }\bibfield  {title} {\enquote {\bibinfo {title}
  {Class of quantum many-body states that can be efficiently simulated},}\
  }\href@noop {} {\bibfield  {journal} {\bibinfo  {journal} {Phys.~Rev.~Lett.}\
  }\textbf {\bibinfo {volume} {101}},\ \bibinfo {pages} {110501} (\bibinfo
  {year} {2008})}\BibitemShut {NoStop}%
\bibitem [{\citenamefont {Wagner}\ \emph {et~al.}(2013)\citenamefont {Wagner},
  \citenamefont {Stoudenmire}, \citenamefont {Burke},\ and\ \citenamefont
  {White}}]{wagner2013guaranteed}%
  \BibitemOpen
  \bibfield  {author} {\bibinfo {author} {\bibfnamefont {Lucas~O}\ \bibnamefont
  {Wagner}}, \bibinfo {author} {\bibfnamefont {EM}~\bibnamefont {Stoudenmire}},
  \bibinfo {author} {\bibfnamefont {Kieron}\ \bibnamefont {Burke}}, \ and\
  \bibinfo {author} {\bibfnamefont {Steven~R}\ \bibnamefont {White}},\
  }\bibfield  {title} {\enquote {\bibinfo {title} {{Guaranteed convergence of
  the Kohn-Sham equations}},}\ }\href@noop {} {\bibfield  {journal} {\bibinfo
  {journal} {Phys.~Rev.~Lett.}\ }\textbf {\bibinfo {volume} {111}},\ \bibinfo
  {pages} {093003} (\bibinfo {year} {2013})}\BibitemShut {NoStop}%
\bibitem [{\citenamefont {Li}\ \emph {et~al.}(2015)\citenamefont {Li},
  \citenamefont {Kermode},\ and\ \citenamefont {De~Vita}}]{li2015molecular}%
  \BibitemOpen
  \bibfield  {author} {\bibinfo {author} {\bibfnamefont {Zhenwei}\ \bibnamefont
  {Li}}, \bibinfo {author} {\bibfnamefont {James~R}\ \bibnamefont {Kermode}}, \
  and\ \bibinfo {author} {\bibfnamefont {Alessandro}\ \bibnamefont {De~Vita}},\
  }\bibfield  {title} {\enquote {\bibinfo {title} {Molecular dynamics with
  on-the-fly machine learning of quantum-mechanical forces},}\ }\href@noop {}
  {\bibfield  {journal} {\bibinfo  {journal} {Physical review letters}\
  }\textbf {\bibinfo {volume} {114}},\ \bibinfo {pages} {096405} (\bibinfo
  {year} {2015})}\BibitemShut {NoStop}%
\bibitem [{\citenamefont {Gilbert}(1975)}]{gilbert1975hohenberg}%
  \BibitemOpen
  \bibfield  {author} {\bibinfo {author} {\bibfnamefont {TL}~\bibnamefont
  {Gilbert}},\ }\bibfield  {title} {\enquote {\bibinfo {title} {{Hohenberg-Kohn
  theorem for nonlocal external potentials}},}\ }\href@noop {} {\bibfield
  {journal} {\bibinfo  {journal} {Phys.~Rev.~B}\ }\textbf {\bibinfo {volume}
  {12}},\ \bibinfo {pages} {2111} (\bibinfo {year} {1975})}\BibitemShut
  {NoStop}%
\bibitem [{\citenamefont {Gidopoulos}\ and\ \citenamefont
  {Gross}(2014)}]{gidopoulos2014electronic}%
  \BibitemOpen
  \bibfield  {author} {\bibinfo {author} {\bibfnamefont {Nikitas~I}\
  \bibnamefont {Gidopoulos}}\ and\ \bibinfo {author} {\bibfnamefont {EKU}\
  \bibnamefont {Gross}},\ }\bibfield  {title} {\enquote {\bibinfo {title}
  {{Electronic non-adiabatic states: towards a density functional theory beyond
  the Born--Oppenheimer approximation}},}\ }\href@noop {} {\bibfield  {journal}
  {\bibinfo  {journal} {Phil. Trans. R. Soc. A}\ }\textbf {\bibinfo {volume}
  {372}},\ \bibinfo {pages} {20130059} (\bibinfo {year} {2014})}\BibitemShut
  {NoStop}%
\bibitem [{\citenamefont {Runge}\ and\ \citenamefont
  {Gross}(1984)}]{runge1984density}%
  \BibitemOpen
  \bibfield  {author} {\bibinfo {author} {\bibfnamefont {Erich}\ \bibnamefont
  {Runge}}\ and\ \bibinfo {author} {\bibfnamefont {Eberhard~KU}\ \bibnamefont
  {Gross}},\ }\bibfield  {title} {\enquote {\bibinfo {title}
  {Density-functional theory for time-dependent systems},}\ }\href@noop {}
  {\bibfield  {journal} {\bibinfo  {journal} {Phys.~Rev.~Lett.}\ }\textbf
  {\bibinfo {volume} {52}},\ \bibinfo {pages} {997} (\bibinfo {year}
  {1984})}\BibitemShut {NoStop}%
\bibitem [{\citenamefont {van Leeuwen}(1998)}]{van1998causality}%
  \BibitemOpen
  \bibfield  {author} {\bibinfo {author} {\bibfnamefont {Robert}\ \bibnamefont
  {van Leeuwen}},\ }\bibfield  {title} {\enquote {\bibinfo {title} {Causality
  and symmetry in time-dependent density-functional theory},}\ }\href@noop {}
  {\bibfield  {journal} {\bibinfo  {journal} {Phys.~Rev.~Lett.}\ }\textbf
  {\bibinfo {volume} {80}},\ \bibinfo {pages} {1280} (\bibinfo {year}
  {1998})}\BibitemShut {NoStop}%
\bibitem [{\citenamefont {Elliott}\ \emph {et~al.}(2009)\citenamefont
  {Elliott}, \citenamefont {Furche},\ and\ \citenamefont
  {Burke}}]{elliott20093}%
  \BibitemOpen
  \bibfield  {author} {\bibinfo {author} {\bibfnamefont {Peter}\ \bibnamefont
  {Elliott}}, \bibinfo {author} {\bibfnamefont {Filipp}\ \bibnamefont
  {Furche}}, \ and\ \bibinfo {author} {\bibfnamefont {Kieron}\ \bibnamefont
  {Burke}},\ }\bibfield  {title} {\enquote {\bibinfo {title} {3 excited states
  from time-dependent density functional theory},}\ }\href@noop {} {\bibfield
  {journal} {\bibinfo  {journal} {Reviews in computational chemistry}\ }\textbf
  {\bibinfo {volume} {26}},\ \bibinfo {pages} {91} (\bibinfo {year}
  {2009})}\BibitemShut {NoStop}%
\bibitem [{\citenamefont {Mermin}(1965)}]{mermin1965thermal}%
  \BibitemOpen
  \bibfield  {author} {\bibinfo {author} {\bibfnamefont {N~David}\ \bibnamefont
  {Mermin}},\ }\bibfield  {title} {\enquote {\bibinfo {title} {Thermal
  properties of the inhomogeneous electron gas},}\ }\href@noop {} {\bibfield
  {journal} {\bibinfo  {journal} {Physical Review}\ }\textbf {\bibinfo {volume}
  {137}},\ \bibinfo {pages} {A1441} (\bibinfo {year} {1965})}\BibitemShut
  {NoStop}%
\bibitem [{\citenamefont {Kohn}\ \emph {et~al.}(1989)\citenamefont {Kohn},
  \citenamefont {Gross},\ and\ \citenamefont {Oliveira}}]{kohn1989orbital}%
  \BibitemOpen
  \bibfield  {author} {\bibinfo {author} {\bibfnamefont {W}~\bibnamefont
  {Kohn}}, \bibinfo {author} {\bibfnamefont {EKU}\ \bibnamefont {Gross}}, \
  and\ \bibinfo {author} {\bibfnamefont {LN}~\bibnamefont {Oliveira}},\
  }\bibfield  {title} {\enquote {\bibinfo {title} {Orbital magnetism in the
  density functional theory of superconductors},}\ }\href@noop {} {\bibfield
  {journal} {\bibinfo  {journal} {Journal de physique}\ }\textbf {\bibinfo
  {volume} {50}},\ \bibinfo {pages} {2601--2612} (\bibinfo {year}
  {1989})}\BibitemShut {NoStop}%
\bibitem [{\citenamefont {Capelle}\ and\ \citenamefont
  {Gross}(1997)}]{capelle1997density}%
  \BibitemOpen
  \bibfield  {author} {\bibinfo {author} {\bibfnamefont {K}~\bibnamefont
  {Capelle}}\ and\ \bibinfo {author} {\bibfnamefont {EKU}\ \bibnamefont
  {Gross}},\ }\bibfield  {title} {\enquote {\bibinfo {title} {Density
  functional theory for triplet superconductors},}\ }\href@noop {} {\bibfield
  {journal} {\bibinfo  {journal} {Int.~J.~Quant.~Chem.}\ }\textbf {\bibinfo
  {volume} {61}},\ \bibinfo {pages} {325--332} (\bibinfo {year}
  {1997})}\BibitemShut {NoStop}%
\bibitem [{\citenamefont {Ruggenthaler}\ \emph {et~al.}(2011)\citenamefont
  {Ruggenthaler}, \citenamefont {Mackenroth},\ and\ \citenamefont
  {Bauer}}]{ruggenthaler2011time}%
  \BibitemOpen
  \bibfield  {author} {\bibinfo {author} {\bibfnamefont {Michael}\ \bibnamefont
  {Ruggenthaler}}, \bibinfo {author} {\bibfnamefont {F}~\bibnamefont
  {Mackenroth}}, \ and\ \bibinfo {author} {\bibfnamefont {Dieter}\ \bibnamefont
  {Bauer}},\ }\bibfield  {title} {\enquote {\bibinfo {title} {{Time-dependent
  Kohn-Sham approach to quantum electrodynamics}},}\ }\href@noop {} {\bibfield
  {journal} {\bibinfo  {journal} {Phys.~Rev.~A}\ }\textbf {\bibinfo {volume}
  {84}},\ \bibinfo {pages} {042107} (\bibinfo {year} {2011})}\BibitemShut
  {NoStop}%
\bibitem [{\citenamefont {Tokatly}(2013)}]{tokatly2013time}%
  \BibitemOpen
  \bibfield  {author} {\bibinfo {author} {\bibfnamefont {IV}~\bibnamefont
  {Tokatly}},\ }\bibfield  {title} {\enquote {\bibinfo {title} {Time-dependent
  density functional theory for many-electron systems interacting with cavity
  photons},}\ }\href@noop {} {\bibfield  {journal} {\bibinfo  {journal}
  {Phys.~Rev.~Lett.}\ }\textbf {\bibinfo {volume} {110}},\ \bibinfo {pages}
  {233001} (\bibinfo {year} {2013})}\BibitemShut {NoStop}%
\bibitem [{\citenamefont {Ruggenthaler}\ \emph {et~al.}(2014)\citenamefont
  {Ruggenthaler}, \citenamefont {Flick}, \citenamefont {Pellegrini},
  \citenamefont {Appel}, \citenamefont {Tokatly},\ and\ \citenamefont
  {Rubio}}]{ruggenthaler2014quantum}%
  \BibitemOpen
  \bibfield  {author} {\bibinfo {author} {\bibfnamefont {Michael}\ \bibnamefont
  {Ruggenthaler}}, \bibinfo {author} {\bibfnamefont {Johannes}\ \bibnamefont
  {Flick}}, \bibinfo {author} {\bibfnamefont {Camilla}\ \bibnamefont
  {Pellegrini}}, \bibinfo {author} {\bibfnamefont {Heiko}\ \bibnamefont
  {Appel}}, \bibinfo {author} {\bibfnamefont {Ilya~V}\ \bibnamefont {Tokatly}},
  \ and\ \bibinfo {author} {\bibfnamefont {Angel}\ \bibnamefont {Rubio}},\
  }\bibfield  {title} {\enquote {\bibinfo {title} {Quantum-electrodynamical
  density-functional theory: Bridging quantum optics and electronic-structure
  theory},}\ }\href@noop {} {\bibfield  {journal} {\bibinfo  {journal}
  {Phys.~Rev.~A}\ }\textbf {\bibinfo {volume} {90}},\ \bibinfo {pages} {012508}
  (\bibinfo {year} {2014})}\BibitemShut {NoStop}%
\bibitem [{\citenamefont {Filatov}(2015)}]{filatov2015ensemble}%
  \BibitemOpen
  \bibfield  {author} {\bibinfo {author} {\bibfnamefont {Michael}\ \bibnamefont
  {Filatov}},\ }\bibfield  {title} {\enquote {\bibinfo {title} {Ensemble dft
  approach to excited states of strongly correlated molecular systems},}\ }in\
  \href@noop {} {\emph {\bibinfo {booktitle} {Density-functional methods for
  excited states}}},\ \bibinfo {series and number} {Topics in Current
  Chemistry}\ (\bibinfo  {publisher} {Springer},\ \bibinfo {year} {2015})\ pp.\
  \bibinfo {pages} {97--124}\BibitemShut {NoStop}%
\bibitem [{\citenamefont {Jouzdani}\ \emph {et~al.}(2019)\citenamefont
  {Jouzdani}, \citenamefont {Bringuier},\ and\ \citenamefont
  {Kostuk}}]{jouzdani2019method}%
  \BibitemOpen
  \bibfield  {author} {\bibinfo {author} {\bibfnamefont {Pejman}\ \bibnamefont
  {Jouzdani}}, \bibinfo {author} {\bibfnamefont {Stefan}\ \bibnamefont
  {Bringuier}}, \ and\ \bibinfo {author} {\bibfnamefont {Mark}\ \bibnamefont
  {Kostuk}},\ }\bibfield  {title} {\enquote {\bibinfo {title} {A method of
  determining excited-states for quantum computation},}\ }\href@noop {}
  {\bibfield  {journal} {\bibinfo  {journal} {arXiv preprint arXiv:1908.05238}\
  } (\bibinfo {year} {2019})}\BibitemShut {NoStop}%
\bibitem [{\citenamefont {Cohen}\ and\ \citenamefont
  {Wasserman}(2006)}]{cohen2006hardness}%
  \BibitemOpen
  \bibfield  {author} {\bibinfo {author} {\bibfnamefont {Morrel~H}\
  \bibnamefont {Cohen}}\ and\ \bibinfo {author} {\bibfnamefont {Adam}\
  \bibnamefont {Wasserman}},\ }\bibfield  {title} {\enquote {\bibinfo {title}
  {On hardness and electronegativity equalization in chemical reactivity
  theory},}\ }\href@noop {} {\bibfield  {journal} {\bibinfo  {journal}
  {J.~Stat.~Phys.}\ }\textbf {\bibinfo {volume} {125}},\ \bibinfo {pages}
  {1121--1139} (\bibinfo {year} {2006})}\BibitemShut {NoStop}%
\bibitem [{\citenamefont {Cohen}\ and\ \citenamefont
  {Wasserman}(2007)}]{cohen2007foundations}%
  \BibitemOpen
  \bibfield  {author} {\bibinfo {author} {\bibfnamefont {Morrel~H}\
  \bibnamefont {Cohen}}\ and\ \bibinfo {author} {\bibfnamefont {Adam}\
  \bibnamefont {Wasserman}},\ }\bibfield  {title} {\enquote {\bibinfo {title}
  {On the foundations of chemical reactivity theory},}\ }\href@noop {}
  {\bibfield  {journal} {\bibinfo  {journal} {The Journal of Physical Chemistry
  A}\ }\textbf {\bibinfo {volume} {111}},\ \bibinfo {pages} {2229--2242}
  (\bibinfo {year} {2007})}\BibitemShut {NoStop}%
\bibitem [{\citenamefont {Elliott}\ \emph {et~al.}(2010)\citenamefont
  {Elliott}, \citenamefont {Burke}, \citenamefont {Cohen},\ and\ \citenamefont
  {Wasserman}}]{elliott2010partition}%
  \BibitemOpen
  \bibfield  {author} {\bibinfo {author} {\bibfnamefont {Peter}\ \bibnamefont
  {Elliott}}, \bibinfo {author} {\bibfnamefont {Kieron}\ \bibnamefont {Burke}},
  \bibinfo {author} {\bibfnamefont {Morrel~H}\ \bibnamefont {Cohen}}, \ and\
  \bibinfo {author} {\bibfnamefont {Adam}\ \bibnamefont {Wasserman}},\
  }\bibfield  {title} {\enquote {\bibinfo {title} {Partition density-functional
  theory},}\ }\href@noop {} {\bibfield  {journal} {\bibinfo  {journal}
  {Phys.~Rev.~A}\ }\textbf {\bibinfo {volume} {82}},\ \bibinfo {pages} {024501}
  (\bibinfo {year} {2010})}\BibitemShut {NoStop}%
\bibitem [{\citenamefont {Cangi}\ \emph {et~al.}(2011)\citenamefont {Cangi},
  \citenamefont {Lee}, \citenamefont {Elliott}, \citenamefont {Burke},\ and\
  \citenamefont {Gross}}]{cangi2011electronic}%
  \BibitemOpen
  \bibfield  {author} {\bibinfo {author} {\bibfnamefont {Attila}\ \bibnamefont
  {Cangi}}, \bibinfo {author} {\bibfnamefont {Donghyung}\ \bibnamefont {Lee}},
  \bibinfo {author} {\bibfnamefont {Peter}\ \bibnamefont {Elliott}}, \bibinfo
  {author} {\bibfnamefont {Kieron}\ \bibnamefont {Burke}}, \ and\ \bibinfo
  {author} {\bibfnamefont {E.K.U.}\ \bibnamefont {Gross}},\ }\bibfield  {title}
  {\enquote {\bibinfo {title} {Electronic structure via potential functional
  approximations},}\ }\href@noop {} {\bibfield  {journal} {\bibinfo  {journal}
  {Phys.~Rev.~Lett.}\ }\textbf {\bibinfo {volume} {106}},\ \bibinfo {pages}
  {236404} (\bibinfo {year} {2011})}\BibitemShut {NoStop}%
\bibitem [{\citenamefont {Cangi}\ \emph {et~al.}(2013)\citenamefont {Cangi},
  \citenamefont {Gross},\ and\ \citenamefont {Burke}}]{cangi2013potential}%
  \BibitemOpen
  \bibfield  {author} {\bibinfo {author} {\bibfnamefont {Attila}\ \bibnamefont
  {Cangi}}, \bibinfo {author} {\bibfnamefont {EKU}\ \bibnamefont {Gross}}, \
  and\ \bibinfo {author} {\bibfnamefont {Kieron}\ \bibnamefont {Burke}},\
  }\bibfield  {title} {\enquote {\bibinfo {title} {Potential functionals versus
  density functionals},}\ }\href@noop {} {\bibfield  {journal} {\bibinfo
  {journal} {Phys.~Rev.~A}\ }\textbf {\bibinfo {volume} {88}},\ \bibinfo
  {pages} {062505} (\bibinfo {year} {2013})}\BibitemShut {NoStop}%
\bibitem [{\citenamefont {Schrieffer}(1963)}]{schrieffer1963theory}%
  \BibitemOpen
  \bibfield  {author} {\bibinfo {author} {\bibfnamefont {JR}~\bibnamefont
  {Schrieffer}},\ }\bibfield  {title} {\enquote {\bibinfo {title} {Theory of
  superconductivity},}\ }\href@noop {} {\  (\bibinfo {year}
  {1963})}\BibitemShut {NoStop}%
\bibitem [{\citenamefont {Chowdhury}\ and\ \citenamefont
  {Somma}(2017)}]{chowdhury2017quantum}%
  \BibitemOpen
  \bibfield  {author} {\bibinfo {author} {\bibfnamefont {Anirban~Narayan}\
  \bibnamefont {Chowdhury}}\ and\ \bibinfo {author} {\bibfnamefont {Rolando~D}\
  \bibnamefont {Somma}},\ }\bibfield  {title} {\enquote {\bibinfo {title}
  {Quantum algorithms for gibbs sampling and hitting-time estimation},}\
  }\href@noop {} {\bibfield  {journal} {\bibinfo  {journal} {Quantum
  Information \& Computation}\ }\textbf {\bibinfo {volume} {17}},\ \bibinfo
  {pages} {41--64} (\bibinfo {year} {2017})}\BibitemShut {NoStop}%
\bibitem [{\citenamefont {Low}\ and\ \citenamefont
  {Chuang}(2017)}]{low2017optimal}%
  \BibitemOpen
  \bibfield  {author} {\bibinfo {author} {\bibfnamefont {Guang~Hao}\
  \bibnamefont {Low}}\ and\ \bibinfo {author} {\bibfnamefont {Isaac~L}\
  \bibnamefont {Chuang}},\ }\bibfield  {title} {\enquote {\bibinfo {title}
  {Optimal hamiltonian simulation by quantum signal processing},}\ }\href@noop
  {} {\bibfield  {journal} {\bibinfo  {journal} {Physical review letters}\
  }\textbf {\bibinfo {volume} {118}},\ \bibinfo {pages} {010501} (\bibinfo
  {year} {2017})}\BibitemShut {NoStop}%
\bibitem [{\citenamefont {Mattsson}\ and\ \citenamefont
  {Desjarlais}(2006)}]{mattsson2006phase}%
  \BibitemOpen
  \bibfield  {author} {\bibinfo {author} {\bibfnamefont {Thomas~R}\
  \bibnamefont {Mattsson}}\ and\ \bibinfo {author} {\bibfnamefont {Michael~P}\
  \bibnamefont {Desjarlais}},\ }\bibfield  {title} {\enquote {\bibinfo {title}
  {Phase diagram and electrical conductivity of high energy-density water from
  density functional theory},}\ }\href@noop {} {\bibfield  {journal} {\bibinfo
  {journal} {Phys.~Rev.~Lett.}\ }\textbf {\bibinfo {volume} {97}},\ \bibinfo
  {pages} {017801} (\bibinfo {year} {2006})}\BibitemShut {NoStop}%
\bibitem [{\citenamefont {Poulin}\ and\ \citenamefont
  {Wocjan}(2009)}]{poulin2009sampling}%
  \BibitemOpen
  \bibfield  {author} {\bibinfo {author} {\bibfnamefont {David}\ \bibnamefont
  {Poulin}}\ and\ \bibinfo {author} {\bibfnamefont {Pawel}\ \bibnamefont
  {Wocjan}},\ }\bibfield  {title} {\enquote {\bibinfo {title} {Sampling from
  the thermal quantum gibbs state and evaluating partition functions with a
  quantum computer},}\ }\href@noop {} {\bibfield  {journal} {\bibinfo
  {journal} {Phys.~Rev.~Lett.}\ }\textbf {\bibinfo {volume} {103}},\ \bibinfo
  {pages} {220502} (\bibinfo {year} {2009})}\BibitemShut {NoStop}%
\bibitem [{\citenamefont {Stoudenmire}\ \emph {et~al.}(2012)\citenamefont
  {Stoudenmire}, \citenamefont {Wagner}, \citenamefont {White},\ and\
  \citenamefont {Burke}}]{stoudenmire2012one}%
  \BibitemOpen
  \bibfield  {author} {\bibinfo {author} {\bibfnamefont {EM}~\bibnamefont
  {Stoudenmire}}, \bibinfo {author} {\bibfnamefont {Lucas~O}\ \bibnamefont
  {Wagner}}, \bibinfo {author} {\bibfnamefont {Steven~R}\ \bibnamefont
  {White}}, \ and\ \bibinfo {author} {\bibfnamefont {Kieron}\ \bibnamefont
  {Burke}},\ }\bibfield  {title} {\enquote {\bibinfo {title} {One-dimensional
  continuum electronic structure with the density-matrix renormalization group
  and its implications for density-functional theory},}\ }\href@noop {}
  {\bibfield  {journal} {\bibinfo  {journal} {Phys.~Rev.~Lett.}\ }\textbf
  {\bibinfo {volume} {109}},\ \bibinfo {pages} {056402} (\bibinfo {year}
  {2012})}\BibitemShut {NoStop}%
\bibitem [{\citenamefont {Shulenburger}\ \emph {et~al.}(2009)\citenamefont
  {Shulenburger}, \citenamefont {Casula}, \citenamefont {Senatore},\ and\
  \citenamefont {Martin}}]{shulenburger2009spin}%
  \BibitemOpen
  \bibfield  {author} {\bibinfo {author} {\bibfnamefont {Luke}\ \bibnamefont
  {Shulenburger}}, \bibinfo {author} {\bibfnamefont {Michele}\ \bibnamefont
  {Casula}}, \bibinfo {author} {\bibfnamefont {Gaetano}\ \bibnamefont
  {Senatore}}, \ and\ \bibinfo {author} {\bibfnamefont {Richard~M}\
  \bibnamefont {Martin}},\ }\bibfield  {title} {\enquote {\bibinfo {title}
  {Spin resolved energy parametrization of a quasi-one-dimensional electron
  gas},}\ }\href@noop {} {\bibfield  {journal} {\bibinfo  {journal} {Journal of
  Physics A: Mathematical and Theoretical}\ }\textbf {\bibinfo {volume} {42}},\
  \bibinfo {pages} {214021} (\bibinfo {year} {2009})}\BibitemShut {NoStop}%
\bibitem [{\citenamefont {Wagner}\ \emph {et~al.}(2012)\citenamefont {Wagner},
  \citenamefont {Stoudenmire}, \citenamefont {Burke},\ and\ \citenamefont
  {White}}]{wagner2012reference}%
  \BibitemOpen
  \bibfield  {author} {\bibinfo {author} {\bibfnamefont {Lucas~O}\ \bibnamefont
  {Wagner}}, \bibinfo {author} {\bibfnamefont {EM}~\bibnamefont {Stoudenmire}},
  \bibinfo {author} {\bibfnamefont {Kieron}\ \bibnamefont {Burke}}, \ and\
  \bibinfo {author} {\bibfnamefont {Steven~R}\ \bibnamefont {White}},\
  }\bibfield  {title} {\enquote {\bibinfo {title} {Reference electronic
  structure calculations in one dimension},}\ }\href@noop {} {\bibfield
  {journal} {\bibinfo  {journal} {Phys.~Chem.~Chem.~Phys.}\ }\textbf {\bibinfo
  {volume} {14}},\ \bibinfo {pages} {8581--8590} (\bibinfo {year}
  {2012})}\BibitemShut {NoStop}%
\bibitem [{\citenamefont {Baker}\ \emph {et~al.}(2015)\citenamefont {Baker},
  \citenamefont {Stoudenmire}, \citenamefont {Wagner}, \citenamefont {Burke},\
  and\ \citenamefont {White}}]{bakerPRB15}%
  \BibitemOpen
  \bibfield  {author} {\bibinfo {author} {\bibfnamefont {T.E.}\ \bibnamefont
  {Baker}}, \bibinfo {author} {\bibfnamefont {E.M.}\ \bibnamefont
  {Stoudenmire}}, \bibinfo {author} {\bibfnamefont {L.O.}\ \bibnamefont
  {Wagner}}, \bibinfo {author} {\bibfnamefont {K.}~\bibnamefont {Burke}}, \
  and\ \bibinfo {author} {\bibfnamefont {S.R.}\ \bibnamefont {White}},\
  }\bibfield  {title} {\enquote {\bibinfo {title} {One-dimensional mimicking of
  electronic structure: The case for exponentials},}\ }\href@noop {} {\bibfield
   {journal} {\bibinfo  {journal} {Phys.~Rev.~B}\ }\textbf {\bibinfo {volume}
  {91}},\ \bibinfo {pages} {235141} (\bibinfo {year} {2015})}\BibitemShut
  {NoStop}%
\bibitem [{\citenamefont {Baker}\ \emph {et~al.}(2016)\citenamefont {Baker},
  \citenamefont {Stoudenmire}, \citenamefont {Wagner}, \citenamefont {Burke},\
  and\ \citenamefont {White}}]{baker2016erratum}%
  \BibitemOpen
  \bibfield  {author} {\bibinfo {author} {\bibfnamefont {T.E.}\ \bibnamefont
  {Baker}}, \bibinfo {author} {\bibfnamefont {E.M.}\ \bibnamefont
  {Stoudenmire}}, \bibinfo {author} {\bibfnamefont {L.O.}\ \bibnamefont
  {Wagner}}, \bibinfo {author} {\bibfnamefont {K.}~\bibnamefont {Burke}}, \
  and\ \bibinfo {author} {\bibfnamefont {S.R.}\ \bibnamefont {White}},\
  }\bibfield  {title} {\enquote {\bibinfo {title} {{Erratum: One-dimensional
  mimicking of electronic structure: The case for exponentials [Phys. Rev. B
  91, 235141 (2015)]}},}\ }\href@noop {} {\bibfield  {journal} {\bibinfo
  {journal} {Phys.~Rev.~B}\ }\textbf {\bibinfo {volume} {93}},\ \bibinfo
  {pages} {119912(E)} (\bibinfo {year} {2016})}\BibitemShut {NoStop}%
\bibitem [{\citenamefont {Helbig}\ \emph {et~al.}(2009)\citenamefont {Helbig},
  \citenamefont {Tokatly},\ and\ \citenamefont {Rubio}}]{helbig2009exact}%
  \BibitemOpen
  \bibfield  {author} {\bibinfo {author} {\bibfnamefont {N}~\bibnamefont
  {Helbig}}, \bibinfo {author} {\bibfnamefont {IV}~\bibnamefont {Tokatly}}, \
  and\ \bibinfo {author} {\bibfnamefont {Angel}\ \bibnamefont {Rubio}},\
  }\bibfield  {title} {\enquote {\bibinfo {title} {{Exact Kohn--Sham potential
  of strongly correlated finite systems}},}\ }\href@noop {} {\bibfield
  {journal} {\bibinfo  {journal} {J.~Chem.~Phys.}\ }\textbf {\bibinfo {volume}
  {131}},\ \bibinfo {pages} {224105} (\bibinfo {year} {2009})}\BibitemShut
  {NoStop}%
\bibitem [{\citenamefont {Helbig}\ \emph {et~al.}(2011)\citenamefont {Helbig},
  \citenamefont {Fuks}, \citenamefont {Casula}, \citenamefont {Verstraete},
  \citenamefont {Marques}, \citenamefont {Tokatly},\ and\ \citenamefont
  {Rubio}}]{helbig2011density}%
  \BibitemOpen
  \bibfield  {author} {\bibinfo {author} {\bibfnamefont {N}~\bibnamefont
  {Helbig}}, \bibinfo {author} {\bibfnamefont {Johanna~I}\ \bibnamefont
  {Fuks}}, \bibinfo {author} {\bibfnamefont {M}~\bibnamefont {Casula}},
  \bibinfo {author} {\bibfnamefont {Matthieu~J}\ \bibnamefont {Verstraete}},
  \bibinfo {author} {\bibfnamefont {MAL}\ \bibnamefont {Marques}}, \bibinfo
  {author} {\bibfnamefont {IV}~\bibnamefont {Tokatly}}, \ and\ \bibinfo
  {author} {\bibfnamefont {Angel}\ \bibnamefont {Rubio}},\ }\bibfield  {title}
  {\enquote {\bibinfo {title} {Density functional theory beyond the linear
  regime: Validating an adiabatic local density approximation},}\ }\href@noop
  {} {\bibfield  {journal} {\bibinfo  {journal} {Phys.~Rev.~A}\ }\textbf
  {\bibinfo {volume} {83}},\ \bibinfo {pages} {032503} (\bibinfo {year}
  {2011})}\BibitemShut {NoStop}%
\bibitem [{\citenamefont {Elliott}\ \emph {et~al.}(2012)\citenamefont
  {Elliott}, \citenamefont {Fuks}, \citenamefont {Rubio},\ and\ \citenamefont
  {Maitra}}]{elliott2012universal}%
  \BibitemOpen
  \bibfield  {author} {\bibinfo {author} {\bibfnamefont {Peter}\ \bibnamefont
  {Elliott}}, \bibinfo {author} {\bibfnamefont {Johanna~I}\ \bibnamefont
  {Fuks}}, \bibinfo {author} {\bibfnamefont {Angel}\ \bibnamefont {Rubio}}, \
  and\ \bibinfo {author} {\bibfnamefont {Neepa~T}\ \bibnamefont {Maitra}},\
  }\bibfield  {title} {\enquote {\bibinfo {title} {Universal dynamical steps in
  the exact time-dependent exchange-correlation potential},}\ }\href@noop {}
  {\bibfield  {journal} {\bibinfo  {journal} {Phys.~Rev.~Lett.}\ }\textbf
  {\bibinfo {volume} {109}},\ \bibinfo {pages} {266404} (\bibinfo {year}
  {2012})}\BibitemShut {NoStop}%
\bibitem [{\citenamefont {Fuks}\ \emph {et~al.}(2015)\citenamefont {Fuks},
  \citenamefont {Luo}, \citenamefont {Sandoval},\ and\ \citenamefont
  {Maitra}}]{fuks2015time}%
  \BibitemOpen
  \bibfield  {author} {\bibinfo {author} {\bibfnamefont {Johanna~I}\
  \bibnamefont {Fuks}}, \bibinfo {author} {\bibfnamefont {Kai}\ \bibnamefont
  {Luo}}, \bibinfo {author} {\bibfnamefont {Ernesto~D}\ \bibnamefont
  {Sandoval}}, \ and\ \bibinfo {author} {\bibfnamefont {Neepa~T}\ \bibnamefont
  {Maitra}},\ }\bibfield  {title} {\enquote {\bibinfo {title} {Time-resolved
  spectroscopy in time-dependent density functional theory: An exact
  condition},}\ }\href@noop {} {\bibfield  {journal} {\bibinfo  {journal}
  {Phys.~Rev.~Lett.}\ }\textbf {\bibinfo {volume} {114}},\ \bibinfo {pages}
  {183002} (\bibinfo {year} {2015})}\BibitemShut {NoStop}%
\bibitem [{\citenamefont {Lima}\ \emph {et~al.}(2003)\citenamefont {Lima},
  \citenamefont {Silva}, \citenamefont {Oliveira},\ and\ \citenamefont
  {Capelle}}]{lima2003density}%
  \BibitemOpen
  \bibfield  {author} {\bibinfo {author} {\bibfnamefont {N.A.}\ \bibnamefont
  {Lima}}, \bibinfo {author} {\bibfnamefont {M.F.}\ \bibnamefont {Silva}},
  \bibinfo {author} {\bibfnamefont {L.N.}\ \bibnamefont {Oliveira}}, \ and\
  \bibinfo {author} {\bibfnamefont {K.}~\bibnamefont {Capelle}},\ }\bibfield
  {title} {\enquote {\bibinfo {title} {{Density functionals not based on the
  electron gas: local-density approximation for a Luttinger liquid}},}\
  }\href@noop {} {\bibfield  {journal} {\bibinfo  {journal} {Phys.~Rev.~Lett.}\
  }\textbf {\bibinfo {volume} {90}},\ \bibinfo {pages} {146402} (\bibinfo
  {year} {2003})}\BibitemShut {NoStop}%
\bibitem [{\citenamefont {Fuks}\ and\ \citenamefont
  {Maitra}(2014{\natexlab{a}})}]{fuks2014challenging}%
  \BibitemOpen
  \bibfield  {author} {\bibinfo {author} {\bibfnamefont {Johanna~I}\
  \bibnamefont {Fuks}}\ and\ \bibinfo {author} {\bibfnamefont {Neepa~T}\
  \bibnamefont {Maitra}},\ }\bibfield  {title} {\enquote {\bibinfo {title}
  {Challenging adiabatic time-dependent density functional theory with a
  hubbard dimer: the case of time-resolved long-range charge transfer},}\
  }\href@noop {} {\bibfield  {journal} {\bibinfo  {journal}
  {Phys.~Chem.~Chem.~Phys.}\ }\textbf {\bibinfo {volume} {16}},\ \bibinfo
  {pages} {14504--14513} (\bibinfo {year} {2014}{\natexlab{a}})}\BibitemShut
  {NoStop}%
\bibitem [{\citenamefont {Fuks}\ and\ \citenamefont
  {Maitra}(2014{\natexlab{b}})}]{fuks2014charge}%
  \BibitemOpen
  \bibfield  {author} {\bibinfo {author} {\bibfnamefont {J.I.}\ \bibnamefont
  {Fuks}}\ and\ \bibinfo {author} {\bibfnamefont {N.T.}\ \bibnamefont
  {Maitra}},\ }\bibfield  {title} {\enquote {\bibinfo {title} {{Charge transfer
  in time-dependent density-functional theory: Insights from the asymmetric
  Hubbard dimer}},}\ }\href@noop {} {\bibfield  {journal} {\bibinfo  {journal}
  {Phys.~Rev.~A}\ }\textbf {\bibinfo {volume} {89}},\ \bibinfo {pages} {062502}
  (\bibinfo {year} {2014}{\natexlab{b}})}\BibitemShut {NoStop}%
\bibitem [{\citenamefont {Cohen}\ and\ \citenamefont
  {Mori-S{\'a}nchez}(2016)}]{cohen2016landscape}%
  \BibitemOpen
  \bibfield  {author} {\bibinfo {author} {\bibfnamefont {Aron~J}\ \bibnamefont
  {Cohen}}\ and\ \bibinfo {author} {\bibfnamefont {Paula}\ \bibnamefont
  {Mori-S{\'a}nchez}},\ }\bibfield  {title} {\enquote {\bibinfo {title}
  {Landscape of an exact energy functional},}\ }\href@noop {} {\bibfield
  {journal} {\bibinfo  {journal} {Phys.~Rev.~A}\ }\textbf {\bibinfo {volume}
  {93}},\ \bibinfo {pages} {042511} (\bibinfo {year} {2016})}\BibitemShut
  {NoStop}%
\bibitem [{\citenamefont {Carrascal}\ \emph {et~al.}(2015)\citenamefont
  {Carrascal}, \citenamefont {Ferrer}, \citenamefont {Smith},\ and\
  \citenamefont {Burke}}]{carrascal2015hubbard}%
  \BibitemOpen
  \bibfield  {author} {\bibinfo {author} {\bibfnamefont {DJ}~\bibnamefont
  {Carrascal}}, \bibinfo {author} {\bibfnamefont {Jaime}\ \bibnamefont
  {Ferrer}}, \bibinfo {author} {\bibfnamefont {Justin~C}\ \bibnamefont
  {Smith}}, \ and\ \bibinfo {author} {\bibfnamefont {Kieron}\ \bibnamefont
  {Burke}},\ }\bibfield  {title} {\enquote {\bibinfo {title} {{The Hubbard
  dimer: a density functional case study of a many-body problem}},}\
  }\href@noop {} {\bibfield  {journal} {\bibinfo  {journal} {Journal of
  Physics: Condensed Matter}\ }\textbf {\bibinfo {volume} {27}},\ \bibinfo
  {pages} {393001} (\bibinfo {year} {2015})}\BibitemShut {NoStop}%
\bibitem [{\citenamefont {Carrascal}\ \emph {et~al.}(2018)\citenamefont
  {Carrascal}, \citenamefont {Ferrer}, \citenamefont {Maitra},\ and\
  \citenamefont {Burke}}]{carrascal2018linear}%
  \BibitemOpen
  \bibfield  {author} {\bibinfo {author} {\bibfnamefont {Diego~J}\ \bibnamefont
  {Carrascal}}, \bibinfo {author} {\bibfnamefont {Jaime}\ \bibnamefont
  {Ferrer}}, \bibinfo {author} {\bibfnamefont {Neepa}\ \bibnamefont {Maitra}},
  \ and\ \bibinfo {author} {\bibfnamefont {Kieron}\ \bibnamefont {Burke}},\
  }\bibfield  {title} {\enquote {\bibinfo {title} {{Linear response
  time-dependent density functional theory of the Hubbard dimer}},}\
  }\href@noop {} {\bibfield  {journal} {\bibinfo  {journal} {The European
  Physical Journal B}\ }\textbf {\bibinfo {volume} {91}},\ \bibinfo {pages}
  {142} (\bibinfo {year} {2018})}\BibitemShut {NoStop}%
\bibitem [{\citenamefont {Smith}\ and\ \citenamefont
  {Burke}(2018)}]{smith2018thermal}%
  \BibitemOpen
  \bibfield  {author} {\bibinfo {author} {\bibfnamefont {Justin~C}\
  \bibnamefont {Smith}}\ and\ \bibinfo {author} {\bibfnamefont {Kieron}\
  \bibnamefont {Burke}},\ }\bibfield  {title} {\enquote {\bibinfo {title}
  {Thermal stitching: Combining the advantages of different quantum fermion
  solvers},}\ }\href@noop {} {\bibfield  {journal} {\bibinfo  {journal}
  {Phys.~Rev.~B}\ }\textbf {\bibinfo {volume} {98}},\ \bibinfo {pages} {075148}
  (\bibinfo {year} {2018})}\BibitemShut {NoStop}%
\bibitem [{\citenamefont {Sagredo}\ and\ \citenamefont
  {Burke}(2018)}]{sagredo2018accurate}%
  \BibitemOpen
  \bibfield  {author} {\bibinfo {author} {\bibfnamefont {Francisca}\
  \bibnamefont {Sagredo}}\ and\ \bibinfo {author} {\bibfnamefont {Kieron}\
  \bibnamefont {Burke}},\ }\bibfield  {title} {\enquote {\bibinfo {title}
  {Accurate double excitations from ensemble density functional
  calculations},}\ }\href@noop {} {\bibfield  {journal} {\bibinfo  {journal}
  {J.~Chem.~Phys.}\ }\textbf {\bibinfo {volume} {149}},\ \bibinfo {pages}
  {134103} (\bibinfo {year} {2018})}\BibitemShut {NoStop}%
\bibitem [{\citenamefont {Smith}\ \emph {et~al.}(2016)\citenamefont {Smith},
  \citenamefont {Pribram-Jones},\ and\ \citenamefont {Burke}}]{smith2016exact}%
  \BibitemOpen
  \bibfield  {author} {\bibinfo {author} {\bibfnamefont {Justin~C}\
  \bibnamefont {Smith}}, \bibinfo {author} {\bibfnamefont {Aurora}\
  \bibnamefont {Pribram-Jones}}, \ and\ \bibinfo {author} {\bibfnamefont
  {Kieron}\ \bibnamefont {Burke}},\ }\bibfield  {title} {\enquote {\bibinfo
  {title} {Exact thermal density functional theory for a model system:
  Correlation components and accuracy of the zero-temperature
  exchange-correlation approximation},}\ }\href@noop {} {\bibfield  {journal}
  {\bibinfo  {journal} {Phys.~Rev.~B}\ }\textbf {\bibinfo {volume} {93}},\
  \bibinfo {pages} {245131} (\bibinfo {year} {2016})}\BibitemShut {NoStop}%
\bibitem [{\citenamefont {Herrera}\ \emph {et~al.}(2018)\citenamefont
  {Herrera}, \citenamefont {Zawadzki},\ and\ \citenamefont
  {D’Amico}}]{herrera2018melting}%
  \BibitemOpen
  \bibfield  {author} {\bibinfo {author} {\bibfnamefont {Marcela}\ \bibnamefont
  {Herrera}}, \bibinfo {author} {\bibfnamefont {Krissia}\ \bibnamefont
  {Zawadzki}}, \ and\ \bibinfo {author} {\bibfnamefont {Irene}\ \bibnamefont
  {D’Amico}},\ }\bibfield  {title} {\enquote {\bibinfo {title} {Melting a
  hubbard dimer: benchmarks of ‘alda’for quantum thermodynamics},}\
  }\href@noop {} {\bibfield  {journal} {\bibinfo  {journal} {The European
  Physical Journal B}\ }\textbf {\bibinfo {volume} {91}},\ \bibinfo {pages}
  {248} (\bibinfo {year} {2018})}\BibitemShut {NoStop}%
\bibitem [{\citenamefont {Vigor}\ \emph {et~al.}(2015)\citenamefont {Vigor},
  \citenamefont {Spencer}, \citenamefont {Bearpark},\ and\ \citenamefont
  {Thom}}]{vigor2015minimising}%
  \BibitemOpen
  \bibfield  {author} {\bibinfo {author} {\bibfnamefont {WA}~\bibnamefont
  {Vigor}}, \bibinfo {author} {\bibfnamefont {JS}~\bibnamefont {Spencer}},
  \bibinfo {author} {\bibfnamefont {MJ}~\bibnamefont {Bearpark}}, \ and\
  \bibinfo {author} {\bibfnamefont {AJW}\ \bibnamefont {Thom}},\ }\bibfield
  {title} {\enquote {\bibinfo {title} {Minimising biases in full configuration
  interaction quantum monte carlo},}\ }\href@noop {} {\bibfield  {journal}
  {\bibinfo  {journal} {J.~Chem.~Phys.}\ }\textbf {\bibinfo {volume} {142}},\
  \bibinfo {pages} {104101} (\bibinfo {year} {2015})}\BibitemShut {NoStop}%
\bibitem [{\citenamefont {McClean}\ \emph {et~al.}(2014)\citenamefont
  {McClean}, \citenamefont {Babbush}, \citenamefont {Love},\ and\ \citenamefont
  {Aspuru-Guzik}}]{mcclean2014exploiting}%
  \BibitemOpen
  \bibfield  {author} {\bibinfo {author} {\bibfnamefont {Jarrod~R}\
  \bibnamefont {McClean}}, \bibinfo {author} {\bibfnamefont {Ryan}\
  \bibnamefont {Babbush}}, \bibinfo {author} {\bibfnamefont {Peter~J}\
  \bibnamefont {Love}}, \ and\ \bibinfo {author} {\bibfnamefont {Al{\'a}n}\
  \bibnamefont {Aspuru-Guzik}},\ }\bibfield  {title} {\enquote {\bibinfo
  {title} {Exploiting locality in quantum computation for quantum chemistry},}\
  }\href@noop {} {\bibfield  {journal} {\bibinfo  {journal} {The journal of
  physical chemistry letters}\ }\textbf {\bibinfo {volume} {5}},\ \bibinfo
  {pages} {4368--4380} (\bibinfo {year} {2014})}\BibitemShut {NoStop}%
\bibitem [{\citenamefont {Schollw{\"o}ck}(2011)}]{schollwock2011density}%
  \BibitemOpen
  \bibfield  {author} {\bibinfo {author} {\bibfnamefont {Ulrich}\ \bibnamefont
  {Schollw{\"o}ck}},\ }\bibfield  {title} {\enquote {\bibinfo {title} {The
  density-matrix renormalization group in the age of matrix product states},}\
  }\href@noop {} {\bibfield  {journal} {\bibinfo  {journal} {Annals of
  Physics}\ }\textbf {\bibinfo {volume} {326}},\ \bibinfo {pages} {96--192}
  (\bibinfo {year} {2011})}\BibitemShut {NoStop}%
\bibitem [{\citenamefont {Chan}\ and\ \citenamefont
  {Sharma}(2011)}]{chan2011density}%
  \BibitemOpen
  \bibfield  {author} {\bibinfo {author} {\bibfnamefont {Garnet Kin-Lic}\
  \bibnamefont {Chan}}\ and\ \bibinfo {author} {\bibfnamefont {Sandeep}\
  \bibnamefont {Sharma}},\ }\bibfield  {title} {\enquote {\bibinfo {title} {The
  density matrix renormalization group in quantum chemistry},}\ }\href@noop {}
  {\bibfield  {journal} {\bibinfo  {journal} {Annual review of physical
  chemistry}\ }\textbf {\bibinfo {volume} {62}},\ \bibinfo {pages} {465--481}
  (\bibinfo {year} {2011})}\BibitemShut {NoStop}%
\bibitem [{\citenamefont {Yang}\ \emph {et~al.}(2014)\citenamefont {Yang},
  \citenamefont {Hu}, \citenamefont {Usvyat}, \citenamefont {Matthews},
  \citenamefont {Sch{\"u}tz},\ and\ \citenamefont {Chan}}]{yang2014ab}%
  \BibitemOpen
  \bibfield  {author} {\bibinfo {author} {\bibfnamefont {Jun}\ \bibnamefont
  {Yang}}, \bibinfo {author} {\bibfnamefont {Weifeng}\ \bibnamefont {Hu}},
  \bibinfo {author} {\bibfnamefont {Denis}\ \bibnamefont {Usvyat}}, \bibinfo
  {author} {\bibfnamefont {Devin}\ \bibnamefont {Matthews}}, \bibinfo {author}
  {\bibfnamefont {Martin}\ \bibnamefont {Sch{\"u}tz}}, \ and\ \bibinfo {author}
  {\bibfnamefont {Garnet Kin-Lic}\ \bibnamefont {Chan}},\ }\bibfield  {title}
  {\enquote {\bibinfo {title} {Ab initio determination of the crystalline
  benzene lattice energy to sub-kilojoule/mole accuracy},}\ }\href@noop {}
  {\bibfield  {journal} {\bibinfo  {journal} {Science}\ }\textbf {\bibinfo
  {volume} {345}},\ \bibinfo {pages} {640--643} (\bibinfo {year}
  {2014})}\BibitemShut {NoStop}%
\bibitem [{\citenamefont {Arodola}\ and\ \citenamefont
  {Soliman}(2017)}]{arodola2017quantum}%
  \BibitemOpen
  \bibfield  {author} {\bibinfo {author} {\bibfnamefont {Olayide~A}\
  \bibnamefont {Arodola}}\ and\ \bibinfo {author} {\bibfnamefont {Mahmoud~ES}\
  \bibnamefont {Soliman}},\ }\bibfield  {title} {\enquote {\bibinfo {title}
  {Quantum mechanics implementation in drug-design workflows: does it really
  help?}}\ }\href@noop {} {\bibfield  {journal} {\bibinfo  {journal} {Drug
  design, development and therapy}\ }\textbf {\bibinfo {volume} {11}},\
  \bibinfo {pages} {2551} (\bibinfo {year} {2017})}\BibitemShut {NoStop}%
\bibitem [{\citenamefont {Foulkes}\ \emph {et~al.}(2001)\citenamefont
  {Foulkes}, \citenamefont {Mitas}, \citenamefont {Needs},\ and\ \citenamefont
  {Rajagopal}}]{foulkes2001quantum}%
  \BibitemOpen
  \bibfield  {author} {\bibinfo {author} {\bibfnamefont {WMC}\ \bibnamefont
  {Foulkes}}, \bibinfo {author} {\bibfnamefont {Lubos}\ \bibnamefont {Mitas}},
  \bibinfo {author} {\bibfnamefont {RJ}~\bibnamefont {Needs}}, \ and\ \bibinfo
  {author} {\bibfnamefont {G}~\bibnamefont {Rajagopal}},\ }\bibfield  {title}
  {\enquote {\bibinfo {title} {{Quantum Monte Carlo simulations of solids}},}\
  }\href@noop {} {\bibfield  {journal} {\bibinfo  {journal} {Rev.~Mod.~Phys.}\
  }\textbf {\bibinfo {volume} {73}},\ \bibinfo {pages} {33} (\bibinfo {year}
  {2001})}\BibitemShut {NoStop}%
\bibitem [{\citenamefont {Eshuis}\ \emph {et~al.}(2012)\citenamefont {Eshuis},
  \citenamefont {Bates},\ and\ \citenamefont {Furche}}]{eshuis2012electron}%
  \BibitemOpen
  \bibfield  {author} {\bibinfo {author} {\bibfnamefont {Henk}\ \bibnamefont
  {Eshuis}}, \bibinfo {author} {\bibfnamefont {Jefferson~E}\ \bibnamefont
  {Bates}}, \ and\ \bibinfo {author} {\bibfnamefont {Filipp}\ \bibnamefont
  {Furche}},\ }\bibfield  {title} {\enquote {\bibinfo {title} {Electron
  correlation methods based on the random phase approximation},}\ }\href@noop
  {} {\bibfield  {journal} {\bibinfo  {journal} {Theoretical Chemistry
  Accounts}\ }\textbf {\bibinfo {volume} {131}},\ \bibinfo {pages} {1084}
  (\bibinfo {year} {2012})}\BibitemShut {NoStop}%
\bibitem [{\citenamefont {Chen}\ \emph {et~al.}(2017)\citenamefont {Chen},
  \citenamefont {Voora}, \citenamefont {Agee}, \citenamefont {Balasubramani},\
  and\ \citenamefont {Furche}}]{chen2017random}%
  \BibitemOpen
  \bibfield  {author} {\bibinfo {author} {\bibfnamefont {Guo~P}\ \bibnamefont
  {Chen}}, \bibinfo {author} {\bibfnamefont {Vamsee~K}\ \bibnamefont {Voora}},
  \bibinfo {author} {\bibfnamefont {Matthew~M}\ \bibnamefont {Agee}}, \bibinfo
  {author} {\bibfnamefont {Sree~Ganesh}\ \bibnamefont {Balasubramani}}, \ and\
  \bibinfo {author} {\bibfnamefont {Filipp}\ \bibnamefont {Furche}},\
  }\bibfield  {title} {\enquote {\bibinfo {title} {Random-phase approximation
  methods},}\ }\href@noop {} {\bibfield  {journal} {\bibinfo  {journal} {Annual
  review of physical chemistry}\ }\textbf {\bibinfo {volume} {68}},\ \bibinfo
  {pages} {421--445} (\bibinfo {year} {2017})}\BibitemShut {NoStop}%
\bibitem [{\citenamefont {Helgaker}\ \emph {et~al.}(2014)\citenamefont
  {Helgaker}, \citenamefont {Jorgensen},\ and\ \citenamefont
  {Olsen}}]{helgaker2014molecular}%
  \BibitemOpen
  \bibfield  {author} {\bibinfo {author} {\bibfnamefont {Trygve}\ \bibnamefont
  {Helgaker}}, \bibinfo {author} {\bibfnamefont {Poul}\ \bibnamefont
  {Jorgensen}}, \ and\ \bibinfo {author} {\bibfnamefont {Jeppe}\ \bibnamefont
  {Olsen}},\ }\href@noop {} {\emph {\bibinfo {title} {Molecular
  electronic-structure theory}}}\ (\bibinfo  {publisher} {John Wiley \& Sons},\
  \bibinfo {year} {2014})\BibitemShut {NoStop}%
\bibitem [{\citenamefont {Pople}(1999)}]{pople1999nobel}%
  \BibitemOpen
  \bibfield  {author} {\bibinfo {author} {\bibfnamefont {John~A}\ \bibnamefont
  {Pople}},\ }\bibfield  {title} {\enquote {\bibinfo {title} {Nobel lecture:
  Quantum chemical models},}\ }\href@noop {} {\bibfield  {journal} {\bibinfo
  {journal} {Rev.~Mod.~Phys.}\ }\textbf {\bibinfo {volume} {71}},\ \bibinfo
  {pages} {1267} (\bibinfo {year} {1999})}\BibitemShut {NoStop}%
\bibitem [{\citenamefont {Bartlett}\ and\ \citenamefont
  {Musia{\l}}(2007)}]{bartlett2007coupled}%
  \BibitemOpen
  \bibfield  {author} {\bibinfo {author} {\bibfnamefont {Rodney~J}\
  \bibnamefont {Bartlett}}\ and\ \bibinfo {author} {\bibfnamefont {Monika}\
  \bibnamefont {Musia{\l}}},\ }\bibfield  {title} {\enquote {\bibinfo {title}
  {Coupled-cluster theory in quantum chemistry},}\ }\href@noop {} {\bibfield
  {journal} {\bibinfo  {journal} {Rev.~Mod.~Phys.}\ }\textbf {\bibinfo {volume}
  {79}},\ \bibinfo {pages} {291} (\bibinfo {year} {2007})}\BibitemShut
  {NoStop}%
\bibitem [{\citenamefont {Babbush}\ \emph {et~al.}(2018)\citenamefont
  {Babbush}, \citenamefont {Wiebe}, \citenamefont {McClean}, \citenamefont
  {McClain}, \citenamefont {Neven},\ and\ \citenamefont
  {Chan}}]{babbush2018low}%
  \BibitemOpen
  \bibfield  {author} {\bibinfo {author} {\bibfnamefont {Ryan}\ \bibnamefont
  {Babbush}}, \bibinfo {author} {\bibfnamefont {Nathan}\ \bibnamefont {Wiebe}},
  \bibinfo {author} {\bibfnamefont {Jarrod}\ \bibnamefont {McClean}}, \bibinfo
  {author} {\bibfnamefont {James}\ \bibnamefont {McClain}}, \bibinfo {author}
  {\bibfnamefont {Hartmut}\ \bibnamefont {Neven}}, \ and\ \bibinfo {author}
  {\bibfnamefont {Garnet~K.L.}\ \bibnamefont {Chan}},\ }\bibfield  {title}
  {\enquote {\bibinfo {title} {Low-depth quantum simulation of materials},}\
  }\href@noop {} {\bibfield  {journal} {\bibinfo  {journal} {Phys.~Rev.~X}\
  }\textbf {\bibinfo {volume} {8}},\ \bibinfo {pages} {011044} (\bibinfo {year}
  {2018})}\BibitemShut {NoStop}%
\bibitem [{\citenamefont {Kimball}(1975)}]{kimball1975short}%
  \BibitemOpen
  \bibfield  {author} {\bibinfo {author} {\bibfnamefont {JC}~\bibnamefont
  {Kimball}},\ }\bibfield  {title} {\enquote {\bibinfo {title} {Short-range
  correlations and the structure factor and momentum distribution of
  electrons},}\ }\href@noop {} {\bibfield  {journal} {\bibinfo  {journal}
  {Journal of Physics A: Mathematical and General}\ }\textbf {\bibinfo {volume}
  {8}},\ \bibinfo {pages} {1513} (\bibinfo {year} {1975})}\BibitemShut
  {NoStop}%
\bibitem [{\citenamefont {Klahn}\ and\ \citenamefont
  {Morgan~III}(1984)}]{klahn1984rates}%
  \BibitemOpen
  \bibfield  {author} {\bibinfo {author} {\bibfnamefont {Bruno}\ \bibnamefont
  {Klahn}}\ and\ \bibinfo {author} {\bibfnamefont {John~D}\ \bibnamefont
  {Morgan~III}},\ }\bibfield  {title} {\enquote {\bibinfo {title} {Rates of
  convergence of variational calculations and of expectation values},}\
  }\href@noop {} {\bibfield  {journal} {\bibinfo  {journal} {J.~Chem.~Phys.}\
  }\textbf {\bibinfo {volume} {81}},\ \bibinfo {pages} {410--433} (\bibinfo
  {year} {1984})}\BibitemShut {NoStop}%
\bibitem [{\citenamefont {Hill}(1985)}]{hill1985rates}%
  \BibitemOpen
  \bibfield  {author} {\bibinfo {author} {\bibfnamefont {Robert~Nyden}\
  \bibnamefont {Hill}},\ }\bibfield  {title} {\enquote {\bibinfo {title}
  {{Rates of convergence and error estimation formulas for the Rayleigh--Ritz
  variational method}},}\ }\href@noop {} {\bibfield  {journal} {\bibinfo
  {journal} {J.~Chem.~Phys.}\ }\textbf {\bibinfo {volume} {83}},\ \bibinfo
  {pages} {1173--1196} (\bibinfo {year} {1985})}\BibitemShut {NoStop}%
\bibitem [{\citenamefont {Baker}\ \emph {et~al.}(2018)\citenamefont {Baker},
  \citenamefont {Burke},\ and\ \citenamefont {White}}]{bakerPRB18}%
  \BibitemOpen
  \bibfield  {author} {\bibinfo {author} {\bibfnamefont {Thomas~E.}\
  \bibnamefont {Baker}}, \bibinfo {author} {\bibfnamefont {Kieron}\
  \bibnamefont {Burke}}, \ and\ \bibinfo {author} {\bibfnamefont {Steven~R.}\
  \bibnamefont {White}},\ }\bibfield  {title} {\enquote {\bibinfo {title}
  {Accurate correlation energies in one-dimensional systems from small
  system-adapted basis functions},}\ }\href {\doibase
  10.1103/PhysRevB.97.085139} {\bibfield  {journal} {\bibinfo  {journal}
  {Phys.~Rev.~B}\ }\textbf {\bibinfo {volume} {97}},\ \bibinfo {pages} {085139}
  (\bibinfo {year} {2018})}\BibitemShut {NoStop}%
\bibitem [{\citenamefont {Beylkin}\ and\ \citenamefont
  {Mohlenkamp}(2002)}]{beylkin2002numerical}%
  \BibitemOpen
  \bibfield  {author} {\bibinfo {author} {\bibfnamefont {Gregory}\ \bibnamefont
  {Beylkin}}\ and\ \bibinfo {author} {\bibfnamefont {Martin~J}\ \bibnamefont
  {Mohlenkamp}},\ }\bibfield  {title} {\enquote {\bibinfo {title} {Numerical
  operator calculus in higher dimensions},}\ }\href@noop {} {\bibfield
  {journal} {\bibinfo  {journal} {Proceedings of the National Academy of
  Sciences}\ }\textbf {\bibinfo {volume} {99}},\ \bibinfo {pages}
  {10246--10251} (\bibinfo {year} {2002})}\BibitemShut {NoStop}%
\bibitem [{\citenamefont {Oliveira}\ and\ \citenamefont
  {Terhal}(2008)}]{oliveira2005complexity}%
  \BibitemOpen
  \bibfield  {author} {\bibinfo {author} {\bibfnamefont {Roberto}\ \bibnamefont
  {Oliveira}}\ and\ \bibinfo {author} {\bibfnamefont {Barbara~M}\ \bibnamefont
  {Terhal}},\ }\bibfield  {title} {\enquote {\bibinfo {title} {The complexity
  of quantum spin systems on a two-dimensional square lattice},}\ }\href@noop
  {} {\bibfield  {journal} {\bibinfo  {journal} {Quant.~Inf.~Comput.}\ }\textbf
  {\bibinfo {volume} {8}},\ \bibinfo {pages} {900--924} (\bibinfo {year}
  {2008})}\BibitemShut {NoStop}%
\bibitem [{\citenamefont {McArdle}\ \emph {et~al.}(2019)\citenamefont
  {McArdle}, \citenamefont {Jones}, \citenamefont {Endo}, \citenamefont {Li},
  \citenamefont {Benjamin},\ and\ \citenamefont
  {Yuan}}]{mcardle2019variational}%
  \BibitemOpen
  \bibfield  {author} {\bibinfo {author} {\bibfnamefont {Sam}\ \bibnamefont
  {McArdle}}, \bibinfo {author} {\bibfnamefont {Tyson}\ \bibnamefont {Jones}},
  \bibinfo {author} {\bibfnamefont {Suguru}\ \bibnamefont {Endo}}, \bibinfo
  {author} {\bibfnamefont {Ying}\ \bibnamefont {Li}}, \bibinfo {author}
  {\bibfnamefont {Simon~C}\ \bibnamefont {Benjamin}}, \ and\ \bibinfo {author}
  {\bibfnamefont {Xiao}\ \bibnamefont {Yuan}},\ }\bibfield  {title} {\enquote
  {\bibinfo {title} {Variational ansatz-based quantum simulation of imaginary
  time evolution},}\ }\href@noop {} {\bibfield  {journal} {\bibinfo  {journal}
  {npj Quantum Information}\ }\textbf {\bibinfo {volume} {5}},\ \bibinfo
  {pages} {1--6} (\bibinfo {year} {2019})}\BibitemShut {NoStop}%
\bibitem [{\citenamefont {Hackl}\ and\ \citenamefont
  {Kehrein}(2008)}]{hackl2008real}%
  \BibitemOpen
  \bibfield  {author} {\bibinfo {author} {\bibfnamefont {A}~\bibnamefont
  {Hackl}}\ and\ \bibinfo {author} {\bibfnamefont {S}~\bibnamefont {Kehrein}},\
  }\bibfield  {title} {\enquote {\bibinfo {title} {Real time evolution in
  quantum many-body systems with unitary perturbation theory},}\ }\href@noop {}
  {\bibfield  {journal} {\bibinfo  {journal} {Phys.~Rev.~B}\ }\textbf {\bibinfo
  {volume} {78}},\ \bibinfo {pages} {092303} (\bibinfo {year}
  {2008})}\BibitemShut {NoStop}%
\bibitem [{\citenamefont {Schwarz}\ \emph {et~al.}(2012)\citenamefont
  {Schwarz}, \citenamefont {Temme},\ and\ \citenamefont
  {Verstraete}}]{schwarz2012preparing}%
  \BibitemOpen
  \bibfield  {author} {\bibinfo {author} {\bibfnamefont {Martin}\ \bibnamefont
  {Schwarz}}, \bibinfo {author} {\bibfnamefont {Kristan}\ \bibnamefont
  {Temme}}, \ and\ \bibinfo {author} {\bibfnamefont {Frank}\ \bibnamefont
  {Verstraete}},\ }\bibfield  {title} {\enquote {\bibinfo {title} {Preparing
  projected entangled pair states on a quantum computer},}\ }\href@noop {}
  {\bibfield  {journal} {\bibinfo  {journal} {Phys.~Rev.~Lett.}\ }\textbf
  {\bibinfo {volume} {108}},\ \bibinfo {pages} {110502} (\bibinfo {year}
  {2012})}\BibitemShut {NoStop}%
\bibitem [{\citenamefont {Gily{\'e}n}\ \emph
  {et~al.}(2019{\natexlab{b}})\citenamefont {Gily{\'e}n}, \citenamefont {Su},
  \citenamefont {Low},\ and\ \citenamefont {Wiebe}}]{gilyen2019quantum}%
  \BibitemOpen
  \bibfield  {author} {\bibinfo {author} {\bibfnamefont {Andr{\'a}s}\
  \bibnamefont {Gily{\'e}n}}, \bibinfo {author} {\bibfnamefont {Yuan}\
  \bibnamefont {Su}}, \bibinfo {author} {\bibfnamefont {Guang~Hao}\
  \bibnamefont {Low}}, \ and\ \bibinfo {author} {\bibfnamefont {Nathan}\
  \bibnamefont {Wiebe}},\ }\bibfield  {title} {\enquote {\bibinfo {title}
  {Quantum singular value transformation and beyond: exponential improvements
  for quantum matrix arithmetics},}\ }in\ \href@noop {} {\emph {\bibinfo
  {booktitle} {Proceedings of the 51st Annual ACM SIGACT Symposium on Theory of
  Computing}}}\ (\bibinfo {year} {2019})\ pp.\ \bibinfo {pages}
  {193--204}\BibitemShut {NoStop}%
\bibitem [{\citenamefont {Yanofsky}\ and\ \citenamefont
  {Mannucci}(2008)}]{yanofsky2008quantum}%
  \BibitemOpen
  \bibfield  {author} {\bibinfo {author} {\bibfnamefont {Noson~S}\ \bibnamefont
  {Yanofsky}}\ and\ \bibinfo {author} {\bibfnamefont {Mirco~A}\ \bibnamefont
  {Mannucci}},\ }\href@noop {} {\emph {\bibinfo {title} {Quantum computing for
  computer scientists}}}\ (\bibinfo  {publisher} {Cambridge University Press},\
  \bibinfo {year} {2008})\BibitemShut {NoStop}%
\bibitem [{\citenamefont {Liu}\ \emph {et~al.}(2007)\citenamefont {Liu},
  \citenamefont {Christandl},\ and\ \citenamefont
  {Verstraete}}]{liu2007quantum}%
  \BibitemOpen
  \bibfield  {author} {\bibinfo {author} {\bibfnamefont {Yi-Kai}\ \bibnamefont
  {Liu}}, \bibinfo {author} {\bibfnamefont {Matthias}\ \bibnamefont
  {Christandl}}, \ and\ \bibinfo {author} {\bibfnamefont {Frank}\ \bibnamefont
  {Verstraete}},\ }\bibfield  {title} {\enquote {\bibinfo {title} {Quantum
  computational complexity of the n-representability problem: Qma complete},}\
  }\href@noop {} {\bibfield  {journal} {\bibinfo  {journal} {Physical review
  letters}\ }\textbf {\bibinfo {volume} {98}},\ \bibinfo {pages} {110503}
  (\bibinfo {year} {2007})}\BibitemShut {NoStop}%
\bibitem [{\citenamefont {Schuch}\ and\ \citenamefont
  {Verstraete}(2009)}]{schuch2009computational}%
  \BibitemOpen
  \bibfield  {author} {\bibinfo {author} {\bibfnamefont {Norbert}\ \bibnamefont
  {Schuch}}\ and\ \bibinfo {author} {\bibfnamefont {Frank}\ \bibnamefont
  {Verstraete}},\ }\bibfield  {title} {\enquote {\bibinfo {title}
  {Computational complexity of interacting electrons and fundamental
  limitations of density functional theory},}\ }\href@noop {} {\bibfield
  {journal} {\bibinfo  {journal} {Nature Physics}\ }\textbf {\bibinfo {volume}
  {5}},\ \bibinfo {pages} {732--735} (\bibinfo {year} {2009})}\BibitemShut
  {NoStop}%
\bibitem [{\citenamefont {Szegedy}(2004)}]{szegedy2004quantum}%
  \BibitemOpen
  \bibfield  {author} {\bibinfo {author} {\bibfnamefont {Mario}\ \bibnamefont
  {Szegedy}},\ }\bibfield  {title} {\enquote {\bibinfo {title} {Quantum
  speed-up of markov chain based algorithms},}\ }in\ \href@noop {} {\emph
  {\bibinfo {booktitle} {45th Annual IEEE symposium on foundations of computer
  science}}}\ (\bibinfo {organization} {IEEE},\ \bibinfo {year} {2004})\ pp.\
  \bibinfo {pages} {32--41}\BibitemShut {NoStop}%
\bibitem [{\citenamefont {Lemieux}\ \emph
  {et~al.}(2020{\natexlab{b}})\citenamefont {Lemieux}, \citenamefont {Heim},
  \citenamefont {Poulin}, \citenamefont {Svore},\ and\ \citenamefont
  {Troyer}}]{lemieux2020efficient}%
  \BibitemOpen
  \bibfield  {author} {\bibinfo {author} {\bibfnamefont {Jessica}\ \bibnamefont
  {Lemieux}}, \bibinfo {author} {\bibfnamefont {Bettina}\ \bibnamefont {Heim}},
  \bibinfo {author} {\bibfnamefont {David}\ \bibnamefont {Poulin}}, \bibinfo
  {author} {\bibfnamefont {Krysta}\ \bibnamefont {Svore}}, \ and\ \bibinfo
  {author} {\bibfnamefont {Matthias}\ \bibnamefont {Troyer}},\ }\bibfield
  {title} {\enquote {\bibinfo {title} {Efficient quantum walk circuits for
  metropolis-hastings algorithm},}\ }\href@noop {} {\bibfield  {journal}
  {\bibinfo  {journal} {Quantum}\ }\textbf {\bibinfo {volume} {4}},\ \bibinfo
  {pages} {287} (\bibinfo {year} {2020}{\natexlab{b}})}\BibitemShut {NoStop}%
\bibitem [{\citenamefont {Peruzzo}\ \emph {et~al.}(2014)\citenamefont
  {Peruzzo}, \citenamefont {McClean}, \citenamefont {Shadbolt}, \citenamefont
  {Yung}, \citenamefont {Zhou}, \citenamefont {Love}, \citenamefont
  {Aspuru-Guzik},\ and\ \citenamefont {O'brien}}]{peruzzo2014variational}%
  \BibitemOpen
  \bibfield  {author} {\bibinfo {author} {\bibfnamefont {Alberto}\ \bibnamefont
  {Peruzzo}}, \bibinfo {author} {\bibfnamefont {Jarrod}\ \bibnamefont
  {McClean}}, \bibinfo {author} {\bibfnamefont {Peter}\ \bibnamefont
  {Shadbolt}}, \bibinfo {author} {\bibfnamefont {Man-Hong}\ \bibnamefont
  {Yung}}, \bibinfo {author} {\bibfnamefont {Xiao-Qi}\ \bibnamefont {Zhou}},
  \bibinfo {author} {\bibfnamefont {Peter~J}\ \bibnamefont {Love}}, \bibinfo
  {author} {\bibfnamefont {Al{\'a}n}\ \bibnamefont {Aspuru-Guzik}}, \ and\
  \bibinfo {author} {\bibfnamefont {Jeremy~L}\ \bibnamefont {O'brien}},\
  }\bibfield  {title} {\enquote {\bibinfo {title} {A variational eigenvalue
  solver on a photonic quantum processor},}\ }\href@noop {} {\bibfield
  {journal} {\bibinfo  {journal} {Nature communications}\ }\textbf {\bibinfo
  {volume} {5}},\ \bibinfo {pages} {4213} (\bibinfo {year} {2014})}\BibitemShut
  {NoStop}%
\bibitem [{\citenamefont {Rist{\`e}}\ \emph {et~al.}(2017)\citenamefont
  {Rist{\`e}}, \citenamefont {Da~Silva}, \citenamefont {Ryan}, \citenamefont
  {Cross}, \citenamefont {C{\'o}rcoles}, \citenamefont {Smolin}, \citenamefont
  {Gambetta}, \citenamefont {Chow},\ and\ \citenamefont
  {Johnson}}]{riste2017demonstration}%
  \BibitemOpen
  \bibfield  {author} {\bibinfo {author} {\bibfnamefont {Diego}\ \bibnamefont
  {Rist{\`e}}}, \bibinfo {author} {\bibfnamefont {Marcus~P}\ \bibnamefont
  {Da~Silva}}, \bibinfo {author} {\bibfnamefont {Colm~A}\ \bibnamefont {Ryan}},
  \bibinfo {author} {\bibfnamefont {Andrew~W}\ \bibnamefont {Cross}}, \bibinfo
  {author} {\bibfnamefont {Antonio~D}\ \bibnamefont {C{\'o}rcoles}}, \bibinfo
  {author} {\bibfnamefont {John~A}\ \bibnamefont {Smolin}}, \bibinfo {author}
  {\bibfnamefont {Jay~M}\ \bibnamefont {Gambetta}}, \bibinfo {author}
  {\bibfnamefont {Jerry~M}\ \bibnamefont {Chow}}, \ and\ \bibinfo {author}
  {\bibfnamefont {Blake~R}\ \bibnamefont {Johnson}},\ }\bibfield  {title}
  {\enquote {\bibinfo {title} {Demonstration of quantum advantage in machine
  learning},}\ }\href@noop {} {\bibfield  {journal} {\bibinfo  {journal} {npj
  Quantum Information}\ }\textbf {\bibinfo {volume} {3}},\ \bibinfo {pages}
  {16} (\bibinfo {year} {2017})}\BibitemShut {NoStop}%
\bibitem [{\citenamefont {Arunachalam}\ and\ \citenamefont
  {De~Wolf}(2018)}]{arunachalam2018optimal}%
  \BibitemOpen
  \bibfield  {author} {\bibinfo {author} {\bibfnamefont {Srinivasan}\
  \bibnamefont {Arunachalam}}\ and\ \bibinfo {author} {\bibfnamefont {Ronald}\
  \bibnamefont {De~Wolf}},\ }\bibfield  {title} {\enquote {\bibinfo {title}
  {Optimal quantum sample complexity of learning algorithms},}\ }\href@noop {}
  {\bibfield  {journal} {\bibinfo  {journal} {The Journal of Machine Learning
  Research}\ }\textbf {\bibinfo {volume} {19}},\ \bibinfo {pages} {2879--2878}
  (\bibinfo {year} {2018})}\BibitemShut {NoStop}%
\bibitem [{Note1()}]{Note1}%
  \BibitemOpen
  \bibinfo {note} {Theorems listed here are so heavily dependent on
  pre-existing proofs that they are probably more accurately called a corollary
  of those theorems, but this naming convention is used in several physics
  works and here to match.}\BibitemShut {Stop}%
\bibitem [{\citenamefont {Whitfield}\ \emph {et~al.}(2013)\citenamefont
  {Whitfield}, \citenamefont {Love},\ and\ \citenamefont
  {Aspuru-Guzik}}]{whitfield2013computational}%
  \BibitemOpen
  \bibfield  {author} {\bibinfo {author} {\bibfnamefont {James~Daniel}\
  \bibnamefont {Whitfield}}, \bibinfo {author} {\bibfnamefont {Peter~John}\
  \bibnamefont {Love}}, \ and\ \bibinfo {author} {\bibfnamefont {Alan}\
  \bibnamefont {Aspuru-Guzik}},\ }\bibfield  {title} {\enquote {\bibinfo
  {title} {Computational complexity in electronic structure},}\ }\href@noop {}
  {\bibfield  {journal} {\bibinfo  {journal} {Physical Chemistry Chemical
  Physics}\ }\textbf {\bibinfo {volume} {15}},\ \bibinfo {pages} {397--411}
  (\bibinfo {year} {2013})}\BibitemShut {NoStop}%
\bibitem [{\citenamefont {Valiant}(1984)}]{valiant1984theory}%
  \BibitemOpen
  \bibfield  {author} {\bibinfo {author} {\bibfnamefont {Leslie~G}\
  \bibnamefont {Valiant}},\ }\bibfield  {title} {\enquote {\bibinfo {title} {A
  theory of the learnable},}\ }\href@noop {} {\bibfield  {journal} {\bibinfo
  {journal} {Communications of the ACM}\ }\textbf {\bibinfo {volume} {27}},\
  \bibinfo {pages} {1134--1142} (\bibinfo {year} {1984})}\BibitemShut {NoStop}%
\bibitem [{\citenamefont {Amsterdam}(1988)}]{amsterdam1988valiant}%
  \BibitemOpen
  \bibfield  {author} {\bibinfo {author} {\bibfnamefont {Jonathan~Blair}\
  \bibnamefont {Amsterdam}},\ }\emph {\bibinfo {title} {The valiant learning
  model: Extensions and assessment}},\ \href@noop {} {Ph.D. thesis},\ \bibinfo
  {school} {Massachusetts Institute of Technology, Department of Electrical
  Engineering} (\bibinfo {year} {1988})\BibitemShut {NoStop}%
\bibitem [{\citenamefont {Bergadano}\ and\ \citenamefont
  {Saitta}(1989)}]{bergadano1989error}%
  \BibitemOpen
  \bibfield  {author} {\bibinfo {author} {\bibfnamefont {F}~\bibnamefont
  {Bergadano}}\ and\ \bibinfo {author} {\bibfnamefont {L}~\bibnamefont
  {Saitta}},\ }\bibfield  {title} {\enquote {\bibinfo {title} {On the error
  probability of boolean concept descriptions},}\ }in\ \href@noop {} {\emph
  {\bibinfo {booktitle} {Proceedings of the 1989 European Working Session on
  Learning}}}\ (\bibinfo {year} {1989})\ pp.\ \bibinfo {pages}
  {25--35}\BibitemShut {NoStop}%
\bibitem [{\citenamefont {Haussler}(2017)}]{haussler1990probably}%
  \BibitemOpen
  \bibfield  {author} {\bibinfo {author} {\bibfnamefont {David}\ \bibnamefont
  {Haussler}},\ }\enquote {\bibinfo {title} {Probably approximately correct
  learning},}\ in\ \href {\doibase 10.1007/978-1-4899-7687-1_100377} {\emph
  {\bibinfo {booktitle} {Encyclopedia of Machine Learning and Data Mining}}},\
  \bibinfo {editor} {edited by\ \bibinfo {editor} {\bibfnamefont {Claude}\
  \bibnamefont {Sammut}}\ and\ \bibinfo {editor} {\bibfnamefont {Geoffrey~I.}\
  \bibnamefont {Webb}}}\ (\bibinfo  {publisher} {Springer US},\ \bibinfo
  {address} {Boston, MA},\ \bibinfo {year} {2017})\ pp.\ \bibinfo {pages}
  {1017--1017}\BibitemShut {NoStop}%
\bibitem [{\citenamefont {Buntine}(1990)}]{buntine1990theory}%
  \BibitemOpen
  \bibfield  {author} {\bibinfo {author} {\bibfnamefont {Wray~Lindsay}\
  \bibnamefont {Buntine}},\ }\emph {\bibinfo {title} {A theory of learning
  classification rules}},\ \href@noop {} {Ph.D. thesis},\ \bibinfo  {school}
  {University of Technology, Sydney} (\bibinfo {year} {1990})\BibitemShut
  {NoStop}%
\bibitem [{\citenamefont {Pazzani}\ and\ \citenamefont
  {Sarrett}(1992)}]{pazzani1992framework}%
  \BibitemOpen
  \bibfield  {author} {\bibinfo {author} {\bibfnamefont {Michael~J}\
  \bibnamefont {Pazzani}}\ and\ \bibinfo {author} {\bibfnamefont {Wendy}\
  \bibnamefont {Sarrett}},\ }\bibfield  {title} {\enquote {\bibinfo {title} {A
  framework for average case analysis of conjunctive learning algorithms},}\
  }\href@noop {} {\bibfield  {journal} {\bibinfo  {journal} {Machine Learning}\
  }\textbf {\bibinfo {volume} {9}},\ \bibinfo {pages} {349--372} (\bibinfo
  {year} {1992})}\BibitemShut {NoStop}%
\bibitem [{\citenamefont {Fetter}\ and\ \citenamefont
  {Walecka}(2012)}]{fetter2012quantum}%
  \BibitemOpen
  \bibfield  {author} {\bibinfo {author} {\bibfnamefont {Alexander~L}\
  \bibnamefont {Fetter}}\ and\ \bibinfo {author} {\bibfnamefont {John~Dirk}\
  \bibnamefont {Walecka}},\ }\href@noop {} {\emph {\bibinfo {title} {Quantum
  theory of many-particle systems}}}\ (\bibinfo  {publisher} {Courier
  Corporation},\ \bibinfo {year} {2012})\BibitemShut {NoStop}%
\bibitem [{\citenamefont {Raimes}(1972)}]{raimes1972many}%
  \BibitemOpen
  \bibfield  {author} {\bibinfo {author} {\bibfnamefont {Stanley}\ \bibnamefont
  {Raimes}},\ }\href@noop {} {\emph {\bibinfo {title} {Many-electron theory}}}\
  (\bibinfo  {publisher} {North-Holland},\ \bibinfo {year} {1972})\BibitemShut
  {NoStop}%
\bibitem [{\citenamefont {Hubbard}(1963)}]{hubbard1963electron}%
  \BibitemOpen
  \bibfield  {author} {\bibinfo {author} {\bibfnamefont {John}\ \bibnamefont
  {Hubbard}},\ }\bibfield  {title} {\enquote {\bibinfo {title} {Electron
  correlations in narrow energy bands},}\ }\href@noop {} {\bibfield  {journal}
  {\bibinfo  {journal} {Proceedings of the Royal Society of London. Series A.
  Mathematical and Physical Sciences}\ }\textbf {\bibinfo {volume} {276}},\
  \bibinfo {pages} {238--257} (\bibinfo {year} {1963})}\BibitemShut {NoStop}%
\bibitem [{\citenamefont {Born}\ and\ \citenamefont
  {Oppenheimer}(1927)}]{born1927quantentheorie}%
  \BibitemOpen
  \bibfield  {author} {\bibinfo {author} {\bibfnamefont {Max}\ \bibnamefont
  {Born}}\ and\ \bibinfo {author} {\bibfnamefont {Robert}\ \bibnamefont
  {Oppenheimer}},\ }\bibfield  {title} {\enquote {\bibinfo {title} {Zur
  quantentheorie der molekeln},}\ }\href@noop {} {\bibfield  {journal}
  {\bibinfo  {journal} {Annalen der physik}\ }\textbf {\bibinfo {volume}
  {389}},\ \bibinfo {pages} {457--484} (\bibinfo {year} {1927})}\BibitemShut
  {NoStop}%
\bibitem [{\citenamefont {Boys}(1950)}]{boys1950electronic}%
  \BibitemOpen
  \bibfield  {author} {\bibinfo {author} {\bibfnamefont {S~Francis}\
  \bibnamefont {Boys}},\ }\bibfield  {title} {\enquote {\bibinfo {title}
  {Electronic wave functions-{I. A} general method of calculation for the
  stationary states of any molecular system},}\ }\href@noop {} {\bibfield
  {journal} {\bibinfo  {journal} {Proceedings of the Royal Society of London.
  Series A. Mathematical and Physical Sciences}\ }\textbf {\bibinfo {volume}
  {200}},\ \bibinfo {pages} {542--554} (\bibinfo {year} {1950})}\BibitemShut
  {NoStop}%
\bibitem [{\citenamefont {Bischoff}\ \emph {et~al.}(2012)\citenamefont
  {Bischoff}, \citenamefont {Harrison},\ and\ \citenamefont
  {Valeev}}]{bischoff2012computing}%
  \BibitemOpen
  \bibfield  {author} {\bibinfo {author} {\bibfnamefont {Florian~A}\
  \bibnamefont {Bischoff}}, \bibinfo {author} {\bibfnamefont {Robert~J}\
  \bibnamefont {Harrison}}, \ and\ \bibinfo {author} {\bibfnamefont {Edward~F}\
  \bibnamefont {Valeev}},\ }\bibfield  {title} {\enquote {\bibinfo {title}
  {{Computing many-body wave functions with guaranteed precision: The
  first-order M{\o}ller-Plesset wave function for the ground state of helium
  atom}},}\ }\href@noop {} {\bibfield  {journal} {\bibinfo  {journal}
  {J.~Chem.~Phys.}\ }\textbf {\bibinfo {volume} {137}},\ \bibinfo {pages}
  {104103} (\bibinfo {year} {2012})}\BibitemShut {NoStop}%
\bibitem [{\citenamefont {Harrison}\ \emph {et~al.}(2016)\citenamefont
  {Harrison}, \citenamefont {Beylkin}, \citenamefont {Bischoff}, \citenamefont
  {Calvin}, \citenamefont {Fann}, \citenamefont {Fosso-Tande}, \citenamefont
  {Galindo}, \citenamefont {Hammond}, \citenamefont {Hartman-Baker},
  \citenamefont {Hill} \emph {et~al.}}]{harrison2016madness}%
  \BibitemOpen
  \bibfield  {author} {\bibinfo {author} {\bibfnamefont {Robert~J}\
  \bibnamefont {Harrison}}, \bibinfo {author} {\bibfnamefont {Gregory}\
  \bibnamefont {Beylkin}}, \bibinfo {author} {\bibfnamefont {Florian~A}\
  \bibnamefont {Bischoff}}, \bibinfo {author} {\bibfnamefont {Justus~A}\
  \bibnamefont {Calvin}}, \bibinfo {author} {\bibfnamefont {George~I}\
  \bibnamefont {Fann}}, \bibinfo {author} {\bibfnamefont {Jacob}\ \bibnamefont
  {Fosso-Tande}}, \bibinfo {author} {\bibfnamefont {Diego}\ \bibnamefont
  {Galindo}}, \bibinfo {author} {\bibfnamefont {Jeff~R}\ \bibnamefont
  {Hammond}}, \bibinfo {author} {\bibfnamefont {Rebecca}\ \bibnamefont
  {Hartman-Baker}}, \bibinfo {author} {\bibfnamefont {Judith~C}\ \bibnamefont
  {Hill}},  \emph {et~al.},\ }\bibfield  {title} {\enquote {\bibinfo {title}
  {{MADNESS: A multiresolution, adaptive numerical environment for scientific
  simulation}},}\ }\href@noop {} {\bibfield  {journal} {\bibinfo  {journal}
  {SIAM Journal on Scientific Computing}\ }\textbf {\bibinfo {volume} {38}},\
  \bibinfo {pages} {S123--S142} (\bibinfo {year} {2016})}\BibitemShut {NoStop}%
\bibitem [{\citenamefont {Reed}\ and\ \citenamefont
  {Simon}(2012)}]{reed2012methods}%
  \BibitemOpen
  \bibfield  {author} {\bibinfo {author} {\bibfnamefont {Michael}\ \bibnamefont
  {Reed}}\ and\ \bibinfo {author} {\bibfnamefont {Barry}\ \bibnamefont
  {Simon}},\ }\href@noop {} {\emph {\bibinfo {title} {Methods of modern
  mathematical physics: Functional analysis}}}\ (\bibinfo  {publisher}
  {Elsevier},\ \bibinfo {year} {2012})\BibitemShut {NoStop}%
\bibitem [{\citenamefont {Levy}\ and\ \citenamefont
  {Perdew}(1985)}]{levy1985constrained}%
  \BibitemOpen
  \bibfield  {author} {\bibinfo {author} {\bibfnamefont {Mel}\ \bibnamefont
  {Levy}}\ and\ \bibinfo {author} {\bibfnamefont {John~P}\ \bibnamefont
  {Perdew}},\ }\bibfield  {title} {\enquote {\bibinfo {title} {The constrained
  search formulation of density functional theory},}\ }in\ \href@noop {} {\emph
  {\bibinfo {booktitle} {Density functional methods in physics}}}\ (\bibinfo
  {publisher} {Springer},\ \bibinfo {year} {1985})\ pp.\ \bibinfo {pages}
  {11--30}\BibitemShut {NoStop}%
\bibitem [{\citenamefont {Fermi}(1927)}]{fermi1927metodo}%
  \BibitemOpen
  \bibfield  {author} {\bibinfo {author} {\bibfnamefont {Enrico}\ \bibnamefont
  {Fermi}},\ }\bibfield  {title} {\enquote {\bibinfo {title} {Un metodo
  statistico per la determinazione di alcune priorieta dell'atome},}\
  }\href@noop {} {\bibfield  {journal} {\bibinfo  {journal} {Rend. Accad. Naz.
  Lincei}\ }\textbf {\bibinfo {volume} {6}},\ \bibinfo {pages} {32} (\bibinfo
  {year} {1927})}\BibitemShut {NoStop}%
\bibitem [{\citenamefont {Thomas}(1927)}]{thomas1927calculation}%
  \BibitemOpen
  \bibfield  {author} {\bibinfo {author} {\bibfnamefont {Llewellyn~H}\
  \bibnamefont {Thomas}},\ }\bibfield  {title} {\enquote {\bibinfo {title} {The
  calculation of atomic fields},}\ }in\ \href@noop {} {\emph {\bibinfo
  {booktitle} {Mathematical Proceedings of the Cambridge Philosophical
  Society}}},\ Vol.~\bibinfo {volume} {23}\ (\bibinfo {organization} {Cambridge
  University Press},\ \bibinfo {year} {1927})\ pp.\ \bibinfo {pages}
  {542--548}\BibitemShut {NoStop}%
\bibitem [{\citenamefont {Weizs{\"a}cker}(1935)}]{weizsacker1935theorie}%
  \BibitemOpen
  \bibfield  {author} {\bibinfo {author} {\bibfnamefont {CF~v}\ \bibnamefont
  {Weizs{\"a}cker}},\ }\bibfield  {title} {\enquote {\bibinfo {title} {Zur
  theorie der kernmassen},}\ }\href@noop {} {\bibfield  {journal} {\bibinfo
  {journal} {Zeitschrift f{\"u}r Physik A Hadrons and Nuclei}\ }\textbf
  {\bibinfo {volume} {96}},\ \bibinfo {pages} {431--458} (\bibinfo {year}
  {1935})}\BibitemShut {NoStop}%
\bibitem [{\citenamefont {Ribeiro}\ \emph {et~al.}(2015)\citenamefont
  {Ribeiro}, \citenamefont {Lee}, \citenamefont {Cangi}, \citenamefont
  {Elliott},\ and\ \citenamefont {Burke}}]{ribeiro2015corrections}%
  \BibitemOpen
  \bibfield  {author} {\bibinfo {author} {\bibfnamefont {Raphael~F}\
  \bibnamefont {Ribeiro}}, \bibinfo {author} {\bibfnamefont {Donghyung}\
  \bibnamefont {Lee}}, \bibinfo {author} {\bibfnamefont {Attila}\ \bibnamefont
  {Cangi}}, \bibinfo {author} {\bibfnamefont {Peter}\ \bibnamefont {Elliott}},
  \ and\ \bibinfo {author} {\bibfnamefont {Kieron}\ \bibnamefont {Burke}},\
  }\bibfield  {title} {\enquote {\bibinfo {title} {{Corrections to Thomas-Fermi
  densities at turning points and beyond}},}\ }\href@noop {} {\bibfield
  {journal} {\bibinfo  {journal} {Phys.~Rev.~Lett.}\ }\textbf {\bibinfo
  {volume} {114}},\ \bibinfo {pages} {050401} (\bibinfo {year}
  {2015})}\BibitemShut {NoStop}%
\bibitem [{\citenamefont {Lieb}\ and\ \citenamefont
  {Oxford}(1981)}]{lieb1981improved}%
  \BibitemOpen
  \bibfield  {author} {\bibinfo {author} {\bibfnamefont {Elliott~H}\
  \bibnamefont {Lieb}}\ and\ \bibinfo {author} {\bibfnamefont {Stephen}\
  \bibnamefont {Oxford}},\ }\bibfield  {title} {\enquote {\bibinfo {title}
  {Improved lower bound on the indirect coulomb energy},}\ }\href@noop {}
  {\bibfield  {journal} {\bibinfo  {journal} {Int.~J.~Quant.~Chem.}\ }\textbf
  {\bibinfo {volume} {19}},\ \bibinfo {pages} {427--439} (\bibinfo {year}
  {1981})}\BibitemShut {NoStop}%
\bibitem [{\citenamefont {Pittalis}\ \emph {et~al.}(2011)\citenamefont
  {Pittalis}, \citenamefont {Proetto}, \citenamefont {Floris}, \citenamefont
  {Sanna}, \citenamefont {Bersier}, \citenamefont {Burke},\ and\ \citenamefont
  {Gross}}]{pittalis2011exact}%
  \BibitemOpen
  \bibfield  {author} {\bibinfo {author} {\bibfnamefont {Stefano}\ \bibnamefont
  {Pittalis}}, \bibinfo {author} {\bibfnamefont {CR}~\bibnamefont {Proetto}},
  \bibinfo {author} {\bibfnamefont {A}~\bibnamefont {Floris}}, \bibinfo
  {author} {\bibfnamefont {A}~\bibnamefont {Sanna}}, \bibinfo {author}
  {\bibfnamefont {C}~\bibnamefont {Bersier}}, \bibinfo {author} {\bibfnamefont
  {K}~\bibnamefont {Burke}}, \ and\ \bibinfo {author} {\bibfnamefont
  {Eberhard~KU}\ \bibnamefont {Gross}},\ }\bibfield  {title} {\enquote
  {\bibinfo {title} {Exact conditions in finite-temperature density-functional
  theory},}\ }\href@noop {} {\bibfield  {journal} {\bibinfo  {journal}
  {Phys.~Rev.~Lett.}\ }\textbf {\bibinfo {volume} {107}},\ \bibinfo {pages}
  {163001} (\bibinfo {year} {2011})}\BibitemShut {NoStop}%
\bibitem [{\citenamefont {Sun}\ \emph {et~al.}(2015)\citenamefont {Sun},
  \citenamefont {Ruzsinszky},\ and\ \citenamefont {Perdew}}]{sun2015strongly}%
  \BibitemOpen
  \bibfield  {author} {\bibinfo {author} {\bibfnamefont {Jianwei}\ \bibnamefont
  {Sun}}, \bibinfo {author} {\bibfnamefont {Adrienn}\ \bibnamefont
  {Ruzsinszky}}, \ and\ \bibinfo {author} {\bibfnamefont {John~P}\ \bibnamefont
  {Perdew}},\ }\bibfield  {title} {\enquote {\bibinfo {title} {Strongly
  constrained and appropriately normed semilocal density functional},}\
  }\href@noop {} {\bibfield  {journal} {\bibinfo  {journal} {Phys.~Rev.~Lett.}\
  }\textbf {\bibinfo {volume} {115}},\ \bibinfo {pages} {036402} (\bibinfo
  {year} {2015})}\BibitemShut {NoStop}%
\bibitem [{\citenamefont {Mori-S{\'a}nchez}\ \emph {et~al.}(2008)\citenamefont
  {Mori-S{\'a}nchez}, \citenamefont {Cohen},\ and\ \citenamefont
  {Yang}}]{mori2008localization}%
  \BibitemOpen
  \bibfield  {author} {\bibinfo {author} {\bibfnamefont {Paula}\ \bibnamefont
  {Mori-S{\'a}nchez}}, \bibinfo {author} {\bibfnamefont {Aron~J}\ \bibnamefont
  {Cohen}}, \ and\ \bibinfo {author} {\bibfnamefont {Weitao}\ \bibnamefont
  {Yang}},\ }\bibfield  {title} {\enquote {\bibinfo {title} {Localization and
  delocalization errors in density functional theory and implications for
  band-gap prediction},}\ }\href@noop {} {\bibfield  {journal} {\bibinfo
  {journal} {Physical review letters}\ }\textbf {\bibinfo {volume} {100}},\
  \bibinfo {pages} {146401} (\bibinfo {year} {2008})}\BibitemShut {NoStop}%
\bibitem [{\citenamefont {Cohen}\ \emph {et~al.}(2008)\citenamefont {Cohen},
  \citenamefont {Mori-S{\'a}nchez},\ and\ \citenamefont
  {Yang}}]{cohen2008insights}%
  \BibitemOpen
  \bibfield  {author} {\bibinfo {author} {\bibfnamefont {Aron~J}\ \bibnamefont
  {Cohen}}, \bibinfo {author} {\bibfnamefont {Paula}\ \bibnamefont
  {Mori-S{\'a}nchez}}, \ and\ \bibinfo {author} {\bibfnamefont {Weitao}\
  \bibnamefont {Yang}},\ }\bibfield  {title} {\enquote {\bibinfo {title}
  {Insights into current limitations of density functional theory},}\
  }\href@noop {} {\bibfield  {journal} {\bibinfo  {journal} {Science}\ }\textbf
  {\bibinfo {volume} {321}},\ \bibinfo {pages} {792--794} (\bibinfo {year}
  {2008})}\BibitemShut {NoStop}%
\bibitem [{\citenamefont {Goldstein}\ \emph {et~al.}(2014)\citenamefont
  {Goldstein}, \citenamefont {Poole},\ and\ \citenamefont
  {Safko}}]{goldstein2014classical}%
  \BibitemOpen
  \bibfield  {author} {\bibinfo {author} {\bibfnamefont {Herbert}\ \bibnamefont
  {Goldstein}}, \bibinfo {author} {\bibfnamefont {Charles~P}\ \bibnamefont
  {Poole}}, \ and\ \bibinfo {author} {\bibfnamefont {John~L}\ \bibnamefont
  {Safko}},\ }\href@noop {} {\emph {\bibinfo {title} {Classical Mechanics}}}\
  (\bibinfo  {publisher} {Pearson Higher Ed},\ \bibinfo {year}
  {2014})\BibitemShut {NoStop}%
\bibitem [{\citenamefont {Gross}\ and\ \citenamefont
  {Maitra}(2012)}]{gross2012introduction}%
  \BibitemOpen
  \bibfield  {author} {\bibinfo {author} {\bibfnamefont {Eberhard~KU}\
  \bibnamefont {Gross}}\ and\ \bibinfo {author} {\bibfnamefont {Neepa~T}\
  \bibnamefont {Maitra}},\ }\bibfield  {title} {\enquote {\bibinfo {title}
  {{Introduction to TDDFT}},}\ }in\ \href@noop {} {\emph {\bibinfo {booktitle}
  {Fundamentals of Time-Dependent Density Functional Theory}}}\ (\bibinfo
  {publisher} {Springer},\ \bibinfo {year} {2012})\ pp.\ \bibinfo {pages}
  {53--99}\BibitemShut {NoStop}%
\bibitem [{\citenamefont {Seidl}(1999)}]{seidl1999strong}%
  \BibitemOpen
  \bibfield  {author} {\bibinfo {author} {\bibfnamefont {Michael}\ \bibnamefont
  {Seidl}},\ }\bibfield  {title} {\enquote {\bibinfo {title}
  {Strong-interaction limit of density-functional theory},}\ }\href@noop {}
  {\bibfield  {journal} {\bibinfo  {journal} {Phys.~Rev.~A}\ }\textbf {\bibinfo
  {volume} {60}},\ \bibinfo {pages} {4387} (\bibinfo {year}
  {1999})}\BibitemShut {NoStop}%
\bibitem [{\citenamefont {Langreth}\ and\ \citenamefont
  {Perdew}(1975)}]{langreth1975exchange}%
  \BibitemOpen
  \bibfield  {author} {\bibinfo {author} {\bibfnamefont {David~C}\ \bibnamefont
  {Langreth}}\ and\ \bibinfo {author} {\bibfnamefont {John~P}\ \bibnamefont
  {Perdew}},\ }\bibfield  {title} {\enquote {\bibinfo {title} {The
  exchange-correlation energy of a metallic surface},}\ }\href@noop {}
  {\bibfield  {journal} {\bibinfo  {journal} {Solid State Communications}\
  }\textbf {\bibinfo {volume} {17}},\ \bibinfo {pages} {1425--1429} (\bibinfo
  {year} {1975})}\BibitemShut {NoStop}%
\bibitem [{\citenamefont {Seidl}\ \emph {et~al.}(1999)\citenamefont {Seidl},
  \citenamefont {Perdew},\ and\ \citenamefont {Levy}}]{seidl1999strictly}%
  \BibitemOpen
  \bibfield  {author} {\bibinfo {author} {\bibfnamefont {Michael}\ \bibnamefont
  {Seidl}}, \bibinfo {author} {\bibfnamefont {John~P}\ \bibnamefont {Perdew}},
  \ and\ \bibinfo {author} {\bibfnamefont {Mel}\ \bibnamefont {Levy}},\
  }\bibfield  {title} {\enquote {\bibinfo {title} {Strictly correlated
  electrons in density-functional theory},}\ }\href@noop {} {\bibfield
  {journal} {\bibinfo  {journal} {Physical Review A}\ }\textbf {\bibinfo
  {volume} {59}},\ \bibinfo {pages} {51} (\bibinfo {year} {1999})}\BibitemShut
  {NoStop}%
\bibitem [{\citenamefont {Townsend}(2000)}]{townsend2000modern}%
  \BibitemOpen
  \bibfield  {author} {\bibinfo {author} {\bibfnamefont {John~S}\ \bibnamefont
  {Townsend}},\ }\href@noop {} {\emph {\bibinfo {title} {A modern approach to
  quantum mechanics}}}\ (\bibinfo  {publisher} {University Science Books},\
  \bibinfo {year} {2000})\BibitemShut {NoStop}%
\bibitem [{\citenamefont {G{\"o}rling}\ and\ \citenamefont
  {Levy}(1994)}]{gorling1994exact}%
  \BibitemOpen
  \bibfield  {author} {\bibinfo {author} {\bibfnamefont {Andreas}\ \bibnamefont
  {G{\"o}rling}}\ and\ \bibinfo {author} {\bibfnamefont {Mel}\ \bibnamefont
  {Levy}},\ }\bibfield  {title} {\enquote {\bibinfo {title} {{Exact Kohn-Sham
  scheme based on perturbation theory}},}\ }\href@noop {} {\bibfield  {journal}
  {\bibinfo  {journal} {Phys.~Rev.~A}\ }\textbf {\bibinfo {volume} {50}},\
  \bibinfo {pages} {196} (\bibinfo {year} {1994})}\BibitemShut {NoStop}%
\bibitem [{\citenamefont {Pribram-Jones}\ \emph {et~al.}(2014)\citenamefont
  {Pribram-Jones}, \citenamefont {Pittalis}, \citenamefont {Gross},\ and\
  \citenamefont {Burke}}]{pribram2014thermal}%
  \BibitemOpen
  \bibfield  {author} {\bibinfo {author} {\bibfnamefont {Aurora}\ \bibnamefont
  {Pribram-Jones}}, \bibinfo {author} {\bibfnamefont {Stefano}\ \bibnamefont
  {Pittalis}}, \bibinfo {author} {\bibfnamefont {EKU}\ \bibnamefont {Gross}}, \
  and\ \bibinfo {author} {\bibfnamefont {Kieron}\ \bibnamefont {Burke}},\
  }\bibfield  {title} {\enquote {\bibinfo {title} {Thermal density functional
  theory in context},}\ }in\ \href@noop {} {\emph {\bibinfo {booktitle}
  {Frontiers and Challenges in Warm Dense Matter}}}\ (\bibinfo  {publisher}
  {Springer},\ \bibinfo {year} {2014})\ pp.\ \bibinfo {pages}
  {25--60}\BibitemShut {NoStop}%
\bibitem [{\citenamefont {Ceperley}\ and\ \citenamefont
  {Alder}(1980)}]{ceperley1980ground}%
  \BibitemOpen
  \bibfield  {author} {\bibinfo {author} {\bibfnamefont {David~M}\ \bibnamefont
  {Ceperley}}\ and\ \bibinfo {author} {\bibfnamefont {Berni~J}\ \bibnamefont
  {Alder}},\ }\bibfield  {title} {\enquote {\bibinfo {title} {Ground state of
  the electron gas by a stochastic method},}\ }\href@noop {} {\bibfield
  {journal} {\bibinfo  {journal} {Physical Review Letters}\ }\textbf {\bibinfo
  {volume} {45}},\ \bibinfo {pages} {566} (\bibinfo {year} {1980})}\BibitemShut
  {NoStop}%
\bibitem [{\citenamefont {Vosko}\ \emph {et~al.}(1980)\citenamefont {Vosko},
  \citenamefont {Wilk},\ and\ \citenamefont {Nusair}}]{vosko1980accurate}%
  \BibitemOpen
  \bibfield  {author} {\bibinfo {author} {\bibfnamefont {Seymour~H}\
  \bibnamefont {Vosko}}, \bibinfo {author} {\bibfnamefont {Leslie}\
  \bibnamefont {Wilk}}, \ and\ \bibinfo {author} {\bibfnamefont {Marwan}\
  \bibnamefont {Nusair}},\ }\bibfield  {title} {\enquote {\bibinfo {title}
  {Accurate spin-dependent electron liquid correlation energies for local spin
  density calculations: a critical analysis},}\ }\href@noop {} {\bibfield
  {journal} {\bibinfo  {journal} {Canadian Journal of physics}\ }\textbf
  {\bibinfo {volume} {58}},\ \bibinfo {pages} {1200--1211} (\bibinfo {year}
  {1980})}\BibitemShut {NoStop}%
\bibitem [{\citenamefont {Becke}(1993)}]{becke1993new}%
  \BibitemOpen
  \bibfield  {author} {\bibinfo {author} {\bibfnamefont {Axel~D}\ \bibnamefont
  {Becke}},\ }\bibfield  {title} {\enquote {\bibinfo {title} {A new mixing of
  hartree--fock and local density-functional theories},}\ }\href@noop {}
  {\bibfield  {journal} {\bibinfo  {journal} {The Journal of chemical physics}\
  }\textbf {\bibinfo {volume} {98}},\ \bibinfo {pages} {1372--1377} (\bibinfo
  {year} {1993})}\BibitemShut {NoStop}%
\bibitem [{\citenamefont {Perdew}\ \emph {et~al.}(1982)\citenamefont {Perdew},
  \citenamefont {Parr}, \citenamefont {Levy},\ and\ \citenamefont
  {Balduz~Jr}}]{perdew1982density}%
  \BibitemOpen
  \bibfield  {author} {\bibinfo {author} {\bibfnamefont {John~P}\ \bibnamefont
  {Perdew}}, \bibinfo {author} {\bibfnamefont {Robert~G}\ \bibnamefont {Parr}},
  \bibinfo {author} {\bibfnamefont {Mel}\ \bibnamefont {Levy}}, \ and\ \bibinfo
  {author} {\bibfnamefont {Jose~L}\ \bibnamefont {Balduz~Jr}},\ }\bibfield
  {title} {\enquote {\bibinfo {title} {Density-functional theory for fractional
  particle number: derivative discontinuities of the energy},}\ }\href@noop {}
  {\bibfield  {journal} {\bibinfo  {journal} {Physical Review Letters}\
  }\textbf {\bibinfo {volume} {49}},\ \bibinfo {pages} {1691} (\bibinfo {year}
  {1982})}\BibitemShut {NoStop}%
\bibitem [{\citenamefont {Ruder}(2016)}]{ruder2016overview}%
  \BibitemOpen
  \bibfield  {author} {\bibinfo {author} {\bibfnamefont {Sebastian}\
  \bibnamefont {Ruder}},\ }\bibfield  {title} {\enquote {\bibinfo {title} {An
  overview of gradient descent optimization algorithms},}\ }\href@noop {}
  {\bibfield  {journal} {\bibinfo  {journal} {arXiv preprint arXiv:1609.04747}\
  } (\bibinfo {year} {2016})}\BibitemShut {NoStop}%
\bibitem [{\citenamefont {{\v{C}}ern{\`y}}(1985)}]{vcerny1985thermodynamical}%
  \BibitemOpen
  \bibfield  {author} {\bibinfo {author} {\bibfnamefont {Vladim{\'\i}r}\
  \bibnamefont {{\v{C}}ern{\`y}}},\ }\bibfield  {title} {\enquote {\bibinfo
  {title} {Thermodynamical approach to the traveling salesman problem: An
  efficient simulation algorithm},}\ }\href@noop {} {\bibfield  {journal}
  {\bibinfo  {journal} {Journal of optimization theory and applications}\
  }\textbf {\bibinfo {volume} {45}},\ \bibinfo {pages} {41--51} (\bibinfo
  {year} {1985})}\BibitemShut {NoStop}%
\bibitem [{\citenamefont {Bertsimas}\ and\ \citenamefont
  {Tsitsiklis}(1993)}]{bertsimas1993simulated}%
  \BibitemOpen
  \bibfield  {author} {\bibinfo {author} {\bibfnamefont {Dimitris}\
  \bibnamefont {Bertsimas}}\ and\ \bibinfo {author} {\bibfnamefont {John}\
  \bibnamefont {Tsitsiklis}},\ }\bibfield  {title} {\enquote {\bibinfo {title}
  {Simulated annealing},}\ }\href@noop {} {\bibfield  {journal} {\bibinfo
  {journal} {Statistical science}\ }\textbf {\bibinfo {volume} {8}},\ \bibinfo
  {pages} {10--15} (\bibinfo {year} {1993})}\BibitemShut {NoStop}%
\bibitem [{\citenamefont {Tierney}(1994)}]{tierney1994markov}%
  \BibitemOpen
  \bibfield  {author} {\bibinfo {author} {\bibfnamefont {Luke}\ \bibnamefont
  {Tierney}},\ }\bibfield  {title} {\enquote {\bibinfo {title} {Markov chains
  for exploring posterior distributions},}\ }\href@noop {} {\bibfield
  {journal} {\bibinfo  {journal} {the Annals of Statistics}\ }\textbf {\bibinfo
  {volume} {22}},\ \bibinfo {pages} {1701--1728} (\bibinfo {year}
  {1994})}\BibitemShut {NoStop}%
\bibitem [{\citenamefont {Vishwanathan}\ \emph {et~al.}(2006)\citenamefont
  {Vishwanathan}, \citenamefont {Borgwardt},\ and\ \citenamefont
  {Schraudolph}}]{vishwanathan2006fast}%
  \BibitemOpen
  \bibfield  {author} {\bibinfo {author} {\bibfnamefont {SVN}\ \bibnamefont
  {Vishwanathan}}, \bibinfo {author} {\bibfnamefont {Karsten~M}\ \bibnamefont
  {Borgwardt}}, \ and\ \bibinfo {author} {\bibfnamefont {Nicol~N}\ \bibnamefont
  {Schraudolph}},\ }\bibfield  {title} {\enquote {\bibinfo {title} {Fast
  computation of graph kernels},}\ }in\ \href@noop {} {\emph {\bibinfo
  {booktitle} {NIPS}}},\ Vol.~\bibinfo {volume} {19}\ (\bibinfo {year} {2006})\
  pp.\ \bibinfo {pages} {131--138}\BibitemShut {NoStop}%
\bibitem [{\citenamefont {Neal}(2012)}]{neal2012bayesian}%
  \BibitemOpen
  \bibfield  {author} {\bibinfo {author} {\bibfnamefont {Radford~M}\
  \bibnamefont {Neal}},\ }\href@noop {} {\emph {\bibinfo {title} {Bayesian
  learning for neural networks}}},\ \bibinfo {series} {Lecture Notes in
  Statistics}, Vol.\ \bibinfo {volume} {118}\ (\bibinfo  {publisher} {Springer
  Science \& Business Media},\ \bibinfo {year} {2012})\BibitemShut {NoStop}%
\bibitem [{\citenamefont {Rios}\ and\ \citenamefont
  {Sahinidis}(2013)}]{rios2013derivative}%
  \BibitemOpen
  \bibfield  {author} {\bibinfo {author} {\bibfnamefont {Luis~Miguel}\
  \bibnamefont {Rios}}\ and\ \bibinfo {author} {\bibfnamefont {Nikolaos~V}\
  \bibnamefont {Sahinidis}},\ }\bibfield  {title} {\enquote {\bibinfo {title}
  {Derivative-free optimization: a review of algorithms and comparison of
  software implementations},}\ }\href@noop {} {\bibfield  {journal} {\bibinfo
  {journal} {Journal of Global Optimization}\ }\textbf {\bibinfo {volume}
  {56}},\ \bibinfo {pages} {1247--1293} (\bibinfo {year} {2013})}\BibitemShut
  {NoStop}%
\bibitem [{\citenamefont {Pillai}\ \emph {et~al.}(2014)\citenamefont {Pillai},
  \citenamefont {Stuart},\ and\ \citenamefont {Thi{\'e}ry}}]{pillai2014noisy}%
  \BibitemOpen
  \bibfield  {author} {\bibinfo {author} {\bibfnamefont {Natesh~S}\
  \bibnamefont {Pillai}}, \bibinfo {author} {\bibfnamefont {Andrew~M}\
  \bibnamefont {Stuart}}, \ and\ \bibinfo {author} {\bibfnamefont
  {Alexandre~H}\ \bibnamefont {Thi{\'e}ry}},\ }\bibfield  {title} {\enquote
  {\bibinfo {title} {Noisy gradient flow from a random walk in hilbert
  space},}\ }\href@noop {} {\bibfield  {journal} {\bibinfo  {journal}
  {Stochastic Partial Differential Equations: Analysis and Computations}\
  }\textbf {\bibinfo {volume} {2}},\ \bibinfo {pages} {196--232} (\bibinfo
  {year} {2014})}\BibitemShut {NoStop}%
\bibitem [{\citenamefont {Perozzi}\ \emph {et~al.}(2016)\citenamefont
  {Perozzi}, \citenamefont {Kulkarni},\ and\ \citenamefont
  {Skiena}}]{perozzi2016walklets}%
  \BibitemOpen
  \bibfield  {author} {\bibinfo {author} {\bibfnamefont {Bryan}\ \bibnamefont
  {Perozzi}}, \bibinfo {author} {\bibfnamefont {Vivek}\ \bibnamefont
  {Kulkarni}}, \ and\ \bibinfo {author} {\bibfnamefont {Steven}\ \bibnamefont
  {Skiena}},\ }\bibfield  {title} {\enquote {\bibinfo {title} {Walklets:
  Multiscale graph embeddings for interpretable network classification},}\
  }\href@noop {} {\bibfield  {journal} {\bibinfo  {journal} {arXiv preprint
  arXiv:1605.02115}\ } (\bibinfo {year} {2016})}\BibitemShut {NoStop}%
\bibitem [{\citenamefont {Tejedor}(2018)}]{tejedor2018dimensional}%
  \BibitemOpen
  \bibfield  {author} {\bibinfo {author} {\bibfnamefont {Vincent}\ \bibnamefont
  {Tejedor}},\ }\bibfield  {title} {\enquote {\bibinfo {title} {A dimensional
  acceleration of gradient descent-like methods, using persistent random
  walkers},}\ }\href@noop {} {\bibfield  {journal} {\bibinfo  {journal} {arXiv
  preprint arXiv:1801.04532}\ } (\bibinfo {year} {2018})}\BibitemShut {NoStop}%
\bibitem [{\citenamefont {Ghosh}\ \emph {et~al.}(2018)\citenamefont {Ghosh},
  \citenamefont {Das}, \citenamefont {Gon{\c{c}}alves}, \citenamefont
  {Quaresma},\ and\ \citenamefont {Kundu}}]{ghosh2018journey}%
  \BibitemOpen
  \bibfield  {author} {\bibinfo {author} {\bibfnamefont {Swarnendu}\
  \bibnamefont {Ghosh}}, \bibinfo {author} {\bibfnamefont {Nibaran}\
  \bibnamefont {Das}}, \bibinfo {author} {\bibfnamefont {Teresa}\ \bibnamefont
  {Gon{\c{c}}alves}}, \bibinfo {author} {\bibfnamefont {Paulo}\ \bibnamefont
  {Quaresma}}, \ and\ \bibinfo {author} {\bibfnamefont {Mahantapas}\
  \bibnamefont {Kundu}},\ }\bibfield  {title} {\enquote {\bibinfo {title} {The
  journey of graph kernels through two decades},}\ }\href@noop {} {\bibfield
  {journal} {\bibinfo  {journal} {Computer Science Review}\ }\textbf {\bibinfo
  {volume} {27}},\ \bibinfo {pages} {88--111} (\bibinfo {year}
  {2018})}\BibitemShut {NoStop}%
\bibitem [{\citenamefont {Tejedor}\ \emph {et~al.}(2012)\citenamefont
  {Tejedor}, \citenamefont {Voituriez},\ and\ \citenamefont
  {B{\'e}nichou}}]{tejedor2012optimizing}%
  \BibitemOpen
  \bibfield  {author} {\bibinfo {author} {\bibfnamefont {Vincent}\ \bibnamefont
  {Tejedor}}, \bibinfo {author} {\bibfnamefont {Raphael}\ \bibnamefont
  {Voituriez}}, \ and\ \bibinfo {author} {\bibfnamefont {Olivier}\ \bibnamefont
  {B{\'e}nichou}},\ }\bibfield  {title} {\enquote {\bibinfo {title} {Optimizing
  persistent random searches},}\ }\href@noop {} {\bibfield  {journal} {\bibinfo
   {journal} {Phys.~Rev.~Lett.}\ }\textbf {\bibinfo {volume} {108}},\ \bibinfo
  {pages} {088103} (\bibinfo {year} {2012})}\BibitemShut {NoStop}%
\bibitem [{\citenamefont {Maclaurin}\ \emph {et~al.}(2015)\citenamefont
  {Maclaurin}, \citenamefont {Duvenaud},\ and\ \citenamefont
  {Adams}}]{maclaurin2015gradient}%
  \BibitemOpen
  \bibfield  {author} {\bibinfo {author} {\bibfnamefont {Dougal}\ \bibnamefont
  {Maclaurin}}, \bibinfo {author} {\bibfnamefont {David}\ \bibnamefont
  {Duvenaud}}, \ and\ \bibinfo {author} {\bibfnamefont {Ryan}\ \bibnamefont
  {Adams}},\ }\bibfield  {title} {\enquote {\bibinfo {title} {Gradient-based
  hyperparameter optimization through reversible learning},}\ }in\ \href@noop
  {} {\emph {\bibinfo {booktitle} {International Conference on Machine
  Learning}}}\ (\bibinfo {year} {2015})\ pp.\ \bibinfo {pages}
  {2113--2122}\BibitemShut {NoStop}%
\bibitem [{\citenamefont {Blais}(2003)}]{blais2003algorithmes}%
  \BibitemOpen
  \bibfield  {author} {\bibinfo {author} {\bibfnamefont {Alexandre}\
  \bibnamefont {Blais}},\ }\bibfield  {title} {\enquote {\bibinfo {title}
  {Algorithmes et architectures pour ordinateurs quantiques
  supraconducteurs},}\ }in\ \href@noop {} {\emph {\bibinfo {booktitle} {Annales
  de Physique}}},\ Vol.~\bibinfo {volume} {28}\ (\bibinfo {organization} {EDP
  Sciences},\ \bibinfo {year} {2003})\ pp.\ \bibinfo {pages}
  {1--147}\BibitemShut {NoStop}%
\bibitem [{\citenamefont {Kittel}(1987)}]{kittel1987quantum}%
  \BibitemOpen
  \bibfield  {author} {\bibinfo {author} {\bibfnamefont {Charles}\ \bibnamefont
  {Kittel}},\ }\href@noop {} {\emph {\bibinfo {title} {Quantum theory of
  solids}}}\ (\bibinfo  {publisher} {Wiley},\ \bibinfo {year}
  {1987})\BibitemShut {NoStop}%
\bibitem [{\citenamefont {Hales}\ and\ \citenamefont
  {Hallgren}(2000)}]{hales2000improved}%
  \BibitemOpen
  \bibfield  {author} {\bibinfo {author} {\bibfnamefont {Lisa}\ \bibnamefont
  {Hales}}\ and\ \bibinfo {author} {\bibfnamefont {Sean}\ \bibnamefont
  {Hallgren}},\ }\bibfield  {title} {\enquote {\bibinfo {title} {An improved
  quantum fourier transform algorithm and applications},}\ }in\ \href@noop {}
  {\emph {\bibinfo {booktitle} {Foundations of Computer Science, 2000.
  Proceedings. 41st Annual Symposium on}}}\ (\bibinfo {organization} {IEEE},\
  \bibinfo {year} {2000})\ pp.\ \bibinfo {pages} {515--525}\BibitemShut
  {NoStop}%
\bibitem [{\citenamefont {Suzuki}(1993)}]{suzuki1993improved}%
  \BibitemOpen
  \bibfield  {author} {\bibinfo {author} {\bibfnamefont {Masuo}\ \bibnamefont
  {Suzuki}},\ }\bibfield  {title} {\enquote {\bibinfo {title} {Improved
  trotter-like formula},}\ }\href@noop {} {\bibfield  {journal} {\bibinfo
  {journal} {Phys.~Lett.~A}\ }\textbf {\bibinfo {volume} {180}},\ \bibinfo
  {pages} {232--234} (\bibinfo {year} {1993})}\BibitemShut {NoStop}%
\bibitem [{\citenamefont {Poulin}\ \emph {et~al.}(2018)\citenamefont {Poulin},
  \citenamefont {Kitaev}, \citenamefont {Steiger}, \citenamefont {Hastings},\
  and\ \citenamefont {Troyer}}]{poulin2018quantum}%
  \BibitemOpen
  \bibfield  {author} {\bibinfo {author} {\bibfnamefont {David}\ \bibnamefont
  {Poulin}}, \bibinfo {author} {\bibfnamefont {Alexei}\ \bibnamefont {Kitaev}},
  \bibinfo {author} {\bibfnamefont {Damian~S}\ \bibnamefont {Steiger}},
  \bibinfo {author} {\bibfnamefont {Matthew~B}\ \bibnamefont {Hastings}}, \
  and\ \bibinfo {author} {\bibfnamefont {Matthias}\ \bibnamefont {Troyer}},\
  }\bibfield  {title} {\enquote {\bibinfo {title} {Quantum algorithm for
  spectral measurement with a lower gate count},}\ }\href@noop {} {\bibfield
  {journal} {\bibinfo  {journal} {Phys.~Rev.~Lett.}\ }\textbf {\bibinfo
  {volume} {121}},\ \bibinfo {pages} {010501} (\bibinfo {year}
  {2018})}\BibitemShut {NoStop}%
\bibitem [{\citenamefont {Novak}\ and\ \citenamefont
  {Wo{\'z}niakowski}(2008)}]{novak2008tractability}%
  \BibitemOpen
  \bibfield  {author} {\bibinfo {author} {\bibfnamefont {Erich}\ \bibnamefont
  {Novak}}\ and\ \bibinfo {author} {\bibfnamefont {Henryk}\ \bibnamefont
  {Wo{\'z}niakowski}},\ }\href@noop {} {\emph {\bibinfo {title} {Tractability
  of Multivariate Problems: Standard information for functionals}}},\
  Vol.~\bibinfo {volume} {12}\ (\bibinfo  {publisher} {European Mathematical
  Society},\ \bibinfo {year} {2008})\BibitemShut {NoStop}%
\bibitem [{\citenamefont {Oh}(2008)}]{oh2008quantum}%
  \BibitemOpen
  \bibfield  {author} {\bibinfo {author} {\bibfnamefont {Sangchul}\
  \bibnamefont {Oh}},\ }\bibfield  {title} {\enquote {\bibinfo {title} {Quantum
  computational method of finding the ground-state energy and expectation
  values},}\ }\href@noop {} {\bibfield  {journal} {\bibinfo  {journal}
  {Phys.~Rev.~A}\ }\textbf {\bibinfo {volume} {77}},\ \bibinfo {pages} {012326}
  (\bibinfo {year} {2008})}\BibitemShut {NoStop}%
\bibitem [{\citenamefont {Kassal}\ \emph {et~al.}(2011)\citenamefont {Kassal},
  \citenamefont {Whitfield}, \citenamefont {Perdomo-Ortiz}, \citenamefont
  {Yung},\ and\ \citenamefont {Aspuru-Guzik}}]{kassal2011simulating}%
  \BibitemOpen
  \bibfield  {author} {\bibinfo {author} {\bibfnamefont {Ivan}\ \bibnamefont
  {Kassal}}, \bibinfo {author} {\bibfnamefont {James~D}\ \bibnamefont
  {Whitfield}}, \bibinfo {author} {\bibfnamefont {Alejandro}\ \bibnamefont
  {Perdomo-Ortiz}}, \bibinfo {author} {\bibfnamefont {Man-Hong}\ \bibnamefont
  {Yung}}, \ and\ \bibinfo {author} {\bibfnamefont {Al{\'a}n}\ \bibnamefont
  {Aspuru-Guzik}},\ }\bibfield  {title} {\enquote {\bibinfo {title} {Simulating
  chemistry using quantum computers},}\ }\href@noop {} {\bibfield  {journal}
  {\bibinfo  {journal} {Annual review of physical chemistry}\ }\textbf
  {\bibinfo {volume} {62}},\ \bibinfo {pages} {185--207} (\bibinfo {year}
  {2011})}\BibitemShut {NoStop}%
\bibitem [{Note2()}]{Note2}%
  \BibitemOpen
  \bibinfo {note} {This algorithm can go by other names, notably quantum
  counting, but we use `QAE' instead of the abbreviated `QC' for quantum
  counting since this would overlap with `quantum chemistry' and `quantum
  computing' if we chose to also use abbreviations for those.}\BibitemShut
  {Stop}%
\bibitem [{\citenamefont {Knill}\ \emph {et~al.}(2007)\citenamefont {Knill},
  \citenamefont {Ortiz},\ and\ \citenamefont {Somma}}]{knill2007optimal}%
  \BibitemOpen
  \bibfield  {author} {\bibinfo {author} {\bibfnamefont {Emanuel}\ \bibnamefont
  {Knill}}, \bibinfo {author} {\bibfnamefont {Gerardo}\ \bibnamefont {Ortiz}},
  \ and\ \bibinfo {author} {\bibfnamefont {Rolando~D}\ \bibnamefont {Somma}},\
  }\bibfield  {title} {\enquote {\bibinfo {title} {Optimal quantum measurements
  of expectation values of observables},}\ }\href@noop {} {\bibfield  {journal}
  {\bibinfo  {journal} {Phys.~Rev.~A}\ }\textbf {\bibinfo {volume} {75}},\
  \bibinfo {pages} {012328} (\bibinfo {year} {2007})}\BibitemShut {NoStop}%
\bibitem [{\citenamefont {Marriott}\ and\ \citenamefont
  {Watrous}(2005)}]{marriott2005quantum}%
  \BibitemOpen
  \bibfield  {author} {\bibinfo {author} {\bibfnamefont {Chris}\ \bibnamefont
  {Marriott}}\ and\ \bibinfo {author} {\bibfnamefont {John}\ \bibnamefont
  {Watrous}},\ }\bibfield  {title} {\enquote {\bibinfo {title} {{Quantum
  Arthur--Merlin games}},}\ }\href@noop {} {\bibfield  {journal} {\bibinfo
  {journal} {Computational Complexity}\ }\textbf {\bibinfo {volume} {14}},\
  \bibinfo {pages} {122--152} (\bibinfo {year} {2005})}\BibitemShut {NoStop}%
\end{thebibliography}%

\end{document}